\documentclass[12pt]{article}
\pdfoutput=1

\usepackage[utf8]{inputenc}
\usepackage{microtype}
\usepackage{graphicx}
\usepackage{mathtools}
\usepackage[usenames,dvipsnames]{xcolor}
\usepackage{amssymb,amsmath,amsthm}
\usepackage{amsfonts}
\usepackage{shellesc}
\usepackage{stmaryrd}
\usepackage{xspace}
\usepackage{blkarray}
\usepackage{multicol}
\usepackage{bm}
\usepackage{url}
\usepackage{tikz}
\usetikzlibrary{snakes}
\usepackage{feynmf}
\usepackage{comment}
\usepackage{cite}
\usepackage{breqn}
\usepackage{textgreek}
\usepackage{xfrac}
\usepackage{enumitem}
\usepackage{bm}
\usepackage[cachedir={./},finalizecache=false,frozencache=true]{minted}
\setminted[julia]{frame=lines,
rulecolor=\color{white!80!black},
fontsize=\small,
numbers=right,
numbersep=-5pt,
obeytabs=true,
tabsize=4,
escapeinside=??}

\mathtoolsset{showonlyrefs}
\usepackage{algorithm}
\usepackage[noend]{algpseudocode}
\usepackage{algxpar}
\usepackage{lmodern}

\tikzset{
    massive/.style={
        decoration={},
        decorate
    },
    massless/.style={
        color=white!20!black,
        decoration={coil, aspect=1, mirror, segment length=1.2mm, amplitude=0.6mm},
        decorate
    },
    photon/.style={
        decoration={snake, aspect=0.75, mirror, segment length=2mm, amplitude=0.5mm},
        decorate
    }
}
\usetikzlibrary{arrows.meta}

\usepackage{caption}
\usepackage{subcaption}
\usepackage[export]{adjustbox}
\usepackage{jheppub}

\usepackage{hyperref}
\hypersetup{
	colorlinks=true,
	urlcolor=RoyalBlue,
	linkcolor=orange!80!black,
	citecolor=RoyalBlue,
	bookmarks=true,
	pdftitle={Principal Landau Determinants},
	pdfauthor={Claudia Fevola, Sebastian Mizera, Simon Telen},
	pdfdisplaydoctitle=true,
	pdfstartview=FitH
}

\theoremstyle{plain}
\newtheorem{theorem}{Theorem}[section]
\newtheorem{conjecture}[theorem]{Conjecture}
\newtheorem{lemma}[theorem]{Lemma}

\newtheorem{remark}[theorem]{Remark}
\newtheorem{proposition}[theorem]{Proposition}

\theoremstyle{definition}
\numberwithin{equation}{section}
\newtheorem{definition}[theorem]{Definition}

\newtheorem{examplex}[theorem]{Example}

\newenvironment{example}
  {\pushQED{\qed}\examplex}
  {\popQED\endexamplex}

\newcommand{\m}{\mathfrak{m}}

\newcommand{\F}{\mathscr{F}}

\newcommand{\G}{\mathscr{G}}
\newcommand{\I}{\mathscr{I}}

\newcommand{\M}{\mathcal{M}}

\newcommand{\Z}{\mathbb{Z}}
\newcommand{\E}{\mathcal{E}}

\newcommand{\C}{\mathbb{C}}

\newcommand{\V}{\mathcal{V}}

\newcommand{\D}{\delta}
\renewcommand{\d}{\partial}

\newcommand{\R}{\mathbb{R}}

\newcommand{\Newt}{\textup{Newt}}

\definecolor{mycolor1}{rgb}{0.00000,0.44700,0.74100}
\definecolor{mycolor2}{rgb}{0.8500, 0.3250, 0.0980}
\definecolor{mycolor3}{rgb}{0.9290, 0.6940, 0.1250}
\definecolor{mycolor4}{rgb}{0.4940, 0.1840, 0.5560}
\definecolor{mycolor5}{rgb}{0.4660, 0.6740, 0.1880}
\colorlet{blue}{RoyalBlue}

\newcommand{\be}{\begin{equation}}
\newcommand{\ee}{\end{equation}}
\newcommand{\nn}{\nonumber}
\renewcommand{\I}{\mathcal{I}}
\renewcommand{\R}{\mathbb{R}}

\renewcommand{\E}{\mathrm{E}}
\renewcommand{\d}{\mathrm{d}}
\newcommand{\U}{\mathcal{U}}
\renewcommand{\F}{\mathcal{F}}
\renewcommand{\V}{\mathcal{V}}
\renewcommand{\D}{\mathrm{D}}

\renewcommand{\L}{\mathrm{L}}
\newcommand{\eps}{\varepsilon}
\renewcommand{\Z}{\mathbb{Z}}
\renewcommand{\m}{\mathsf{m}}
\newcommand{\n}{\mathsf{n}}
\renewcommand{\M}{\mathsf{M}}
\renewcommand{\C}{\mathbb{C}}

\newcommand{\Ue}{\mathfrak{U}}
\newcommand{\Fe}{\mathfrak{F}}
\renewcommand{\Newt}{\mathrm{Newt}}
\newcommand{\conv}{\textup{Conv}}

\newcommand{\e}{\mathrm{e}}

\newcommand{\An}{{\tt A}_n}
\renewcommand{\G}{\mathcal{G}}
\newcommand{\PLDwebsite}{\url{https://mathrepo.mis.mpg.de/PLD/}}

\allowdisplaybreaks
\title{\huge Principal Landau Determinants}

\author[1]{Claudia Fevola,}\emailAdd{claudia.fevola@inria.fr}
\author[2]{Sebastian Mizera,}\emailAdd{smizera@ias.edu}
\author[3]{Simon Telen}\emailAdd{simon.telen@mis.mpg.de}

\affiliation[1]{Université Paris-Saclay, Inria, 91120 Palaiseau, France}
\affiliation[2]{Institute for Advanced Study, Einstein Drive, Princeton, NJ 08540, USA}
\affiliation[3]{Max Planck Institute for Mathematics in the Sciences, Inselstra\ss e 22, 04103 Leipzig, Germany}

\abstract{%
We reformulate the Landau analysis of Feynman integrals with the aim of advancing the state of the art in modern particle-physics computations. We contribute new algorithms for computing Landau singularities, using tools from polyhedral geometry and symbolic/numerical elimination. Inspired by the work of Gelfand, Kapranov, and Zelevinsky (GKZ) on generalized Euler integrals, we define the principal Landau determinant of a Feynman diagram. We illustrate with a number of examples that this algebraic formalism allows to compute many components of the Landau singular locus. We adapt the GKZ framework by carefully specializing Euler integrals to Feynman integrals. For instance, ultraviolet and infrared singularities are detected as irreducible components of an incidence variety, which project dominantly to the kinematic space. We compute principal Landau determinants for the infinite families of one-loop and banana diagrams with different mass configurations, and for a range of cutting-edge Standard Model processes. Our algorithms build on the \texttt{Julia} package \texttt{Landau.jl} and are implemented in the new open-source package \texttt{PLD.jl} available at \PLDwebsite.\\
\newpage
\noindent
PROGRAM SUMMARY\\
Program title: \texttt{PLD.jl}\\
Developer's respository link: \PLDwebsite \\
Licensing provisions: Creative Commons by 4.0 (CC by 4.0)\\
Programming language: \texttt{Julia}\\
Supplementary material: The repository includes the source code with documentation (PLD\_code.zip), a jupyter notebook tutorial providing installation and usage instructions (PLD\_notebook.zip), a database containing the output of our algorithm on 114 examples of Feynman integrals (PLD\_database.zip).\\ 
Nature of problem: A fundamental challenge in scattering amplitude is to determine the values of complexified kinematic invariants for which an
amplitude can develop singularities. Bjorken, Landau, and Nakanishi wrote a system of polynomial constraints, nowadays known as the Landau equations. This project aims to rigorously revisit the Landau analysis of the singularity locus of Feynman integrals with a practical view towards explicit computations.\\
Solution method: We define the principal Landau determinant (PLD), which is a variety inspired by the work of Gelfand, Kapranov, and Zelevinsky (GKZ). We conjecture that it provides a subset of the singularity locus, and we implement effective algorithms to compute its defining equation explicitly.\\
References: OSCAR \cite{OSCAR}, HomotopyContinuation.jl \cite{10.1007/978-3-319-96418-8_54}, Landau.jl \cite{Mizera:2021icv}
}
\begin{document} 

\setcounter{tocdepth}{2}
\maketitle
\setcounter{page}{2}

\section{Introduction}

Our ability to perform high-precision computations of scattering amplitudes in quantum field theory relies on new insights into their analytic structure. A fundamental challenge in this field is to determine the values of complexified kinematic invariants for which a given amplitude can develop singularities. These are poles or branch points, interchangeably called \emph{anomalous thresholds} or \emph{Landau singularities}. A deeper understanding of this problem would have an immediate impact on the cutting-edge computations in the method of differential equations \cite{Badger:2023eqz}, symbol-level constraints on polylogarithmic Feynman integrals and beyond \cite{Arkani-Hamed:2022rwr,Bourjaily:2022bwx}, and the non-perturbative bootstrap~\cite{Kruczenski:2022lot}.

The question itself has a long history and dates back to the work of Bjorken, Landau, and Nakanishi \cite{Bjorken:1959fd,Landau:1959fi,10.1143/PTP.22.128}. These authors wrote a system of polynomial constraints, nowadays known as the \emph{Landau equations}, for determining the singularities. See \cite{Eden:1966dnq,bjorken1965relativistic,Itzykson:1980rh} for textbook expositions. As is well-known, Landau analysis was never formulated precisely enough to be applicable to the Standard Model computations of current importance to collider phenomenology, especially when massless particles are involved, see, e.g., \cite[Ex.~5.3]{Collins:2011zzd}. There are also known examples where a naive application of Landau equations does not detect all singularities of Feynman integrals; instead, a more careful blow-up analysis is needed \cite{Landshoff1966,doi:10.1063/1.1724262,Berghoff:2022mqu}.
The outstanding problems that need to be addressed before a large-scale application of Landau analysis are (i) practical formulation of Landau conditions in the presence of massless particles and ultraviolet/infrared (UV/IR) divergences which make Feynman integrals singular everywhere in the kinematic space; (ii) systematic classification of the systems of equations that need to be solved to actually account for all singularities; and (iii) providing practical tools for solving such systems.

In this work, we formalize Landau singularities as the subspace of kinematics on which the Feynman integrand is ``more singular'' than generically, thus addressing point (i).
More concretely, we introduce the \emph{Euler discriminant variety}. Calling the integration space $X$, the Euler discriminant variety is the locus of kinematic invariants for which the signed Euler characteristic $|\chi(X)|$ drops compared to its generic value. To perform explicit computations, we define the \emph{principal Landau determinant} (PLD), which is an approach to (ii) that employs polyhedral geometry to scan over different ways Schwinger parameters can go to zero or infinity. Finally, to address point (iii), we introduce the package \texttt{PLD.jl} available open-source at
\begin{center}
\PLDwebsite.
\end{center}
It implements symbolic and numerical elimination algorithms introduced in this paper. As concrete examples, we will apply it to the Feynman diagrams shown in Fig.~\ref{fig:diagrams} and \cite[Fig.~1]{Mizera:2021icv}. They are summarized in a database of $114$ examples of different graph topologies and mass assignments accessible through the above link. 

Our results were announced in \cite{Fevola:2023kaw} with an emphasis on the physical aspects. The present paper motivates our definitions and fleshes out the algorithmic details.

\paragraph{\bf Summary of contributions.} The mathematical problem at hand is described as follows. We consider an $\E$-dimensional integral ${\cal I}(z)$ whose integrand depends on parameters $z$. Here $\E$ is the number of internal edges in a Feynman diagram, and $z = (z_1, \ldots, z_s)$ represents all kinematic invariants. This integral is a holomorphic function of $z$ on a neighborhood of generic complex parameters $z^* \in \mathbb{C}^s$. The Landau singular locus is an algebraic variety in $\mathbb{C}^s$, at which analytic continuation of ${\cal I}(z)$ may fail. This is formalized via differential equations satisfied by our integral, using the language of \emph{$D$-modules}. For a friendly introduction, see \cite{sattelberger2019d} and references therein. The steps are (a) to find a holonomic $D$-ideal annihilating ${\cal I}(z)$ and (b) to compute its singular locus \cite[Def.~1.12]{sattelberger2019d}. We conjecture that the result is the Euler discriminant. Unfortunately, while algorithms for step (b) exist, step (a) is usually problematic. 

Gelfand, Kapranov, and Zelevinsky (GKZ) consider particular integrals ${\cal I}(z)$, which they call \emph{generalized Euler integrals}, whose holonomic $D$-ideal can be constructed purely combinatorially \cite{gelfand1990generalized}. The result is nowadays referred to as a \emph{GKZ system}, or \emph{$A$-hypergeometric system}. The singular locus is defined by the \emph{principal $A$-determinant} $E_A$, which is a homogeneous polynomial in the parameters $z$ (Thm. \ref{thm:GKZsinglocus}). At the same time, the principal $A$-determinant characterizes when the topology of the integration space changes: it detects drops in the Euler characteristic (Thm. \ref{thm:eulercharvolume}). In other words, $E_A$ is a first example of an Euler discriminant. We recall the GKZ framework in Sec.~\ref{sec:2}.

Feynman integrals can be seen as specializations of GKZ integrals: $z$ is restricted to lie in the kinematic space, which can be viewed as a linear subspace of the GKZ parameter space. At the level of the singular locus, this specialization is quite tricky: 
\begin{equation} \label{eq:rule} \tag{*}
\text{One can \emph{not} just substitute kinematic variables in the principal $A$-determinant.}
\end{equation} 
We discuss this slogan at length in Sec.~\ref{sec:2}. Nonetheless, with the necessary care, the algebraic techniques for computing principal $A$-determinants can be adapted to the Feynman setting to compute components of the Landau singular locus. These components form the principal Landau determinant (PLD). The precise definition is given in Sec.~\ref{sec:3}, and we include a comparison with the Euler discriminant. We conjecture that the variety defined by the PLD is contained in the Euler discriminant variety, and verify this in all our examples.

Exceptions to the rule \eqref{eq:rule} are discussed in Sec.~\ref{sec:4}. We prove that for one-loop diagrams with several different mass configurations, the Euler discriminant equals the intersection of the principal $A$-determinant with kinematic space. This justifies the emphasis on one-loop examples in previous approaches \cite{Dlapa:2023cvx}. Our proofs use combinatorics and tools from~\cite{gelfand2008discriminants}.

Sec.~\ref{sec:algorithm} is on how to compute principal Landau determinants. Like in \cite{Mizera:2021icv}, we present symbolic and symbolic-numerical algorithms, relying on computer algebra and numerical nonlinear algebra. We explain how to use our open-source software \texttt{PLD.jl}, which finds components of Landau singular loci that had not been computed before. 

Sec.~\ref{sec:conclusion} provides an outlook and a list of open questions. This paper also comes with three appendices. In App.~\ref{sec:appendixHyperInt}, we explain how to use the compatibility graph algorithm implemented in \cite{Panzer:2014caa} for computing Landau singularities and contrast it with PLD. In App.~\ref{sec:appendix}, we review the derivation of the Schwinger parameter formula for Feynman integrals. Finally, App.~\ref{sec:appendix2} discusses PLD in the language of toric geometry.

\begin{figure}[t]
\vspace{-1em}
\centering
\captionsetup{justification=centering}

\begin{subfigure}[c]{0.3\textwidth}
\centering
\begin{tikzpicture}[scale = 0.9,line width=1,scale=1.3]
    \coordinate (v1) at (0,0);
    \coordinate (v2) at (0,1);
    \coordinate (v3) at (1,1);
    \coordinate (v4) at (2,1);
    \coordinate (v5) at (2,0);
    \coordinate (v6) at (1,0);
    \draw[massless] (v6) -- (v1) -- (v2) -- (v3);
    \draw[massive] (v6) -- (v3) -- (v4) -- (v5) -- (v6);
    \draw[massless] (v1) -- ++(-135:0.5);
    \draw[massless] (v2) -- ++(135:0.5);
    \draw[massless] (v4) -- ++(45:0.5);
    \draw[massless] (v5) -- ++(-45:0.5);
    \foreach \point in {v1, v2, v3, v4, v5, v6} {
    	\fill[white!20!black] (\point) circle (1.5pt);
    }
	\node at (0,0.5) [font=\footnotesize, anchor=east] {$\alpha_1$};
	\node at (0.5,-0.1) [font=\footnotesize, anchor=north] {$\alpha_2$};
	\node at (1.5,-0.1) [font=\footnotesize, anchor=north] {$\alpha_3$};
	\node at (2,0.5) [font=\footnotesize, anchor=west] {$\alpha_4$};
	\node at (1.5,1) [font=\footnotesize, anchor=south] {$\alpha_5$};
	\node at (0.5,1) [font=\footnotesize, anchor=south] {$\alpha_6$};
	\node at (1,0.5) [font=\footnotesize, anchor=west] {$\alpha_7$};
	\node at (-0.5,1.5) [font=\footnotesize] {$p_1$};
	\node at (-0.5,-0.5) [font=\footnotesize] {$p_2$};
	\node at (2.5,-0.5) [font=\footnotesize] {$p_3$};
	\node at (2.5,1.5) [font=\footnotesize] {$p_4$};
\end{tikzpicture}
\caption{Double-box with\\ an inner massive loop,\\ $G = \texttt{inner-dbox}$}
\end{subfigure}
\begin{subfigure}[c]{0.3\textwidth}
\centering
\begin{tikzpicture}[scale = 0.9,line width=1,scale=1.3]
    \coordinate (v1) at (0,0);
    \coordinate (v2) at (0,1);
    \coordinate (v3) at (1,1);
    \coordinate (v4) at (2,1);
    \coordinate (v5) at (2,0);
    \coordinate (v6) at (1,0);
    \draw[massive] (v1) -- (v2) -- (v3) -- (v4) -- (v5) -- (v6) -- (v1);
    \draw[massless] (v3) -- (v6);
    \draw[massless] (v1) -- ++(-135:0.5);
    \draw[massless] (v2) -- ++(135:0.5);
    \draw[massless] (v4) -- ++(45:0.5);
    \draw[massless] (v5) -- ++(-45:0.5);
    \foreach \point in {v1, v2, v3, v4, v5, v6} {
    	\fill[white!20!black] (\point) circle (1.5pt);
    }
	\node at (0,0.5) [font=\footnotesize, anchor=east] {$\alpha_1$};
	\node at (0.5,-0.1) [font=\footnotesize, anchor=north] {$\alpha_2$};
	\node at (1.5,-0.1) [font=\footnotesize, anchor=north] {$\alpha_3$};
	\node at (2,0.5) [font=\footnotesize, anchor=west] {$\alpha_4$};
	\node at (1.5,1) [font=\footnotesize, anchor=south] {$\alpha_5$};
	\node at (0.5,1) [font=\footnotesize, anchor=south] {$\alpha_6$};
	\node at (1,0.5) [font=\footnotesize, anchor=west] {$\alpha_7$};
	\node at (-0.5,1.5) [font=\footnotesize] {$p_1$};
	\node at (-0.5,-0.5) [font=\footnotesize] {$p_2$};
	\node at (2.5,-0.5) [font=\footnotesize] {$p_3$};
	\node at (2.5,1.5) [font=\footnotesize] {$p_4$};
\end{tikzpicture}
\caption{Double-box with\\ an outer massive loop,\\ $G = \texttt{outer-dbox}$}
\end{subfigure}
\begin{subfigure}[c]{0.3\textwidth}
	\centering
	\begin{tikzpicture}[scale = 0.9,line width=1,scale=1.3]
		\coordinate (v1) at (0,0);
		\coordinate (v2) at (0,1);
		\coordinate (v3) at (1,1);
		\coordinate (v4) at (2,1);
		\coordinate (v5) at (2,0);
		\coordinate (v6) at (1,0);
		
		\draw[massless] (v1) -- (v2);
		\draw[massless] (v2) -- (v3);
		\draw[massless] (v1) -- (v6);
		\draw[massive] (v3) -- (v4);
		\draw[massive] (v4) -- (v6);
		\filldraw[white] (1.5,0.5) circle (4pt);
		\draw[massive] (v3) -- (v5);
		\draw[massive] (v5) -- (v6);
		\draw[massless] (v1) -- ++(-135:0.5);
		\draw[massless] (v2) -- ++(135:0.5);
		\draw[dashed] (v4) -- ++(45:0.6);
		\draw[massless] (v5) -- ++(-45:0.5);
		\foreach \point in {v1, v2, v3, v4, v5, v6} {
			\fill[white!20!black] (\point) circle (1.5pt);
		}
		\node at (0,0.5) [font=\footnotesize, anchor=east] {$\alpha_1$};
		\node at (0.5,-0.1) [font=\footnotesize, anchor=north] {$\alpha_2$};
		\node at (1.5,-0.1) [font=\footnotesize, anchor=north] {$\alpha_3$};
		\node at (1.7,0.6) [font=\footnotesize, anchor=west] {$\alpha_4$};
		\node at (1.5,1) [font=\footnotesize, anchor=south] {$\alpha_5$};
		\node at (0.5,1) [font=\footnotesize, anchor=south] {$\alpha_6$};
		\node at (0.8,0.6) [font=\footnotesize, anchor=west] {$\alpha_7$};
		\node at (-0.5,1.5) [font=\footnotesize] {$p_1$};
		\node at (-0.5,-0.5) [font=\footnotesize] {$p_2$};
		\node at (2.5,-0.5) [font=\footnotesize] {$p_3$};
		\node at (2.65,1.5) [font=\footnotesize] {$p_4$};
	\end{tikzpicture}
	\caption{Non-planar double-box for\\ Higgs + jet production, $G = \texttt{Hj-npl-dbox}$}
\end{subfigure}

\begin{subfigure}[c]{0.3\textwidth}
\centering
\begin{tikzpicture}[scale = 0.9,line width=1,scale=1.3]
    \coordinate (v1) at (0,0);
    \coordinate (v2) at (0,1);
    \coordinate (v3) at (1,1);
    \coordinate (v4) at (2,1);
    \coordinate (v5) at (2,0);
    \coordinate (v6) at (1,0);
    \draw[photon] (v1) -- (v2);
    \draw[photon] (v6) -- (v3);
    \draw[photon] (v5) -- (v4);
    \draw[massive,Maroon] (v2) -- (v3) -- (v4);
    \draw[massive,RoyalBlue] (v1) -- (v6) -- (v5);
    \draw[massive,RoyalBlue] (v1) -- ++(-135:0.5);
    \draw[massive,Maroon] (v2) -- ++(135:0.5);
    \draw[massive,Maroon] (v4) -- ++(45:0.5);
    \draw[massive,RoyalBlue] (v5) -- ++(-45:0.5);
    \foreach \point in {v1, v2, v3, v4, v5, v6} {
    	\fill[white!20!black] (\point) circle (1.5pt);
    }
	\node at (0,0.5) [font=\footnotesize, anchor=east] {$\alpha_1$};
	\node at (0.5,-0.1) [font=\footnotesize, anchor=north] {$\alpha_2$};
	\node at (1.5,-0.1) [font=\footnotesize, anchor=north] {$\alpha_3$};
	\node at (2,0.5) [font=\footnotesize, anchor=west] {$\alpha_4$};
	\node at (1.5,1) [font=\footnotesize, anchor=south] {$\alpha_5$};
	\node at (0.5,1) [font=\footnotesize, anchor=south] {$\alpha_6$};
	\node at (1,0.5) [font=\footnotesize, anchor=west] {$\alpha_7$};
	\node at (-0.5,1.5) [font=\footnotesize] {$p_1$};
	\node at (-0.5,-0.5) [font=\footnotesize] {$p_2$};
	\node at (2.5,-0.5) [font=\footnotesize] {$p_3$};
	\node at (2.5,1.5) [font=\footnotesize] {$p_4$};
\end{tikzpicture}
\caption{Double-box for Bhabha\\ scattering, $G = \texttt{Bhabha-dbox}$}
\end{subfigure}
\begin{subfigure}[c]{0.3\textwidth}
\centering
\begin{tikzpicture}[scale = 0.9,line width=1,scale=1.3]
    \coordinate (v1) at (0,0);
    \coordinate (v2) at (0,1);
    \coordinate (v3) at (1,1);
    \coordinate (v4) at (2,1);
    \coordinate (v5) at (2,0);
    \coordinate (v6) at (1,0);
    \draw[photon] (v1) -- (v2);
    \draw[photon] (v3) -- (v4);
    \draw[photon] (v5) -- (v6);
    \draw[massive,RoyalBlue] (v4) -- (v5);
    \draw[massive,Maroon] (v1) -- (v6) -- (v3) -- (v2);
    \draw[massive,Maroon] (v1) -- ++(-135:0.5);
    \draw[massive,Maroon] (v2) -- ++(135:0.5);
    \draw[massive,RoyalBlue] (v4) -- ++(45:0.5);
    \draw[massive,RoyalBlue] (v5) -- ++(-45:0.5);
    \foreach \point in {v1, v2, v3, v4, v5, v6} {
    	\fill[white!20!black] (\point) circle (1.5pt);
    }
	\node at (0,0.5) [font=\footnotesize, anchor=east] {$\alpha_1$};
	\node at (0.5,-0.1) [font=\footnotesize, anchor=north] {$\alpha_2$};
	\node at (1.5,-0.1) [font=\footnotesize, anchor=north] {$\alpha_3$};
	\node at (2,0.5) [font=\footnotesize, anchor=west] {$\alpha_4$};
	\node at (1.5,1) [font=\footnotesize, anchor=south] {$\alpha_5$};
	\node at (0.5,1) [font=\footnotesize, anchor=south] {$\alpha_6$};
	\node at (1,0.5) [font=\footnotesize, anchor=west] {$\alpha_7$};
	\node at (-0.5,1.5) [font=\footnotesize] {$p_1$};
	\node at (-0.5,-0.5) [font=\footnotesize] {$p_2$};
	\node at (2.5,-0.5) [font=\footnotesize] {$p_3$};
	\node at (2.5,1.5) [font=\footnotesize] {$p_4$};
\end{tikzpicture}
\caption{Second double-box for Bhabha scattering, $G = \texttt{Bhabha2-dbox}$}
\end{subfigure}
\begin{subfigure}[c]{0.3\textwidth}
\centering
\vspace{0.7em}
\begin{tikzpicture}[scale = 0.9,line width=1,scale=1.3]
    \coordinate (v1) at (0,0);
    \coordinate (v2) at (0,1);
    \coordinate (v3) at (1,1);
    \coordinate (v4) at (2,1);
    \coordinate (v5) at (2,0);
    \coordinate (v6) at (1,0);
    
    \draw[photon] (v6) -- (v4);
    \filldraw[white] (1.5,0.5) circle (4pt);
    \draw[photon] (v3) -- (v5);
    \draw[photon] (v1) -- (v2);
    \draw[massive,Maroon] (v2) -- (v3) -- (v4);
    \draw[massive,RoyalBlue] (v1) -- (v6) -- (v5);
    \draw[massive,RoyalBlue] (v1) -- ++(-135:0.5);
    \draw[massive,Maroon] (v2) -- ++(135:0.5);
    \draw[massive,Maroon] (v4) -- ++(45:0.5);
    \draw[massive,RoyalBlue] (v5) -- ++(-45:0.5);
    \foreach \point in {v1, v2, v3, v4, v5, v6} {
    	\fill[white!20!black] (\point) circle (1.5pt);
    }
	\node at (0,0.5) [font=\footnotesize, anchor=east] {$\alpha_1$};
	\node at (0.5,-0.1) [font=\footnotesize, anchor=north] {$\alpha_2$};
	\node at (1.5,-0.1) [font=\footnotesize, anchor=north] {$\alpha_3$};
	\node at (1.7,0.6) [font=\footnotesize, anchor=west] {$\alpha_4$};
	\node at (1.5,1) [font=\footnotesize, anchor=south] {$\alpha_5$};
	\node at (0.5,1) [font=\footnotesize, anchor=south] {$\alpha_6$};
	\node at (0.7,0.6) [font=\footnotesize, anchor=west] {$\alpha_7$};
	\node at (-0.5,1.5) [font=\footnotesize] {$p_1$};
	\node at (-0.5,-0.5) [font=\footnotesize] {$p_2$};
	\node at (2.5,-0.5) [font=\footnotesize] {$p_3$};
	\node at (2.5,1.5) [font=\footnotesize] {$p_4$};
\end{tikzpicture}
\caption{Non-planar double-box for Bhabha scattering, $G = \texttt{Bhabha-npl-dbox}$}
\end{subfigure}
\\

\begin{subfigure}[c]{0.3\textwidth}
\centering
\begin{tikzpicture}[scale = 0.9,line width=1,scale=1]
    \coordinate (v1) at (-1,0);
    \coordinate (v2) at (0,1);
    \coordinate (v3) at (0,-1);
    \coordinate (v4) at (1.5,0);
    \draw[massive,Maroon] (v1) -- (v2);
    \draw[massive,RoyalBlue] (v1) -- (v3);
    \draw[massive,orange] (v2) -- (v3);
    \draw[massive,OliveGreen] (v2) -- (v4);
    \draw[massive,red] (v3) -- (v4);
    \draw[massive] (v1) -- ++(-150:0.5);
    \draw[massive] (v1) -- ++(-210:0.5);
    \draw[massive] (v4) -- ++(30:0.5);
    \draw[massive] (v4) -- ++(-30:0.5);
    \foreach \point in {v1, v2, v3, v4} {
    	\fill[white!20!black] (\point) circle (1.5pt);
    }
	\node at (-0.4,-0.6) [font=\footnotesize, anchor=east] {$\alpha_1$};
	\node at (1.4,-0.7) [font=\footnotesize, anchor=east] {$\alpha_2$};
	\node at (-0.4,0.7) [font=\footnotesize, anchor=east] {$\alpha_4$};
	\node at (1.4,0.7) [font=\footnotesize, anchor=east] {$\alpha_3$};
	\node at (0.7,0) [font=\footnotesize, anchor=east] {$\alpha_5$};
	\node at (-1.4,0.3) [font=\footnotesize, anchor=east] {$p_1$};
	\node at (-1.4,-0.3) [font=\footnotesize, anchor=east] {$p_2$};
	\node at (2.55,0.3) [font=\footnotesize, anchor=east] {$p_4$};
	\node at (2.55,-0.3) [font=\footnotesize, anchor=east] {$p_3$};
\end{tikzpicture}
\caption{Kite diagram with generic masses, $G = \texttt{kite}$}
\end{subfigure}
\begin{subfigure}[c]{0.3\textwidth}
\centering
\begin{tikzpicture}[scale = 0.9,line width=1,scale=1]
    \coordinate (v1) at (-1,0);
    \coordinate (v2) at (0,0.8);
    \coordinate (v3) at (0,-0.8);
    \draw[massive,Maroon] (v1) -- (v2);
    \draw[massive,RoyalBlue] (v1) -- (v3);
    \draw[massive,orange] (v2) to[out=-65,in=65] (v3);
    \draw[massive,OliveGreen] (v2) to[out=-105,in=105] (v3);
    \draw[massive] (v1) -- ++(-150:0.5);
    \draw[massive] (v1) -- ++(-210:0.5);
    \draw[massive,violet] (v2) -- ++(15:0.5);
    \draw[massive,pink] (v3) -- ++(-15:0.5);
    \foreach \point in {v1, v2, v3} {
    	\fill[white!20!black] (\point) circle (1.5pt);
    }
	\node at (-0.4,-0.6) [font=\footnotesize, anchor=east] {$\alpha_2$};
	\node at (-0.1,0) [font=\footnotesize, anchor=east] {$\alpha_3$};
	\node at (-0.4,0.7) [font=\footnotesize, anchor=east] {$\alpha_1$};
	\node at (0.9,0) [font=\footnotesize, anchor=east] {$\alpha_4$};
	\node at (-1.4,0.3) [font=\footnotesize, anchor=east] {$p_1$};
	\node at (-1.4,-0.3) [font=\footnotesize, anchor=east] {$p_2$};
	\node at (1.2,0.9) [font=\footnotesize, anchor=east] {$p_4$};
	\node at (1.2,-1) [font=\footnotesize, anchor=east] {$p_3$};
\end{tikzpicture}
\caption{\label{subfig:par}Parachute diagram with generic masses, $G = \texttt{par}$}
\end{subfigure}
\begin{subfigure}[c]{0.3\textwidth}
\hspace*{-0.5cm}
	\begin{tikzpicture}[scale = 0.9,line width=1,scale=1.3]
		\coordinate (v1) at -(0.2,-0.2);
		\coordinate (v7) at -(-0.3,0.5);
		\coordinate (v2) at -(0.2,1.2);
		\coordinate (v3) at -(1,1);
		\coordinate (v4) at -(2,1);
		\coordinate (v5) at -(2,0);
		\coordinate (v6) at -(1,0);
		\draw[massless] (v6) -- (v1) -- (v7) -- (v2) -- (v3);
		\draw[massive] (v3) -- (v4) -- (v6);
		\filldraw[white] (1.5,0.5) circle (4pt);
		\draw[massive] (v3) -- (v5);
		\draw[massive] (v5) -- (v6);
		\draw[massless] (v1) -- ++(-135:0.5);
		\draw[massless] (v2) -- ++(135:0.5);
		\draw[dashed] (v4) -- ++(45:0.6);
		\draw[massless] (v5) -- ++(-45:0.5);
		\draw[massless] (v7) -- ++(-180:0.5);
		\foreach \point in {v1, v2, v3, v4, v5, v6, v7} {
			\fill[white!20!black] (\point) circle (1.5pt);
		}
		\node at (0,1) [font=\footnotesize, anchor=east] {$\alpha_1$};
		\node at (0,0) [font=\footnotesize, anchor=east] {$\alpha_2$};
		\node at (0.6,-0.2) [font=\footnotesize, anchor=north] {$\alpha_3$};
		\node at (1.5,-0.1) [font=\footnotesize, anchor=north] {$\alpha_4$};
		\node at (1.7,0.6) [font=\footnotesize, anchor=west] {$\alpha_5$};
		\node at (1.5,1) [font=\footnotesize, anchor=south] {$\alpha_6$};
		\node at (0.6,1.1) [font=\footnotesize, anchor=south] {$\alpha_7$};
		\node at (0.8,0.6) [font=\footnotesize, anchor=west] {$\alpha_8$};
		\node at (-0.4,1.5) [font=\footnotesize] {$p_1$};
		\node at (-1,0.5) [font=\footnotesize] {$p_2$};
		\node at (-0.4,-0.5) [font=\footnotesize] {$p_3$};
		\node at (2.5,-0.5) [font=\footnotesize] {$p_4$};
		\node at (2.65,1.55) [font=\footnotesize] {$p_5$};
	\end{tikzpicture}
	\caption{Non-planar penta-box fir Higgs + jet production, $G = \texttt{Hj-npl-pentb}$}
\end{subfigure}

\begin{subfigure}[c]{0.3\textwidth}
\hspace*{-0.6cm}
\begin{tikzpicture}[scale = 0.9,line width=1,scale=1.3]
    \coordinate (v1) at (-0.2,0.2);
    \coordinate (v7) at (0.3,-0.5);
    \coordinate (v2) at (-0.2,-1.2);
    \coordinate (v3) at (-1,-1);
    \coordinate (v4) at (-1.8,-1.2);
    \coordinate (v8) at (-2.3,-0.5);
    \coordinate (v5) at (-1.8,0.2);
    \coordinate (v6) at (-1,0);
    \draw[massless] (v1) -- (v7) -- (v2) -- (v3) -- (v4) -- (v8) -- (v5) -- (v6) -- (v1);
    \draw[massless] (v3) -- (v6);
    \draw[massless] (v1) -- ++(45:0.5);
    \draw[massless] (v2) -- ++(-45:0.5);
    \draw[massless] (v4) -- ++(-135:0.5);
    \draw[massless] (v5) -- ++(135:0.5);
    \draw[massless] (v7) -- ++(0:0.5);
    \draw[massless] (v8) -- ++(180:0.5);
    \foreach \point in {v1, v2, v3, v4, v5, v6, v7, v8} {
    	\fill[white!20!black] (\point) circle (1.5pt);
    }
	\node at (-2,0) [font=\footnotesize, anchor=east] {$\alpha_1$};
	\node at (-2,-1) [font=\footnotesize, anchor=east] {$\alpha_2$};
	\node at (-1.4,-1.2) [font=\footnotesize, anchor=north] {$\alpha_3$};
	\node at (-0.6,-1.2) [font=\footnotesize, anchor=north] {$\alpha_4$};
	\node at (-0,-1) [font=\footnotesize, anchor=west] {$\alpha_5$};
	\node at (0,0) [font=\footnotesize, anchor=west] {$\alpha_6$};
	\node at (-0.6,0.15) [font=\footnotesize, anchor=south] {$\alpha_7$};
	\node at (-1.4,0.15) [font=\footnotesize, anchor=south] {$\alpha_8$};
	\node at (-0.95,-0.5) [font=\footnotesize, anchor=west] {$\alpha_9$};
	\node at (-2.4,0.5) [font=\footnotesize] {$p_1$};
	\node at (-3,-0.5) [font=\footnotesize] {$p_2$};
	\node at (-2.4,-1.5) [font=\footnotesize] {$p_3$};
	\node at (0.4,-1.5) [font=\footnotesize] {$p_4$};
	\node at (0.8,-0.7) [font=\footnotesize] {$p_5$};
	\node at (0.4,0.5) [font=\footnotesize] {$p_6$};
\end{tikzpicture}
\caption{Massless planar double-pentagon, $G = \texttt{dpent}$}
\end{subfigure}
\begin{subfigure}[c]{0.3\textwidth}
	\centering
	\begin{tikzpicture}[scale = 0.9,line width=1,scale=1.3]
		\coordinate (v1) at (-0.2,0.2);
		\coordinate (v7) at (0.1,-0.3);
		\coordinate (v2) at (-0.2,-1.2);
		\coordinate (v3) at (-1,-1);
		\coordinate (v4) at (-2,-1);
		\coordinate (v8) at (0.1,-0.8);
		\coordinate (v5) at (-2,0);
		\coordinate (v6) at (-1,0);
		\draw[massless] (v5) -- (v4) -- (v3) -- (v2) -- (v8) -- (v6) -- (v5);
		\filldraw[white] (-0.3,-0.55) circle (4pt);
		\draw[massless] (v3) -- (v7) -- (v1) -- (v6);
		\draw[massless] (v1) -- ++(45:0.5);
		\draw[massless] (v2) -- ++(-45:0.5);
		\draw[massless] (v4) -- ++(-135:0.5);
		\draw[massless] (v5) -- ++(135:0.5);
		\draw[massless] (v7) -- ++(0:0.5);
		\draw[massless] (v8) -- ++(0:0.5);
		\foreach \point in {v1, v2, v3, v4, v5, v6, v7, v8} {
			\fill[white!20!black] (\point) circle (1.5pt);
		}
		\node at (-2,-0.5) [font=\footnotesize, anchor=east] {$\alpha_1$};
		\node at (-1.4,-1.1) [font=\footnotesize, anchor=north] {$\alpha_2$};
		\node at (-0.6,-1.2) [font=\footnotesize, anchor=north] {$\alpha_3$};
		\node at (-0.05,-1.15) [font=\footnotesize, anchor=west] {$\alpha_4$};
		\node at (-0.05,0.05) [font=\footnotesize, anchor=west] {$\alpha_8$};
		\node at (-0.6,0.15) [font=\footnotesize, anchor=south] {$\alpha_7$};
		\node at (-1.4,0.05) [font=\footnotesize, anchor=south] {$\alpha_6$};
		\node at (-1.3,-0.7) [font=\footnotesize, anchor=west] {$\alpha_9$};
		\node at (-1.3,-0.3) [font=\footnotesize, anchor=west] {$\alpha_5$};
		\node at (-2.4,0.5) [font=\footnotesize] {$p_1$};
		\node at (0.8,-0.95) [font=\footnotesize] {$p_4$};
		\node at (-2.4,-1.5) [font=\footnotesize] {$p_2$};
		\node at (0.4,-1.5) [font=\footnotesize] {$p_3$};
		\node at (0.8,-0.25) [font=\footnotesize] {$p_6$};
		\node at (0.4,0.5) [font=\footnotesize] {$p_5$};
	\end{tikzpicture}
	\caption{Massless non-planar double-pentagon, $G = \texttt{npl-dpent}$}
\end{subfigure}
\begin{subfigure}[c]{0.3\textwidth}
	\hspace*{-0.5cm}
	\vspace*{-0.2cm}
	\begin{tikzpicture}[scale = 0.9,line width=1,scale=1.3]
		\coordinate (v1) at (-0.2,0.2);
		\coordinate (v7) at (0.1,-0.3);
		\coordinate (v2) at (-0.2,-1.2);
		\coordinate (v3) at (-1,-1);
		\coordinate (v4) at (-1.8,-1.2);
		\coordinate (v8) at (-2.3,-0.5);
		\coordinate (v5) at (-1.8,0.2);
		\coordinate (v6) at (-1,0);
		\draw[massless] (v5) -- (v8) -- (v4) -- (v3) -- (v2) -- (v6) -- (v5);
		\filldraw[white] (-0.5,-0.7) circle (4pt);
		\draw[massless] (v3) -- (v7) -- (v1) -- (v6);
		\draw[massless] (v1) -- ++(45:0.5);
		\draw[massless] (v2) -- ++(-45:0.5);
		\draw[massless] (v4) -- ++(-135:0.5);
		\draw[massless] (v5) -- ++(135:0.5);
		\draw[massless] (v7) -- ++(0:0.5);
		\draw[massless] (v8) -- ++(180:0.5);
		\foreach \point in {v1, v2, v3, v4, v5, v6, v7, v8} {
			\fill[white!20!black] (\point) circle (1.5pt);
		}
		\node at (-2,0) [font=\footnotesize, anchor=east] {$\alpha_1$};
		\node at (-2,-1) [font=\footnotesize, anchor=east] {$\alpha_2$};
		\node at (-1.4,-1.2) [font=\footnotesize, anchor=north] {$\alpha_3$};
		\node at (-0.6,-1.2) [font=\footnotesize, anchor=north] {$\alpha_9$};
		\node at (-0.25,-0.7) [font=\footnotesize, anchor=west] {$\alpha_4$};
		\node at (-0.05,0.05) [font=\footnotesize, anchor=west] {$\alpha_5$};
		\node at (-0.6,0.15) [font=\footnotesize, anchor=south] {$\alpha_6$};
		\node at (-1.4,0.15) [font=\footnotesize, anchor=south] {$\alpha_7$};
		\node at (-1.3,-0.5) [font=\footnotesize, anchor=west] {$\alpha_8$};
		\node at (-2.4,0.5) [font=\footnotesize] {$p_1$};
		\node at (-3,-0.5) [font=\footnotesize] {$p_2$};
		\node at (-2.4,-1.5) [font=\footnotesize] {$p_3$};
		\node at (0.4,-1.5) [font=\footnotesize] {$p_6$};
		\node at (0.8,-0.35) [font=\footnotesize] {$p_4$};
		\node at (0.4,0.5) [font=\footnotesize] {$p_5$};
	\end{tikzpicture}
	\caption{Second massless non-planar double-pentagon, $G = \texttt{npl-dpent2}$}
\end{subfigure}

\caption{\label{fig:diagrams}Catalogue of two-loop examples relevant to Standard Model computations. Wavy and curly lines represent massless particles, while solid and dashed ones are massive. Propagators with the same color have the same mass.}
\end{figure}
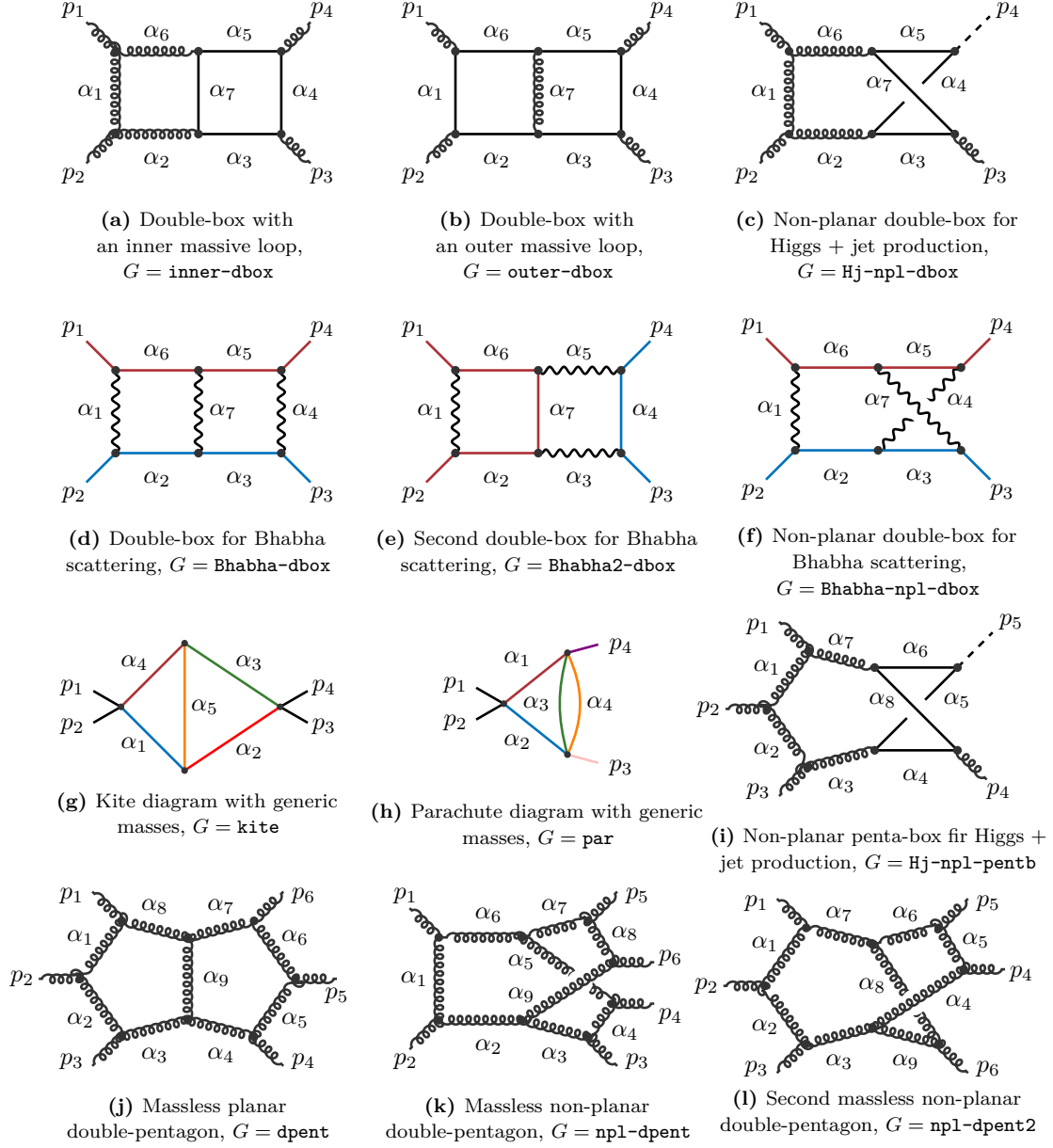

\paragraph{Relation to previous work and historical overview.} 
The literature on Landau singularities is vast and multi-faceted. Here, we outline a few of the most relevant directions that help to put our work in context. After the original papers \cite{Bjorken:1959fd,Landau:1959fi,10.1143/PTP.22.128}, an effort to rigorously define Landau singularities was made by Pham and collaborators in momentum space \cite{AIHPA_1967__6_2_89_0,Pham1968}, see \cite{Hwa:102287,pham2011singularities,Hannesdottir:2022xki} for reviews. His ``Landau variety'' is the projection to the external kinematic space of the critical set of the singularity locus of propagators. At the time, it was only computable for ``generic enough'' integrals such as those associated with one-loop Feynman diagrams with generic masses and no UV/IR divergences, as more complicated cases require compactifications and/or homology with local coefficients \cite{Hwa:102287,pham2011singularities}. Picard--Lefschetz theory was applied to analyze local behavior of finite Feynman integrals around real singularities in generic-mass configurations, see, e.g., \cite{pham1965formules,10.1063/1.1704822,Boyling1966,Hannesdottir:2021kpd,Berghoff:2022mqu}. Independently, Boyling described Landau varieties and compactifications by iterated blow-ups in Schwinger parameter space \cite{Boyling1968}, though they were not applied in practical examples at the time. Decades later, equivalent blow-ups appeared in the motivic approach to Feynman integrals \cite{Bloch:2005bh}. Brown \cite{Brown:2009ta} and Panzer \cite{Panzer:2014caa} reconsidered Landau varieties in the context of linear reducibility and algorithmic evaluation of Feynman integrals in terms of multiple polylogarithms. More recently, compactifications for individual diagrams were studied in \cite{Berghoff:2022mqu} with most advanced examples being the triangle with massless internal edges and the generic-mass parachute diagram.

Parallel work by multiple authors explored the space-time interpretation of Landau singularities and their connection to causality and locality non-perturbatively, where $\alpha$-positive singularities (those with all Schwinger parameters positive or zero) become important, see \cite{iagolnitzer2014scattering,Mizera:2023tfe} for reviews. Coleman and Norton showed that such singularities can be mapped to classical scattering processes \cite{Coleman:1965xm}. Bros, Epstein, and Glaser proved that they cannot appear in certain regions of the kinematic space connecting kinematic channels and establishing crossing symmetry \cite{Bros:1965kbd,Bros:1985gy}, see also \cite{Mizera:2021fap,Caron-Huot:2023ikn}. Chandler and Stapp formulated Landau singularities in terms of macrocausality \cite{PhysRev.174.1749,Chandler:1969bd}, where the notion of essential support of correlation functions \cite{Iagolnitzer:1991wj} plays a central role. Caron-Huot, Giroux, Hannesdottir, and one of the authors extended the Coleman--Norton interpretation to non-$\alpha$-positive singularities for asymptotic observables \cite{Caron-Huot:2023ikn}. Multiple practical ways of calculating $\alpha$-positive singularities are known \cite{Eden:1966dnq}; they can be computed numerically at high loop orders using semi-definite programming techniques \cite{Correia:2021etg}. By contrast with the above approaches, our work considers complex Landau singularities without the positivity condition and is applicable to dimensional regularization.

Landau singularities were studied from the perspective of microlocal analysis and holonomic systems. Sato conjectured that all scattering amplitudes are holonomic \cite{Sato:1975br}, which would imply that around any Landau singularity $\Delta = 0$, they can locally behave only as $\sim \Delta^a \log^b \Delta$ for $a \in \C$ and $b \in \Z_{\geq 0}$ (in the modern language, $\Delta$ are the zeros and singularities of the ``symbol letters'' for polylogarithmic integrals \cite{Maldacena:2015iua}). This conjecture was disproved for scattering amplitudes \cite{Kawai:1981fs,Bros:1983vf}, but it might still hold for individual Feynman integrals, see, e.g., \cite{Kashiwara:1977nf}. It is also known that scattering amplitudes can have accumulations of singularities \cite{Correia:2021etg,Mizera:2022dko}, though it was argued that this cannot happen in physical kinematics \cite{10.1063/1.1705398}; see also \cite{Eberhardt:2022zay} for a discussion in string theory. It has been long known that Feynman integrals can be treated as sufficiently-generalized hypergeometric functions. In particular, techniques from Gelfand--Kapranov--Zelevinsky systems \cite{gelfand2008discriminants} were previously applied to Feynman integrals, see \cite{Klausen:2023gui} for a recent review. Two of the present authors generalized $A$-discriminants to Landau discriminants \cite{Mizera:2021icv}. They did not apply to diagrams with UV/IR divergences (dominant components) and the present work provides an extension to those cases. Principal $A$-determinants were previously applied to Landau analysis in \cite{Klausen:2021yrt,Dlapa:2023cvx}. Our work explains why Feynman integrals are not sufficiently generic for such GKZ results to apply directly, which motivates the introduction of principal Landau determinants.

In massless theories, Landau singularities can be studied in momentum twistor space \cite{Dennen:2015bet,Dennen:2016mdk}. Prlina et al. sketched a proof of a conjecture that in the planar limit, for any diagram with a fixed number of external legs $\n$, all of its first-type Landau singularities are contained in the singular locus of a single ``ziggurat'' diagram \cite{Prlina:2018ukf}. The latter has been determined for $\n \leq 7$ on certain subspaces of the kinematic space \cite{Lippstreu:2023oio}. That work does not take into account different scalings of loop momenta and Schwinger parameters considered here.

Following Libby and Sterman \cite{Libby:1978bx}, Landau equations were also used to determine necessary conditions for IR singularities of off-shell Green's functions (with external masses $\M_i \neq 0$) in QCD in momentum space \cite{Collins:2011zzd}; sufficiency was studied in \cite{Collins:2020euz}. Its modern incarnation is the method of regions \cite{Beneke:1997zp,Jantzen:2012mw} which studies different soft/collinear kinematic regions and uses Newton polytopes to classify rates at which Schwinger parameters contract/expand \cite{Jantzen:2011nz,Ananthanarayan:2018tog,Arkani-Hamed:2022cqe,Gardi:2022khw}. 
This approach is conceptually closest to ours, though it concerns only $\alpha$-positive solutions, while we treat all complex singularities.

\section{Motivation: Singularities and saddle point equations}\label{sec:2}

\subsection{Principal A-determinants} \label{sec:Adet}

Let $A = [m_1 ~ \cdots ~ m_s] \in \mathbb{Z}^{n \times s}$ be an integer matrix with no repeated columns, of rank $n$. The columns $m_i \in \mathbb{Z}^n$ are the exponent vectors appearing in a Laurent polynomial 
\[ f_A(\alpha;z) \, = \, z_1 \, \alpha^{m_1} \, + \, z_2 \, \alpha^{m_2} \, + \, \cdots \, + \, z_s \, \alpha^{m_s}, \]
where $\alpha = (\alpha_1, \ldots, \alpha_n)$ and $\alpha^{m_i}$ is short for the monomial $\alpha_1^{m_{1i}} \cdots \alpha_n^{m_{ni}}$. The coefficients $z_i$ are indeterminates which take complex values. Once coefficients $z \in \mathbb{C}^s$ are fixed, the Laurent polynomial $f(\alpha;z)$ defines a hypersurface in the algebraic torus $(\mathbb{C}^*)^n$: 
\begin{equation} \label{eq:VAz} 
V_{A,z} \, = \, V_{(\mathbb{C}^*)^n}(f_A(\alpha;z)) \, = \, \{ \alpha \in (\mathbb{C}^*)^n \, : \, f_A(\alpha;z) = 0 \}. 
\end{equation}
Here $\mathbb{C}^* = \mathbb{C} \setminus \{0\}$. The coordinate hyperplanes $\{\alpha_i = 0\}$ are excluded since some entries of $A$ may be negative. The \emph{$A$-discriminant} $\Delta_A$ records values of $z$ for which $V_{A,z}$ is a singular hypersurface. More precisely, consider the set  
\[ \nabla_A^\circ \, = \, \Big\{ z \in \mathbb{C}^s \, : \, \exists \alpha \in (\mathbb{C}^*)^n \text{ s.t. } f_A(\alpha;z) = \partial_{\alpha} f_A (\alpha;z) = 0 \Big\}, \]
where we use the notation $\partial_\alpha = (\partial_{\alpha_1}, \ldots,  \partial_{\alpha_n})$ with partial derivatives $\partial_{\alpha_i} = \frac{\partial}{\partial \alpha_i}$ for brevity.
This is in general not a closed subvariety of $\mathbb{C}^s$.
The \emph{$A$-discriminant variety} $\nabla_A$ is obtained by taking the Zariski closure of $\nabla_A^\circ$, which is by definition the smallest algebraic variety containing it. This agrees with the closure in the usual topology. Under mild hypotheses on $A$, the $A$-discriminant variety is a hypersurface (i.e., it has codimension 1 in $\mathbb{C}^s$), and its defining polynomial $\Delta_A$ is the $A$-discriminant polynomial, or simply the $A$-discriminant. This polynomial is defined up to a nonzero scalar multiple, and it can always be taken to have integer coefficients. If ${\rm codim} \nabla_A > 1$, we set $\Delta_A = 1$. 
\begin{example} \label{ex:s=1and2}
When $s = 1$, the Laurent polynomial $f_A = z \, \alpha^m$ only has one term. We have $\Delta_A = z$ in this case. When $s = 2$, one checks that $\Delta_A = 1$.  
\end{example}
The \emph{principal $A$-determinant} $E_A$ is a different polynomial in the coefficients $z_i$, defined via $A$-resultants \cite[Chpt.~10, Sec.~1]{gelfand2008discriminants}. We are mostly interested in the hypersurface defined by this polynomial. With this in mind, it is more convenient to recall the description of $E_A$ as the product of several discriminants, one of which is $\Delta_A$. Let ${\rm Conv}(A) \subset \mathbb{R}^n$ be the convex lattice polytope obtained as the convex hull of the columns $m_i$ of $A$, and let $F(A)$ be the set of all its faces. Here ${\rm Conv}(A)$ is viewed as a face of itself, i.e.~${\rm Conv}(A) \in F(A)$. For a face $Q \in F(A)$, we let $A \cap Q$ be the submatrix of $A$ consisting of all columns $m_i \in Q$. The $(A \cap Q)$-discriminant $\Delta_{A \cap Q}$ is a polynomial in the variables $z_i$, with $m_i \in Q$. Note that $\Delta_A = \Delta_{A \cap Q}$ when $Q = {\rm Conv}(A)$. We have
\begin{equation}\label{eq:EA}
E_A \, = \, \prod_{Q \in F(A)} \Delta_{A \cap Q}^{e_Q},
\end{equation}
for some positive integer exponents $e_Q >0$. A precise formula for these exponents and a statement for the equivalence between the $A$-resultant definition and \eqref{eq:EA} are found in \cite[Chpt.~10]{gelfand2008discriminants}. Here is an easy example. 
\begin{example}[$n = 2, s = 4$] \label{ex:n=2ands=4}
We consider the lattice points of the unit square
\setlength\arraycolsep{5pt}
\[ A = \begin{pmatrix}
0 & 1 & 0 & 1\\ 
0 & 0 & 1 & 1
\end{pmatrix}, \quad \text{which gives} \quad f_A(\alpha,z) \,  =\, z_1 + z_2 \, \alpha_1 + z_3 \, \alpha_2 + z_4 \, \alpha_1\alpha_2.
\]
The curve $V_{A,z}$ is singular when it is the union of a horizontal and a vertical line, this happens when $\Delta_A = z_1z_4-z_2z_3 = 0$. The polygon ${\rm Conv}(A)$ is $[0,1]^2$, and $F(A)$ consists of one 2-dimensional face, four 1-dimensional faces and 4 vertices. For each of the one-dimensional faces $Q$, by Ex.~\ref{ex:s=1and2} we have $\Delta_{A \cap Q} = 1$. The same example says that for the vertex $m_i \in A$ we have $\Delta_{m_i} = z_i$. By the second part of \cite[Chpt.~10, Thm.~1.2]{gelfand2008discriminants}, in this example, all exponents $e_Q$ are equal to 1. Hence, Eq.~\ref{eq:EA} gives
\[ E_A \, = \, z_1 \cdot z_2 \cdot z_3 \cdot z_4 \cdot (z_1 z_4 - z_2 z_3). \qedhere \]
\end{example}

An important topological invariant of the variety $V_{A,z}$ is its Euler characteristic $\chi(V_{A,z})$. The principal $A$-determinant shows up when studying how this number depends on $z$. Let ${\rm Aff}(A) \subset \mathbb{R}^n$ be the smallest affine subspace containing the points $m_1, \ldots, m_s$, i.e., the columns of $A$. We write $\Lambda = {\rm Aff}(A) \cap \mathbb{Z}^n$ for the corresponding affine lattice. Let ${\rm vol}(A)$ be the normalized volume of the lattice polytope ${\rm Conv}(A)$ in the lattice $\Lambda$. 
If ${\rm Conv}(A)$ has dimension $n$, then ${\rm vol}(A)$ is the standard Euclidean volume multiplied with a factor $n!$. The proof of the following result can be found in \cite[Thm. 13]{AMENDOLA2019222}. 
\begin{theorem} \label{thm:eulercharvolume}
The signed Euler characteristic $|\chi(V_{A,z})|$ equals $ {\rm vol}(A)$ if and only if $z \in \mathbb{C}^s \setminus \{ E_A = 0 \}$. Moreover, when $E_A(z) = 0$, we have $|\chi(V_{A,z})|<{\rm vol}(A)$.
\end{theorem}
The Euler characteristic is relevant to us because it counts the number of linearly independent \emph{$A$-hypergeometric functions}. These are given by integrals which are similar to Feynman integrals, as we will see in the next section.

\subsection{GKZ systems vs. Feynman integrals} \label{sec:GKZvsFeyn}
For $i = 1, \ldots, \ell$, let $A_i = [m_{i,1}~\cdots~m_{i,s_i}] \in \mathbb{Z}^{n \times s_i}$ be an integer matrix as in the previous section. The Laurent polynomials $f_{A_i}(\alpha;z) = z_{i1} \alpha^{m_{i,1}} + \cdots + z_{is_i} \alpha^{m_{i,s_i}}$ define an integral 
\begin{align} \label{eq:Eulerintegral} {\cal I}_\Gamma(z) \, &=\, \int_{\Gamma} f^{\mu} \alpha^\nu \, \frac{{\rm d} \alpha}{\alpha}\\
&= \,  \int_{\Gamma} \, f_{A_1}(\alpha;z)^{\mu_1} \cdots f_{A_\ell}(\alpha;z)^{\mu_\ell} \, \alpha_1^{\nu_1} \cdots \alpha_\ell^{\nu_\ell} \, \frac{{\rm d} \alpha_1}{\alpha_1} \wedge \cdots \wedge \frac{{\rm d} \alpha_n}{\alpha_n} .
\end{align}
The exponents $\mu \in \mathbb{C}^\ell$ and $\nu \in \mathbb{C}^n$ are complex numbers, so that the integrand is multi-valued. Let $V_{A_i,z} \subset (\mathbb{C}^*)^n$ be as in \eqref{eq:VAz}. The twisted $n$-cycle $\Gamma$ is an $n$-chain on 
\begin{equation} \label{eq:Xz}
X_z \,=\, (\mathbb{C}^*)^n  \, \setminus \, ( V_{A_1,z} \cup \cdots \cup V_{A_\ell,z} ), 
\end{equation}
with zero \emph{twisted boundary}. Here \emph{twisted} means essentially that $\Gamma$ also records the choice of which branch of $f^\mu \alpha^\nu$ to integrate. The integral \eqref{eq:Eulerintegral} was called a \emph{generalized Euler integral} by Gelfand, Kapranov and Zelevinsky (GKZ) \cite{gelfand1990generalized}. See \cite[Chpt.~2]{AomotoKita} for more details, and \cite{agostini2022vector,matsubara2023four} for recent overviews.

As a function of the coefficients $z_{ij}$, the integral ${\cal I}_\Gamma(z)$ satisfies a system of linear PDE called \emph{GKZ system} \cite[Sec.~4]{agostini2022vector}. This system of diffential equations is encoded by a \emph{$D$-module} denoted $H_{A}(\kappa)$. The parameters are $\kappa = (-\nu,\mu)$, and
\setcounter{MaxMatrixCols}{20}
\setlength\arraycolsep{5pt}
\begin{equation} \label{eq:cayley} A \, = \, \begin{pmatrix}
 & A_1 & & & A_2 & & & \cdots & & & A_\ell  \\
 1& \cdots & 1 & 0 & \cdots & 0 & & & & 0 & \cdots & 0\\ 
 0& \cdots & 0 & 1 & \cdots & 1 & & & & 0 & \cdots & 0\\  
 0& \cdots & 0 & 0 & \cdots & 0 & & \cdots & & 0 & \cdots & 0 \\ 
 0& \cdots & 0 & 0 & \cdots & 0 & & & & 1 & \cdots & 1
\end{pmatrix} \quad \in \,  \mathbb{Z}^{(n + \ell) \times (s_1 + \cdots + s_\ell)}.
\end{equation}
Below, we will write $s = s_1 + \cdots + s_\ell$ and $z = (z_1, \ldots, z_s) \in \mathbb{C}^s$ for brevity. The following remarkable result from \cite[Thms.~1.4 and 2.10]{gelfand1990generalized} demonstrates how the principal $A$-determinant $E_A$ from \eqref{eq:EA} governs the analytic properties of the function ${\cal I}_\Gamma(z)$.
\begin{theorem} \label{thm:GKZsinglocus}
    For generic $\kappa = (-\nu, \mu)$ and for $z^* \in \mathbb{C}^s \setminus \{E_A = 0\}$, the vector space of local solutions to the GKZ system $H_A(\kappa)$ at $z = z^*$ has dimension $(-1)^n \cdot \chi(X_{z^*})  = (-1)^{n + \ell-1} \cdot \chi(V_{A,z^*}) = {\rm vol}(A)$. All solutions are obtained by varying the twisted cycle $\Gamma$ in ${\cal I}_\Gamma(z)$ from \eqref{eq:Eulerintegral}. The singular locus of the $D$-module $H_A(\kappa)$ is the variety $\{E_A = 0 \}$.
\end{theorem}
For the meaning of ${\rm vol}(A)$ in this statement, see the discussion preceding Thm.~\ref{thm:eulercharvolume}.
The Euler integrals \eqref{eq:Eulerintegral} appear in particle physics as \emph{Feynman integrals} \cite{matsubara2023four}. In that case, $\ell = 1$ or $\ell = 2$, and the coefficients of $f_{A_1}, f_{A_2}$ are linear functions of the kinematic parameters. We will discuss the general construction below. See \cite{delaCruz:2019skx,Klausen:2019hrg} for recent literature on the GKZ approach to Feynman integrals.

\begin{example} \label{ex:sunriseproblem}
We work out an illustrative example, corresponding to the \emph{banana diagram} with three internal edges, Fig.~\ref{fig:schlegel} (left). The integral is 
\begin{equation}\label{eq:banana}
{\cal I} \, = \, \int_{\Gamma} [(1 - {\textstyle\sum}_{i=1}^3 \m_i \alpha_i)(\alpha_1\alpha_2 + \alpha_1\alpha_3 + \alpha_2\alpha_3) + s  \alpha_1\alpha_2\alpha_3]^\mu \alpha_1^{\nu_1}\alpha_2^{\nu_2}\alpha_3^{\nu_3} \, \frac{{\rm d}\alpha_1}{\alpha_1} \wedge \frac{{\rm d}\alpha_2}{\alpha_2} \wedge \frac{{\rm d}\alpha_3}{\alpha_3} .\end{equation}
This is a function of $s, \m_1, \m_2, \m_3$. In the above setup, the corresponding matrix is 
\[A \, = \, \begin{pmatrix}
1 & 1 & 0 & 2 & 2 & 0 & 1 & 1 & 0 & 1\\
1 & 0 & 1 & 1 & 0 & 2 & 2 & 0 & 1 & 1\\
0 & 1 & 1 & 0 & 1 & 1 & 0 & 2 & 2 & 1\\ 
1 & 1 & 1 & 1 & 1 & 1 & 1 & 1 & 1 & 1
\end{pmatrix}.
\]
The ten coefficients are either constant, or linear functions of $s, \m_i$:
\begin{equation} \label{eq:zsub}
   (z_1, \ldots, z_{10}) \, = \, (1,1,1,-\m_1,-\m_1,-\m_2,-\m_2,-\m_3,-\m_3,s-\m_1-\m_2-\m_3).
\end{equation}
This parameterizes a four-dimensional affine subspace ${\cal K} \subset \mathbb{C}^{10}$, which in physics is called the \emph{kinematic space}. Motivated by Thm.~\ref{thm:GKZsinglocus}, a first approximation for the singular locus of our integral ${\cal I}(s,\m_1,\m_2,\m_3)$ is ${\cal K} \cap \{E_A = 0 \}$. We chose this example because it nicely illustrates that this is \emph{not} the right approach. 

The polytope ${\rm Conv}(A) \subset \mathbb{R}^3$ has dimension three: it is contained in the 3-dimensional hyperplane in $\mathbb{R}^4$ where the last coordinate is 1. The Schlegel diagram with respect to the hexagonal facet $458967$ of this polytope is shown in Fig.~\ref{fig:schlegel} (right). 
\begin{figure}
    \centering
    \raisebox{3em}{\includegraphics[scale=1]{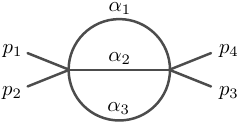}}
    \qquad\qquad
    \begin{tikzpicture}
    \node (img) {\includegraphics[width=5cm]{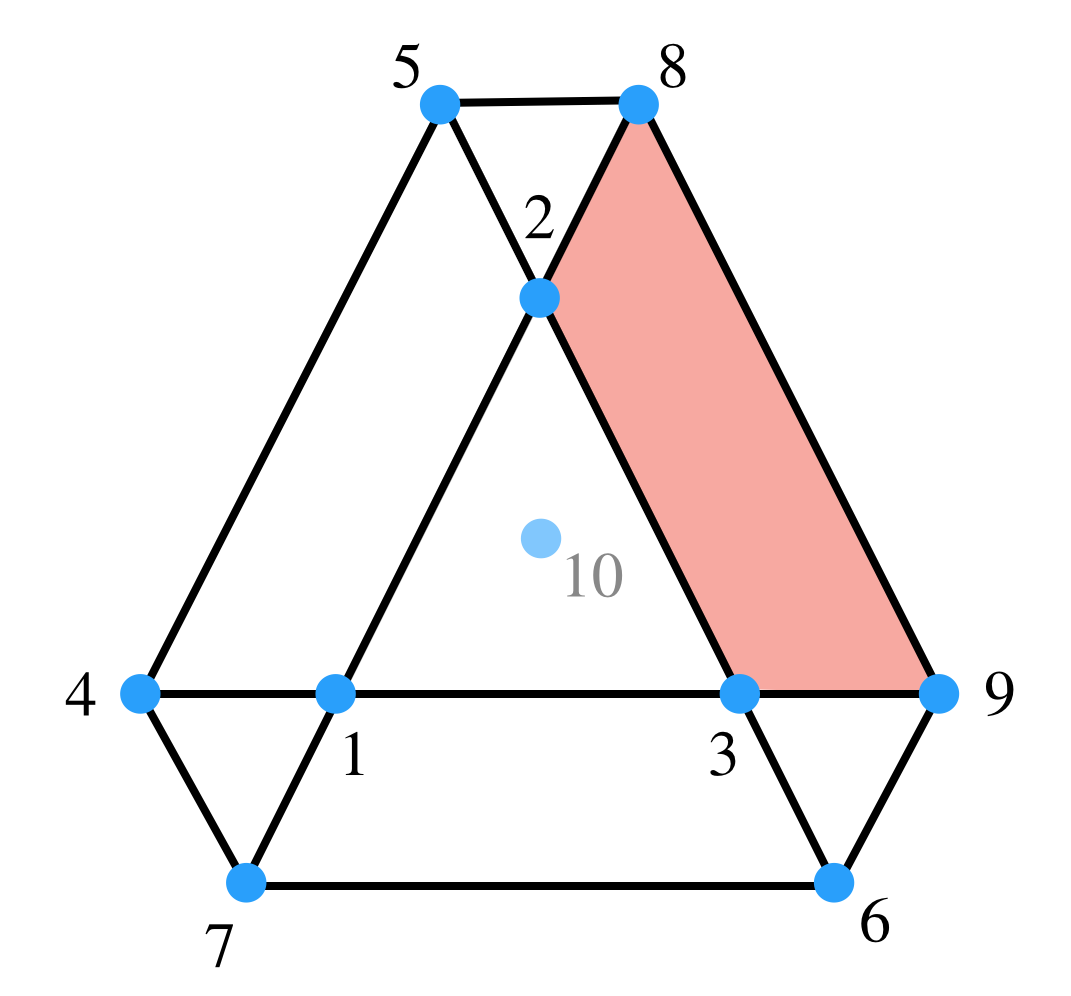}};
    \node [below left, xshift=4cm, yshift=1cm] at (img) {$z_2 z_9 - z_3 z_8$};
    \draw[-{Classical TikZ Rightarrow[length=1mm]}] (1.7,0.6) -- (1,0.2);
    \end{tikzpicture}
    \caption{\label{fig:schlegel}Left: Banana diagram $\mathsf{B}_3$. Right: The polytope ${\rm Conv}(A)$ for the banana diagram with three edges, ${\tt B}_3$. Its $f$-vector is $(9,15,8)$.}
\end{figure}
Here the vertices are labeled consistently with the columns of $A$. The lattice point in the tenth column is an interior point of that facet. Let $Q = 2389 \in F(A)$ be the quadrilateral marked in red. By the formula \eqref{eq:EA}, $\Delta_{A \cap Q}$ is a factor of the principal $A$-determinant $E_A$. One easily checks that the submatrix $A \cap Q$ is obtained from the matrix $A$ from Ex.~\ref{ex:n=2ands=4} by applying the following injective affine integer transformation: 
\[ x^\top \, \longmapsto \, x^\top \begin{pmatrix}
    -1 & 1 &  0  & 0 \\ 0 & 0 &  1  & 0 
\end{pmatrix} + \begin{pmatrix}
    1 & 0 & 1 & 1
\end{pmatrix}.
\]
By \cite[Prop.~1.4]{gelfand2008discriminants}, their discriminants are identical: $\Delta_{A \cap Q} \, = \, z_2z_9-z_3z_8$.
When plugging in \eqref{eq:zsub}, we see that $\Delta_{A\cap Q}({\cal K}) = 0$, and hence $E_A({\cal K}) = 0$. Hence, the containment of the singular locus of ${\cal I}(s,\m_1,\m_2,\m_3)$ in ${\cal K} \cap \{ E_A = 0 \}$ is trivial. In fact, this happens for most Feynman diagrams, see Tab.~\ref{tab:volVSEuler}. To remedy this, a more specialized notion of \emph{discriminants} and \emph{principal determinants} is needed. The first steps along these lines for Feynman integrals were taken in the seminal papers \cite{Bjorken:1959fd,Landau:1959fi,10.1143/PTP.22.128}. In \cite{Mizera:2021icv}, two of the authors formalized the notion of \emph{Landau discriminants}, to account for leading singularities of Feynman integrals. This paper introduces \emph{principal Landau determinants}, see Sec.~\ref{sec:3}. To some extent, these are to Landau discriminants what principal $A$-determinants are to $A$-discriminants. 
\end{example}

\begin{remark}\label{rmk:1}
By Thm.~\ref{thm:eulercharvolume}, the inclusion ${\cal K} \subset \{E_A = 0 \}$ holds if and only if for a generic point $z^* \in {\cal K}$, we have $|\chi(V_{A,z^*})| < {\rm vol}(A)$. This gives a practical way to test this inclusion. The Euler characteristic $\chi(V_{A,z^*})$ can be computed reliably using numerical homotopy methods implemented in the Julia package \texttt{HomotopyContinuation.jl} \cite{10.1007/978-3-319-96418-8_54}, see \cite[Sec.~5]{Mizera:2021icv} or \cite[Sec.~6]{agostini2022vector}. The volume ${\rm vol}(A)$ can be computed, for instance, using the software package \texttt{Oscar.jl} \cite{OSCAR}. In our example above, the inequality is
\be
\chi(V_{A,z^\ast}) \,=\, 7 \,<\, 10 \,=\, {\rm vol}(A)\, .
\ee
Here, ${\rm vol }(A) = 10$ is the volume of ${\rm Conv}(A)$ in the three-dimensional affine space containing it, scaled so that a standard simplex has volume 1. 
\end{remark}

We conclude the section by justifying the above claim that what happens in Ex.~\ref{ex:sunriseproblem} happens for \emph{most Feynman diagrams}. 
Let $G$ be a Feynman diagram with $\E$ internal edges, $\n$ external legs, and graph polynomial ${\cal G}_G={\cal U}_G + {\cal F}_G$. Recall that the coefficients of the graph polynomial are either constant, or linear functions of the kinematic parameters $s_{ij},\m_e,\M_i$. Here we use the notation from \cite[Sec.~2]{Mizera:2021icv}: $s_{ij}$ are Mandelstam invariants, $\m_e = m_e^2$ is the squared mass of the $e$-th internal propagator, and $\M_i = M_i^2$ is the squared mass of the $i$-th external leg. Beyond restricting from generic coefficients in $\mathbb{C}^s$ to generic kinematics in ${\cal K}$, we will allow to put certain internal/external masses to zero. More precisely, we focus on the following meaningful subspaces of parameters:
\begin{itemize} 
    \item ${\cal E}^{(\M_i,0)}\, = \,{\cal K} \cap \{\m_e = 0, \,\, e=1,\dots,\E\},$
    \item ${\cal E}^{(0,\m_e)}\, = \,{\cal K} \cap \{\M_i = 0, \, i=1,\dots,\n\},$ and
    \item ${\cal E}^{(0,0)}\, = \,{\cal K} \cap \{\m_e = \M_i = 0, \, e=1,\dots,\E,\, i=1,\dots,\n\}.$
\end{itemize}
Let ${\cal E} \subseteq {\cal K}$ be any of these spaces. The matrix $A({\cal E})$ has $n=\E$ rows, and its columns are all exponents occurring in ${\cal G}_G$ \textit{for generic choices of kinematics in} ${\cal E}\subset {\cal K}$. The number of columns may depend on ${\cal E}$. In line with Rmk.~\ref{rmk:1}, we tested the inclusion ${\cal E} \subset \{E_A = 0 \}$ for all the graphs from \cite[Fig. 1]{Mizera:2021icv} and Fig.~\ref{fig:diagrams} by comparing the values of the signed Euler characteristic 
and the volume $\text{ vol}(A({\cal E}))$. Here the Euler characteristic we compute is $(-1)^{\E-1}\cdot\chi(V_{A}({\cal E}))$, where $V_{A}({\cal E})$ is the zero locus of ${\cal G}_G$ in $(\mathbb{C}^*)^\E$, for generic kinematics in ${\cal E} \subset {\cal K}$.
In most cases, we indeed have $(-1)^{\E-1}\cdot\chi(V_A({\cal E}))<\hbox{vol}(A({\cal E}))$. This means that ${\cal E}$ is contained in the principal $A$-determinant variety. The only exceptions arise for one-loop and banana diagrams. These cases will be studied in detail in Sec. \ref{sec:4}.

\begin{table}[t]
\begin{center}
\begin{footnotesize}
\begin{tabular}{ c | c | c | c | c}
  $G$ & ${\cal K}$ & ${\cal E}^{(\M_i,0)}$ & ${\cal E}^{(0,\m_e)}$& ${\cal E}^{(0,0)}$   \\
  \hline
 $\texttt{A}_4$ & \textcolor{blue}{(15,\,15)} & \textcolor{blue}{(11,\,11)} & (11,\,15) & \textcolor{blue}{(3,\,3)} \\
 $\texttt{B}_4$ & (15,\,35) & \textcolor{blue}{(1,1)} & (15,35) & \textcolor{blue}{(1,1)}\\ 
 \texttt{par} & (19,\,35) & (4,\,8) & (13,\,35) & (1,\,3) \\
 \texttt{acn} & (55,\,136) & (20,\,54) & (36,\,136) & (3,\,9)\\
 \texttt{env} & (273,\,1496) & (56,\,262) &(181,\,1496) & (10,\,80)\\
 \texttt{npltrb} & (116,\,512) & (28,\,252) &(77,\,512) &(5,\,61) \\
 \texttt{tdetri} & (51,\,201) & (4,\,18) &(33,\,201) & (1,\,5)\\
 \texttt{debox} & (43,\,96) & (11,\,33) & (31,\,96)& (3,\,10)\\
 \texttt{tdebox} & (123,\,705) & (11,\,113) &(87,\,705) &(3,\,41) \\
 \texttt{pltrb} & (81,\,417) & (16,\,201) &(61,\,417) &(4,\,80) \\
 \texttt{dbox} & (227,\,1422) & (75,\,903) & (159,\,1422) & (12,\,238)\\
  \texttt{pentb} & (543,\,4279) & (228,\,3148) & (430,\,4279) & (62,\,1186) \\
\end{tabular}
\quad
\begin{tabular}{ c | c }
  $G$ & ${\cal E}$ \\
  \hline
 \texttt{inner-dbox} & (43,\,834)  \\
 \texttt{outer-dbox} & (64,\,1302)\\
 \texttt{Hj-npl-dbox} & (99,\,1016)\\
 \texttt{Bhabha-dbox} & (64,\,774)\\
 \texttt{Bhabha2-dbox} & (79,\,910) \\
 \texttt{Bhabha-npl-dbox} & (111,\,936) \\
 \texttt{kite} & (30,\,136)\\
 \texttt{par} & (19,\,35)\\
 \texttt{Hj-npl-pentb} & (330,\,3144) \\
  \texttt{dpent} & (281,\,5511)\\
  \texttt{npl-dpent} & (631,\,5784)\\
  \texttt{npl-dpent2} & (458,\,5467)\\
\end{tabular}
\caption{Comparing the signed Euler characteristic and volume $((-1)^{\E-1}\cdot\chi(V_A({\cal E})),\,\text{vol}(A({\cal E})))$ for the Feynman diagrams in \cite[Fig.~1]{Mizera:2021icv} and Fig.~\ref{fig:diagrams}. The kinematic subspace ${\cal E}$ in the left table is the full kinematic space ${\cal K}$ or one of its linear subspaces obtained from setting internal/external masses to zero. In the table on the righthand side, it denotes the custom choices of kinematics fixed in Fig.~\ref{fig:diagrams}. The blue values indicate the cases where volume and Euler characteristic coincide.
Notice that, for each diagram, $\text{vol}(A({\cal K}))=\text{vol}(A({\cal E}^{(0,\m_e)}))$, since the vanishing of external masses does not change the monomial support of the graph polynomial.}
\label{tab:volVSEuler}
\end{footnotesize}
\end{center}
\end{table}
The numbers in Tab. \ref{tab:volVSEuler} were computed using \texttt{Julia}. In particular, the following code lines illustrate how to compute the volume for the graph $G = \texttt{par}$: 
\vspace{-0.5cm}
\begin{minted}{julia}
using PLD
using Oscar

edges = [[3,1],[1,2],[2,3],[2,3]];
nodes = [1,1,2,3];
E = length(edges);
U, F, α, p = getUF(edges, nodes, internal_masses = :generic, 
                                 external_masses = :generic);

ConvA = newton_polytope(U+F)
factorial(E)*volume(ConvA)
\end{minted}

We relied on our \texttt{Julia} package \texttt{PLD.jl} to generate the Symanzik polynomials of each graph $G$ with $\E$ edges. The volume function is instead a feature of the computer algebra system \texttt{Oscar} \cite{OSCAR} available as a \texttt{Julia} package. A snippet for computing the Euler characteristic will be displayed in Sec.~\ref{sec:3}. 

\section{Landau analysis} \label{sec:3}

By Thm.~\ref{thm:eulercharvolume}, the principal $A$-determinant characterizes when the signed Euler characteristic of $X_z$ drops below the generic value. Moreover, by Thm.~\ref{thm:GKZsinglocus}, it coincides with the singular locus of the solutions to a GKZ system. Such solutions are the integrals \eqref{eq:Eulerintegral}, where $\Gamma$ ranges over the twisted homology of $X_z$. The roles of $E_A$ in Thms.~\ref{thm:eulercharvolume} and \ref{thm:GKZsinglocus} are strongly related to each other. When the Euler characteristic drops, there are fewer independent twisted $n$-cycles, which means fewer independent integrals of the form \eqref{eq:Eulerintegral}. Different local solutions to the GKZ system may collide or diverge near $\{E_A = 0\}$, and this causes singularities. An example is found on page 1 of \cite{saito2013grobner}.

Feynman integrals are of the form \eqref{eq:Eulerintegral}, where $\ell = 1$ or $\ell = 2$ and the coefficients $z$ are specialized to the kinematic space ${\cal K} \subset \mathbb{C}^s$. As we illustrated in Sec.~\ref{sec:GKZvsFeyn}, it may well be that the principal $A$-determinant vanishes identically on ${\cal K}$. Then the signed Euler characteristic for generic $z \in {\cal K}$ is $\chi^* < {\rm vol}(A)$. To capture the singular locus of Feynman integrals, we need to detect values $z \in {\cal K}$ for which $(-1)^n \cdot \chi(X_z) < \chi^*$, or, values $z \in {\cal K}$ for which the hypersurface $\{ \alpha \, : \, {\cal U}(\alpha) + {\cal F}(\alpha;z) = 0 \}$ is \emph{more singular than usual}. Here ${\cal U}$ and ${\cal F}$ are \emph{Symanzik polynomials}, whose definition is recalled in \cite[Sec.~2.2]{Mizera:2021icv}. The sum of these two polynomials is called the \emph{graph polynomial} or \emph{Lee-Pomeransky polynomial}. It is denoted by ${\cal G}  = {\cal U} + {\cal F}$. This section describes how we propose to detect such ``extra singular'' $z$-values geometrically. 

\subsection{Euler discriminants}

We start with our general set-up from Sec.~\ref{sec:GKZvsFeyn}. We investigate the Euler characteristic of the very affine variety $X_z$ in \eqref{eq:Xz} as a function of $z \in \mathbb{C}^s$. Recall that
\[ X_z \, = \, \{ \alpha \in (\mathbb{C}^*)^n \, : \, f_i(\alpha,z) \neq 0, \, i = 1, \ldots, \ell \}. \]
Importantly, we do not consider the full parameter space $\mathbb{C}^s$, but we work on a subspace denoted by ${\cal E}$. In Landau analysis, ${\cal E} = {\cal K}$ is kinematic space, or a linear subspace.
\begin{theorem} \label{thm:closedstrata}
    For any irreducible subvariety ${\cal E} \subset \mathbb{C}^s$ and any integer $k$, define
    \[ Z_k({\cal E}) \, = \, \{ z \in {\cal E} \,: \, | \chi(X_z) | \, \leq \,  k \}, \quad \V_k({\cal E}) \, = \, \{ z \in {\cal E} \,: \, | \chi(X_z) | \, = \,  k \}.\]
    For each $k$, $Z_k({\cal E}) = \bigcup_{j=0}^k \V_j({\cal E})$ is Zariski closed in ${\cal E}$. In particular, for the maximal value $\chi^* = \max_{z' \in {\cal E}} |\chi(X_{z'})|$, we have $Z_{\chi^*}({\cal E}) \, = \, {\cal E}$ and $\V_{\chi^*}({\cal E})$ 
    is open and dense in ${\cal E}$.  
\end{theorem}
We will prove Thm.~\ref{thm:closedstrata} in this section. It justifies the following definition of the \emph{Euler discriminant}, which is a polynomial in the coordinate ring $\mathbb{C}[{\cal E}]$ of ${\cal E}$, vanishing at $z \in {\cal E}$ if and only if $|\chi(X_z)|$ is smaller than usual. 
\begin{definition}[Euler discriminant]\label{def:chidiscriminant}
    With the notation of Thm.~\ref{thm:closedstrata}, the \emph{Euler discriminant variety} of the family $X_z$ of very affine varieties over ${\cal E}$ is the closed subvariety $\nabla_\chi({\cal E}) =  Z_{\chi^* - 1}({\cal E}) = {\cal E} \setminus \V_{\chi^*}({\cal E}) \subset \mathbb{C}^s$. If $\nabla_\chi({\cal E})$ is defined by a single equation $\Delta_\chi({\cal E}) \in \mathbb{C}[{\cal E}]$ (unique up to scaling), then we call $\Delta_\chi({\cal E})$ the \emph{Euler discriminant}. 
\end{definition}

\begin{remark}
    The terminology \emph{Euler discriminant} was used by Esterov in \cite[Definition 3.1]{esterov2013discriminant} in a more general context. More precisely, in that paper, the Euler discriminant is a Weil divisor in ${\cal E}$ whose support equals the codimension-1 part of $\nabla_\chi({\cal E})$. 
\end{remark}
\begin{example} \label{ex:chivsPAD}
    By Thm.~\ref{thm:eulercharvolume}, when ${\cal E} = \mathbb{C}^s$, $\chi^* = {\rm vol}(A)$ and the Euler discriminant is the principal $A$-determinant: $\Delta_\chi({\cal E}) = E_A$.
\end{example}
An important challenge in Landau analysis is to compute the Euler discriminant for the family of very affine varieties defined by the graph polynomial ${\cal G}_G= {\cal U}_G + {\cal F}_G$ on a specific parameter subspace ${\cal E}$. This computation is out of reach for large diagrams. The \emph{principal Landau determinant}, as defined in Sec.~\ref{subsec:definition}, provides a state-of-the-art approach to computing many irreducible components of these Euler discriminants from physics.

\begin{proof}[Proof of Thm.~\ref{thm:closedstrata}]
   The signed Euler characteristic of $X_z$ is the number of solutions $\alpha \in X_z$ to the critical point equations for $\log (f(\alpha;z)^\mu \alpha^\nu)$:
   \begin{equation} \label{eq:critpteqs}
   \sum_{i=1}^\ell  \frac{\mu_i \, \frac{\partial f_i(\alpha;z)}{\partial \alpha_j}}{f_i} + \frac{\nu_j}{\alpha_j} \, = \, 0, \quad j = 1, \ldots, n, 
   \end{equation}
   for generic values of the parameters $\mu,\nu$ \cite[Thm.~1]{huh2013maximum}. Moreover, for such generic parameters, all $|\chi(X_z)|$ solutions are regular. This means that the $n\times n$ Jacobian matrix of \eqref{eq:critpteqs} evaluated at each of the solutions has rank $n$. 
   
   In order to apply a powerful theorem from algebraic geometry, called the \emph{generalized parameter continuation theorem} in \cite[Thm. 7.1.4]{sommese2005numerical}, we reformulate \eqref{eq:critpteqs} as a system of equations on a product of projective spaces. We do this by clearing denominators. We also add the new equation $\alpha_1 \cdots \alpha_n f_1\cdots f_\ell \alpha_{n+1} - 1 = 0$, with new variable $\alpha_{n+1}$, so as to impose that $\alpha_1 \cdots \alpha_n f_1\cdots f_\ell \neq 0$. The result is 
   \begin{equation} \label{eq:critpteqsnodenom}\alpha_j \sum_{i=1}^\ell \mu_i \frac{\partial f_i(\alpha;z)}{\partial \alpha_j} \prod_{q \neq i} f_q + \nu_j \prod_{i = 1}^\ell f_i = 0, \, \, j = 1, \ldots, n, \quad \prod_{j=1}^{n+1}\alpha_j \prod_{i=1}^\ell f_i - 1 = 0. \end{equation}
   Finally, we homogenize with respect to the $\alpha$- and $z$-variables. We obtain $n + 1$ equations on $\mathbb{P}^{n+1}$ ($\alpha$-coordinates) with parameters in $\mathbb{P}^{n+\ell-1} \times \mathbb{P}^s$ ($(\mu,\nu)$-coordinates and $z$-coordinates respectively). We now apply \cite[Thm. 7.1.4]{sommese2005numerical}. This theorem requires quite some notation. For transparency, we list its key players here in the notation of \cite{sommese2005numerical}, highlighted in blue, together with their values in our context:
   \begin{itemize}
       \item $\textcolor{blue}{X} = \mathbb{P}^{n+1}$ ($\alpha$-space) and $\textcolor{blue}{Y} = \mathbb{P}^{n+\ell-1} \times \mathbb{P}^s$ ($(\mu,\nu,z)$-space),
       \item $\textcolor{blue}{U \subset X}$ is $\mathbb{C}^{n+1} \subset \mathbb{P}^{n+1}$ with coordinates $\alpha_1, \ldots, \alpha_{n+1}$,
       \item $\textcolor{blue}{Q \subset Y}$ is $\mathbb{P}^{n+\ell-1} \times \overline{\cal E}$, where $\overline{\cal E}$ is the closure of ${\cal E} \subset \mathbb{C}^s$ in $\mathbb{P}^s$,
       \item $\textcolor{blue}{F}$ consists of the $n+1$ equations in \eqref{eq:critpteqsnodenom}.    
   \end{itemize}
    By \cite[Thm.~7.1.4]{sommese2005numerical}, the maximal number of nonsingular solutions $\alpha$ in $\mathbb{C}^{n+1}$ is attained for generic parameters $(\mu,\nu,z) \in \mathbb{P}^{n+\ell-1} \times \overline{\cal E}$. Hence, it is the same as the maximal number of nonsingular solutions for $z$-parameters in the dense open subset ${\cal E} \subset \overline{\cal E}$. By \cite[Thm.~1]{huh2013maximum}, the maximal number of nonsingular solutions with parameters in $\mathbb{P}^{n+\ell-1} \times \{z' \}$ is $|\chi(X_{z'})|$ (here we use $\textcolor{blue}{Q} = \mathbb{P}^{n+\ell-1} \times \{z' \}$ in the parameter continuation theorem). Letting $z'$ run over ${\cal E}$, we see that for almost all parameters $(\mu,\nu,z)$, there are $\chi^* = \max_{z' \in {\cal E}}|\chi(X_{z'})|$ nonsingular solutions. Let $\tilde{\cal U} \subset \mathbb{P}^{n+\ell-1} \times \overline{\cal E}$ be a nonempty Zariski open subset on which this number is attained. 
    There is a Zariski open set ${\cal U} \subset \overline{\cal E}$ such that $\tilde{\cal U} \cap (\mathbb{P}^{n+\ell-1} \times \{ z\} )$ is nonempty, and thus dense in $\mathbb{P}^{n+\ell-1} \times \{ z\}$, for all $z \in {\cal U}$. Therefore, $\chi^*$ is the generic value of $|\chi(X_z)|$ on $\overline{\cal E}$, and this value can only drop on $\overline{\cal E} \setminus {\cal U}$. The same statement is true on the dense open subset ${\cal E} \subset \overline{\cal E}$.
    
    By definition, $Z_k({\cal E}) = \bigcup_{q=0}^k \V_q({\cal E})$. This easily implies $Z_k({\cal E}) \subseteq \bigcup_{q = 0}^k \overline{\V_q({\cal E})}$, where $\overline{\V_q({\cal E})}$ is the closure in $\mathbb{C}^s$. To prove the theorem, it suffices to show that this inclusion is in fact an equality. To show the reverse inclusion, let $\overline{\V_q ({\cal E})} = W_1^q \cup \cdots \cup W_t^q$ be an irreducible decomposition. We repeat the reasoning above, replacing $\overline{\cal E}$ with the closure $\overline{W_i^q}$ in $\mathbb{P}^s$. We find that on each of the irreducible varieties $W_i^q$, the generic signed Euler characteristic is $q$, and this is the maximal value attained on $W_i^q$. This shows $W_i^q \subset Z_k$, for all $i$ and $q \leq k$, and the theorem is proved.    
\end{proof}

\subsection{Definition of the principal Landau determinant} \label{subsec:definition}

As announced above, the principal Landau determinant is meant to be a more-tractable-to-compute replacement of the Euler discriminant, tailored to the case where $X_z$ comes from a Feynman integral. For instance, in Lee-Pomeransky representation, the parameters could be $n = \E, \ell = 1$, and $f_1 = {\cal G}_G$ is the graph polynomial as defined above. A priori, the parameter space ${\cal E} \subset \mathbb{C}^s$ is the kinematic space ${\cal K}$, but it often makes sense to restrict our analysis to smaller subregions, as we did in Sec.~\ref{sec:GKZvsFeyn}. For instance, we might want to implement the preknowledge of some masses being zero, or equal to each other. For this reason, we will work with a subspace ${\cal E} \subset {\cal K}$. For simplicity, we will take it to be a linear subspace. For instance, in \eqref{eq:zsub} we may choose ${\cal E}$ defined by the conditions $\m_1 = \m_2 = \m_3$. The principal Landau determinant will be defined as a nonzero element of the ring $\mathbb{C}[{\cal E}]$ of polynomial functions on ${\cal E}$. In the case of \eqref{eq:zsub}, this is $\mathbb{C}[{\cal E}] = \mathbb{C}[\m_1, \m_2,\m_3, s]$ when ${\cal E} = {\cal K}$, or $\mathbb{C}[{\cal E}] = \mathbb{C}[\m, s]$ when ${\cal E} = {\cal K} \cap \{\m_1 = \m_2 = \m_3 = \m\}$.

Let $G$ be a Feynman diagram with graph polynomial ${\cal G}_G = {\cal U}_G + {\cal F}_G$. The matrix $A$ has $n = \E$ rows, where $\E$ equals the number of internal edges of $G$. Its columns are all exponents occurring in ${\cal G}_G$ \emph{for generic choices of kinematics in ${\cal E}$}. Clearly, this depends on ${\cal E}$. For each face $Q$ of ${\rm Conv}(A)$, we let ${\cal G}_{G,Q}$ be the polynomial obtained by summing only the terms of ${\cal G}_G$ whose exponents lie in $Q$.

For each face $Q$ of ${\rm Conv}(A)$, we consider the \emph{incidence variety}
\begin{equation}\label{eq:incidence-variety}
Y_{G,Q}({\cal E}) \, = \, \Big\{ (\alpha, z) \in (\mathbb{C}^*)^{\E} \times {\cal E} \, : \, {\cal G}_{G,Q}(\alpha;z)  =  \partial_{\alpha} \, {\cal G}_{G,Q}(\alpha;z) \, = \, 0 \Big\}.
\end{equation}
For later convenience, we break this variety up into its irreducible components: 
\begin{equation} \label{eq:irreddecomp} Y_{G,Q}({\cal E}) \, = \, \bigcup_{i \in \mathbb{I}(G,Q)} Y_{G,Q}^{(i)}({\cal E}). \end{equation}
Here $\mathbb{I}(G,Q)$ is some finite indexing set, and $Y_{G,Q}^{(i)}({\cal E})$ are distinct, irreducible varieties. Each of these has a natural projection \[\nabla_{G,Q}^{(i), \circ}({\cal E}) \,=\, \pi_{\cal E}(Y_{G,Q}^{(i)}({\cal E})) \,\subset\, {\cal E},\]
obtained by dropping the $\alpha$-coordinates. This is reminiscent of the open $A$-discriminants $\nabla_A^\circ$ we saw in Sec.~\ref{sec:Adet}. The Zariski closure of 
$\nabla_{G,Q}^{(i), \circ}({\cal E})$ in ${\cal E}$ is $\nabla_{G,Q}^{(i)}({\cal E})$.

To each $i \in \mathbb{I}(G,Q)$, we associate the codimension of this projection:
\[ {\rm codim}(i) \, = \,  {\rm dim}({\cal E}) - {\rm dim}(\nabla_{G,Q}^{(i)}({\cal E})). \] 
We set $\mathbb{I}(G,Q)_1 = \{ i \in \mathbb{I}(G,Q) \, : \, {\rm codim}(i) = 1 \}.$
All varieties $\nabla_{G,Q}^{(i)}({\cal E})$ with $i \in \mathbb{I}(G,Q)_1$ are defined by a single equation:

\[ \nabla_{G,Q}^{(i)}({\cal E}) = \{ \Delta_{G,Q}^{(i)}({\cal E}) = 0 \}, \quad \text{with} \quad \Delta_{G,Q}^{(i)}({\cal E}) \in \mathbb{C}[{\cal E}] \setminus \{ 0 \}.\]

\begin{definition}[Principal Landau determinant] \label{def:PLD}
The \emph{principal Landau determinant} (PLD) associated with the Feynman diagram $G$ and the parameter space ${\cal E}$ is the unique (up to scale) square-free polynomial $E_G({\cal E}) \in \mathbb{C}[{\cal E}]$ defining the \emph{PLD variety} 
\[  {\rm PLD}_G({\cal E}) \, = \, \{ E_G({\cal E}) = 0 \} \, =  \left \{ \, \prod_{Q \in F(A)} \, \prod_{i \in \mathbb{I}(G,Q)_1} \Delta_{G,Q}^{(i)}({\cal E}) = 0 \right \} . \]
\end{definition}
Notice that, since we are primarily interested in the vanishing locus $\{ E_G({\cal E}) = 0\}$ of the principal Landau determinant, our definition takes out the factors $\Delta_{G,Q}^{(i)}({\cal E})$ that appear more than once. 

We end the section with a comment and conjecture on the relation between ${\rm PLD}_G({\cal E})$ and the Euler discriminant variety $\nabla_\chi({\cal E})$. First, in light of Ex. \ref{ex:chivsPAD}, notice that when ${\cal E} = \mathbb{C}^s$ is the entire GKZ parameter space, ${\rm PLD}_G({\cal E}) = \nabla_\chi({\cal E})$ is the variety of the principal $A$-determinant. The proof of \cite[Thm.~13]{AMENDOLA2019222} shows that solutions to the face equations ${\cal G}_{G,Q} = \partial_\alpha {\cal G}_{G,Q} = 0$ correspond to critical points of \eqref{eq:critpteqs} that lie on the boundary of a toric compactification of $X_z$ for generic parameters $\mu, \nu$. Our definition of ${\rm PLD}_G({\cal E})$ aims to detect when there are more such critical points on the boundary then usual, which by \cite[Thm.~1]{huh2013maximum} means that the Euler characteristic is smaller than usual. We could not prove this intuition, but we formalize it with the following conjecture. 
\begin{conjecture} \label{conj:pldvschi}
    For any Feynman diagram $G$ and any linear subspace ${\cal E} \subseteq {\cal K}$, we have ${\rm PLD}_G({\cal E}) \subseteq \nabla_\chi({\cal E})$, where $\nabla_\chi({\cal E})$ is the Euler discriminant for the family of very affine varieties given by $X_z = (\mathbb{C}^*)^\E \setminus \{{\cal G}_G(\alpha;z) = 0 \}, z \in {\cal E}$. 
\end{conjecture}
We have verified Conj.~\ref{conj:pldvschi} numerically in all our examples, and Ex.~\ref{ex:embcomp} below shows that the opposite inclusion ${\rm PLD}_G({\cal E}) \supseteq \nabla_\chi({\cal E})$ may fail. 

\subsection{Examples}

Here we provide four examples illustrating the computation of the principal Landau determinant. We start with the running example of the banana diagram, first with generic masses in Ex.~\ref{ex:banana-generic} and then with one zero mass in Ex.~\ref{ex:banana-nongeneric}. In
Ex.~\ref{ex:singular}, we compare our definition of ${\rm PLD}_G({\cal E})$ with recently proposed alternatives. In particular, in \cite{Klausen:2021yrt} the singular locus of integrals with non-generic parameters ${\cal E}$ is studied by first computing the principal A-determinant $E_A$, then restricting each factor in \eqref{eq:EA} to ${\cal E}$, and finally discarding all discriminants $\Delta_{A \cap \Gamma}$ which vanish after restriction. Ex.~\ref{ex:singular} illustrates how this can fail.
The paper \cite{Dlapa:2023cvx} proposes a more sophisticated approach based on Taylor expansions.
Finally, Ex.~\ref{ex:embcomp} shows that the opposite inclusion in Conj.~\ref{conj:pldvschi} may fail to hold. 

\begin{example}\label{ex:banana-generic}
Consider the example $G = \mathsf{B}_3$ from Sec.~\ref{sec:GKZvsFeyn}. At first, let us pick ${\cal E} = {\cal K}$ to be the entire kinematic space parametrized by $(\m_1,\m_2,\m_3,s)$, as in \eqref{eq:zsub}. On the codimension-$1$ face $Q = 2389$, the initial form ${\cal G}_{\mathsf{B}_3, 2389}$ is given by
\be
{\cal G}_{\mathsf{B}_3, 2389} \,=\, (1 - \m_3 \alpha_3)(\alpha_1 + \alpha_2)\alpha_3.
\ee
The corresponding incidence variety $Y_{\mathsf{B}_3, 2389}({\cal K})$ is carved out by the system of equations
\begin{align}
{\cal G}_{\mathsf{B}_3, 2389} &\,=\, (1 - \m_3 \alpha_3)(\alpha_1 + \alpha_2)\alpha_3 \,=\, 0 ,\\
\partial_{\alpha_1} {\cal G}_{\mathsf{B}_3, 2389} \,=\, \partial_{\alpha_2} {\cal G}_{\mathsf{B}_3, 2389} &\,=\, (1 - \m_3 \alpha_3)\alpha_3 \,=\, 0,\\
\partial_{\alpha_3} {\cal G}_{\mathsf{B}_3, 2389} &\,=\, (1 - 2\m_3 \alpha_3)(\alpha_1 + \alpha_2) \,=\, 0\,
\end{align}
on the total space $(\mathbb{C}^\ast)^3 \times {\cal K}$. It has a single $5$-dimensional component
\be
Y_{\mathsf{B}_3, 2389}^{(1)}({\cal K}) \,=\, \{ (\alpha,z) \in (\mathbb{C}^\ast)^3 \times {\cal K} : \alpha_1 + \alpha_2 = 1- \m_3 \alpha_3 = 0 \}.
\ee
Its projection
$\nabla_{G,\Gamma}^{(i)}({\cal K})$ has codimension $0$ (there is a solution $\alpha \in (\mathbb{C}^\ast)^3$ for almost any $s$, $\m_1$, $\m_2$, $\m_3$) and hence does not contribute to the principal Landau determinant.

On the other hand, the codimension-$2$ face $\Gamma = 89$ contained in $2389$ gives
\be
{\cal G}_{\mathsf{B}_3, 89} \,=\, - \m_3 (\alpha_1 + \alpha_2) \alpha_3^2, 
\ee
and the incidence variety $Y_{\mathsf{B}_3, 89}({\cal K})$ is defined by the equations
\begin{align}
{\cal G}_{\mathsf{B}_3, 89} &\,=\, - \m_3 (\alpha_1 + \alpha_2) \alpha_3^2 \,=\, 0,\\
\partial_{\alpha_1} {\cal G}_{\mathsf{B}_3, 89} \,=\, \partial_{\alpha_2} {\cal G}_{\mathsf{B}_3, 89} &\,=\, - \m_3 \alpha_3^2 \,=\, 0,\\
\partial_{\alpha_3} {\cal G}_{\mathsf{B}_3, 89} &\,=\, - 2\m_3 (\alpha_1 + \alpha_2) \alpha_3 \,=\,0.
\end{align}
It has a single $6$-dimensional component
\be
Y_{\mathsf{B}_3, 89}^{(1)}({\cal K}) \,=\, \{ (\alpha,z) \in (\mathbb{C}^\ast)^3 \times {\cal K} : \m_3 = 0 \},
\ee
whose projection $\pi_{\cal E}(Y_{\mathsf{B}_3, 89}^{(1)}({\cal K}))$ has codimension $1$ and contributes $\Delta_{\mathsf{B}_3, 89}^{(1)}({\cal K}) = \m_3$ to the principal Landau determinant.
Physically, it is a mass divergence, which causes a possible singularity of the integral \eqref{eq:banana} on the subspace \eqref{eq:zsub} if $\m_3=0$.

Repeating analogous analysis for all faces of $\Newt({\cal G}_{\mathsf{B}_3})$ gives the principal Landau determinant:
\begin{align}
E_{\mathsf{B}_3}({\cal K}) \,=\, \m_1 \m_2 \m_3 s \Big[&s^4 - 4 s^3 (\m_1 + \m_2 + \m_3)\\
&+ 2s^2 (3\m^2_1 + 3\m^2_2 + 3\m^2_3 + 2\m_1 \m_2 + 2\m_2 \m_3 + 2\m_3 \m_1) \\
&-4 s (\m_1^3 + \m_2^3 + \m_3^3 - \m_1 \m_2 (\m_1 + \m_2) - \m_2 \m_3 (\m_2 + \m_3) \\
&\qquad - \m_3 \m_1 (\m_3 + \m_1) + 10 \m_1 \m_2 \m_3) + \lambda(\m_1, \m_2, \m_3)^2 \Big],
\end{align}
where
\be\label{eq:Kallen}
\lambda(a,b,c) := a^2 + b^2 + c^2 - 2ab - 2bc - 2ac
\ee
is the K\"all\'en function.
The term in the square brackets comes from the facet $456789,10$ dual to the ray $(-1,-1,-1)$. Once expressed in terms of the particle masses $m_e$ (such that $\m_e = m_e^2$), it factors into four components 
\be
[ s - (m_1 + m_2 + m_3)^2]
[ s - (m_1 + m_2 - m_3)^2]
[ s - (m_1 - m_2 + m_3)^2]
[ s - (m_1 - m_2 - m_3)^2].
\ee
This is a well-known result for the singular locus of the Feynman integral $\I_{\mathsf{B}_3}$. 
\end{example}

\begin{example}\label{ex:banana-nongeneric}
Consider the same example, but on the subspace ${\cal E} \subset {\cal K}$ obtained by setting $\m_1 = 0$. Note that ${\cal E} \subset \{ E_{\mathsf{B}_3}({\cal K}) = 0\}$, which means we cannot reuse the results of the previous example directly. Instead, the polytope $\Newt({\cal G}_{{\tt B}_3}({\cal E}))$ is smaller and has $f$-vector $(7,11,6)$.
The principal Landau determinant on ${\cal E}$ is
\be
E_{\mathsf{B}_3}({\cal E}) \,=\, \m_2 \m_3 s\, \lambda(s, \m_1, \m_2),
\ee
where once again, the final factor factors in terms of $m_e$'s. The $\lambda$ function contribution comes from the face with the weight $(-2,-1,-1)$.

We could have tested the inclusion ${\cal E} \subset \{ E_{\mathsf{B}_3}({\cal K}) = 0\}$ by computing Euler characteristics. One of the simplest ways is to get them is by counting the number of critical points of the log-likelihood function $W = \mu_1 \log f + \sum_{i=1}^{3} \nu_i \log \alpha_i$, see \cite[Sec.~5]{Mizera:2021icv} for details. In practice, this check can be performed by using the following self-contained snippet in \texttt{Julia}, which we first run on the kinematic space ${\cal K}$:
\begin{minted}{julia}
using HomotopyContinuation

@var α[1:3], s, m[1:3], u[1:4]
f = (1 - m[1]*α[1] - m[2]*α[2] - m[3]*α[3])*
    (α[1]*α[2] + α[2]*α[3] + α[3]*α[1]) + s*α[1]*α[2]*α[3]

W = u[1] * log(f) + dot(u[2:4], log.(α))
dW = System(differentiate(W, α), parameters = [s; m; u])

Crit = monodromy_solve(dW)
crt = certify(dW, Crit)
println(ndistinct_certified(crt))
\end{minted}
After loading packages, the lines $4{-}6$ define the variables of the problem and the polynomial $f$. The system of equations is set up in the lines $8{-}9$ and solved with one command in line $11$ using homotopy continuation \cite{10.1007/978-3-319-96418-8_54}, followed by certification \cite{breiding2020certifying} and printing the result. The code returns $7$ in agreement with the result quoted in Rmk.~\ref{rmk:1}. Specializing to the subspace $\cal E$ amounts to inserting the substitution
\begin{minted}[firstnumber=last]{julia}
W = subs(W, m[1]=>0)
\end{minted}
between the lines $8$ and $9$. This changes the result to $4$, indicating that $\cal E$ belongs to the principal Landau determinant hypersurface ${\rm PLD}_{\mathsf{B}_3}({\cal K})$.
\end{example}

\begin{example}\label{ex:singular}
Consider the generalized Euler integral with $\ell=1$, $n=2$, and take
\be
f_1 \,=\, (1+\alpha_1)(a+b \alpha_1+c\alpha_2+d\alpha_1 \alpha_2),
\ee
where $\cal E$ is parametrized by $(a,b,c,d)$.
In analyzing its singularities, one could attempt to apply the definition of the principal $A$-determinant for $f_1$ with generic coefficients first and then specialize to $\cal E$. In this case, the $A$ matrix is
\[A \, = \, \begin{pmatrix}
0 & 1 & 0 & 2 & 1 & 2 \\
0 & 0 & 1 & 0 & 1 & 1 \\
1 & 1 & 1 & 1 & 1 & 1 
\end{pmatrix}.
\]
Direct computation gives the principal A-determinant 
\begin{align}\label{eq:EA-singular}
E_A &\,=\, z_1 z_3 z_4 z_6 (z_2^2 - 4 z_1 z_4)(z_5^2 - 4 z_3 z_6)\\
&\quad\;(z_3^2 z_4^2 - z_2 z_3 z_4 z_5 + z_1 z_4 z_5^2 + 
 z_2^2 z_3 z_6 - 2z_1 z_3 z_4 z_6 - z_1 z_2 z_5 z_6 + 
 z_1^2 z_6^2).
\end{align}
The subspace $\cal E$ giving specialized coefficients we are interested in is
\be
(z_1, z_2, z_3, z_4, z_5, z_6) \,=\, (a, a+b, c, b, c+d, d). 
\ee
On this subspace, the final factor in \eqref{eq:EA-singular} evaluates to zero. A naive approach would be to simply discard it \cite{Klausen:2021yrt}, but keep the rest (other proposals, based on taking limits also exist \cite{Dlapa:2023cvx}):
\be\label{eq:EA-naive}
E_A^{\text{naive}}({\cal E}) \,=\, a b c d (a-b)(c-d). 
\ee
However, this prescription does not correctly account for all singularities of the integral, as one can verify on simple examples. For instance, taking $\mu_1 = -2$, $\nu_1 = \nu_2 = 1$, and $\Gamma = \R_+^2$ one finds that
\begin{align}
I(a,b,c,d) \, = \, \int_{\R_+^2} &\frac{\d \alpha_1\wedge \d \alpha_2}{[(1+\alpha_1)(a + b \alpha_1 + c \alpha_2 + d \alpha_1 \alpha_2)]^2}\\
&\qquad \qquad = \frac{1}{(a-b)(c-d)} - \frac{1}{bc - ad} \left[ \frac{b^2 \log(a/b)}{(a-b)^2}  + \frac{d^2 \log(d/c)}{(c-d)^2} \right]
\end{align}
for $a,b,c,d>0$ and its analytic continuation elsewhere. The integral is singular when $bc = ad$ on most sheets of the logarithms, which was not detected by \eqref{eq:EA-naive}. We note that, in order to interpret this integral as a pairing between a twisted cycle and a twisted cocycle, one needs to replace $\Gamma$ by its \emph{regularization}, see e.g. \cite[Theorem* 3.24]{matsubara2023four}.

Hence, the desired description of the singular locus of $I(a,b,c,d)$ is
\be
E({\cal E}) \,=\, a b c d (a-b) (c-d) (bc - ad) .
\ee
The previously-missing component comes from the dense face $Q$ (the interior of $\conv(A)$). Let us see how it arises. The incidence variety $Y_{Q}({\cal E})$ is defined by 
\begin{align}
f_1 &\,=\, (1+\alpha_1)(a+b \alpha_1+c\alpha_2+d\alpha_1 \alpha_2) \,=\, 0\, ,\\
\partial_{\alpha_1} f_1 &\,=\, a + b + 2b \alpha_1 + (c+d) \alpha_2 + 2 d \alpha_1 \alpha_2 \,=\, 0, \\
\partial_{\alpha_2} f_1 &\,=\, (1 + \alpha_1)(c + d \alpha_1) \,=\, 0 .
\end{align}
The variety corresponding to the radical ideal of the polynomials above has $2$ irreducible components with dimensions $4$ and $3$, respectively. The first one is
\be
Y_Q^{(1)}({\cal E}) \,=\, \Big\{ (\alpha,z) \in (\C^\ast)^2 \times {\cal E}\,:\, \alpha_1 + 1 = a - b + (c-d) \alpha_2 = 0 \Big\}.
\ee
It is a \emph{dominant component}, meaning that $\nabla^{(1)}_Q({\cal E}) = {\cal E}$ (there is a solution for every $(a,b,c,d) \in \mathbb{C}^4$, except for $\{c=d\} \setminus \{a = b\}$, which are included in $\nabla_{Q}^{(1)}({\cal E})$ by closure), and hence is discarded. However, we are not allowed to discard the second component:
\be
Y_Q^{(2)}({\cal E}) \,=\, \Big\{ (\alpha,z) \in (\C^\ast)^2 \times {\cal E}: bc-ad = a + b \alpha_1 = c + d\alpha_1 = a + c\alpha_2  = b + d\alpha_2 = 0 \Big\}.
\ee
Its projection $\pi_{\cal E}(Y_{Q}^{(2)}({\cal E}))$ has codimension $1$ and gives the discriminant surface $\{bc=ad\}$ missed in the naive approach.

This example illustrates why throwing away faces with dominant components from the principal A-determinant might lead to incorrect results, in addition to complicating the computation in intermediate stages, and a more specialized definition of the principal Landau determinant is necessary. Dominant components appear in nearly all examples of Feynman integrals studied in this paper, since they are tied to UV/IR divergences.
\end{example}

\begin{example} \label{ex:embcomp}
    This example serves to illustrate why we cannot have the opposite inclusion in Conj.~\ref{conj:pldvschi}. That is, we might have $\nabla_\chi({\cal E}) \subsetneq {\rm PLD}_G({\cal E})$. The simplest diagram for which we observed this is the parachute diagram $G = \texttt{par}$. To understand the failure of this inclusion from a mathematical perspective, it is instructive to consider a smaller example first. We set $\ell = 1$ and consider the very affine surface $X_z \subset (\mathbb{C}^*)^2$ defined by
    \[ f(\alpha_1,\alpha_2;z) \, = \, (\alpha_2-1)^2 - (\alpha_1-z)\alpha_1^2. \]
    That is, $X_z$ is the complement of a nodal cubic $V_z = \{f(\alpha_1, \alpha_2;z) = 0 \} \subset (\mathbb{C}^*)^2$ in the two-dimensional complex torus. Our parameter space is ${\cal E} = \mathbb{C}$. We claim that the generic Euler characteristic $\chi^*$ is equal to 4. To see this, we define the following set: 
    \[ P = \{ (\alpha_1,\alpha_2) \in V_z \, : \, f(\alpha_1,0;z) = 0   \text{ or } \alpha_1 = z \}. \]
    One checks that $P$ consists of four points for generic $z$. These points are indicated in blue in Fig.~\ref{fig:nodalcubic}. We obtain a fibration $V_z \setminus P \rightarrow \mathbb{C}^* \setminus \{4 \text{ points} \}$ by simply forgetting the $\alpha_2$-coordinate. Fibres of this map consist of two points. Using the excision and fibration properties of the Euler characteristic, we obtain 
    \[ \chi(X_z) \, = \, \chi((\mathbb{C}^*)^2) - \chi(V_z) \, = \, 0 - [\chi(V_z \setminus P) + \chi(P)] \, = \, 0 - [ -4 \cdot 2 + 4] \, = \, 4. \]
    When $z = 0$, $V_0$ is a cuspidal cubic: the cusp is at $(\alpha_1, \alpha_2) = (0,1)$. A similar reasoning, using a fibration $V_0 \setminus P \rightarrow \mathbb{C}^* \setminus \{ 3 \text{ points} \}$, gives $\chi(X_0) = 3$. The dots with a dashed border on the right part of Fig. \ref{fig:nodalcubic} are not real. We have shown that $0 \in \nabla_\chi({\cal E})$. 
    \begin{figure}
        \centering        \includegraphics[height = 4.5cm]{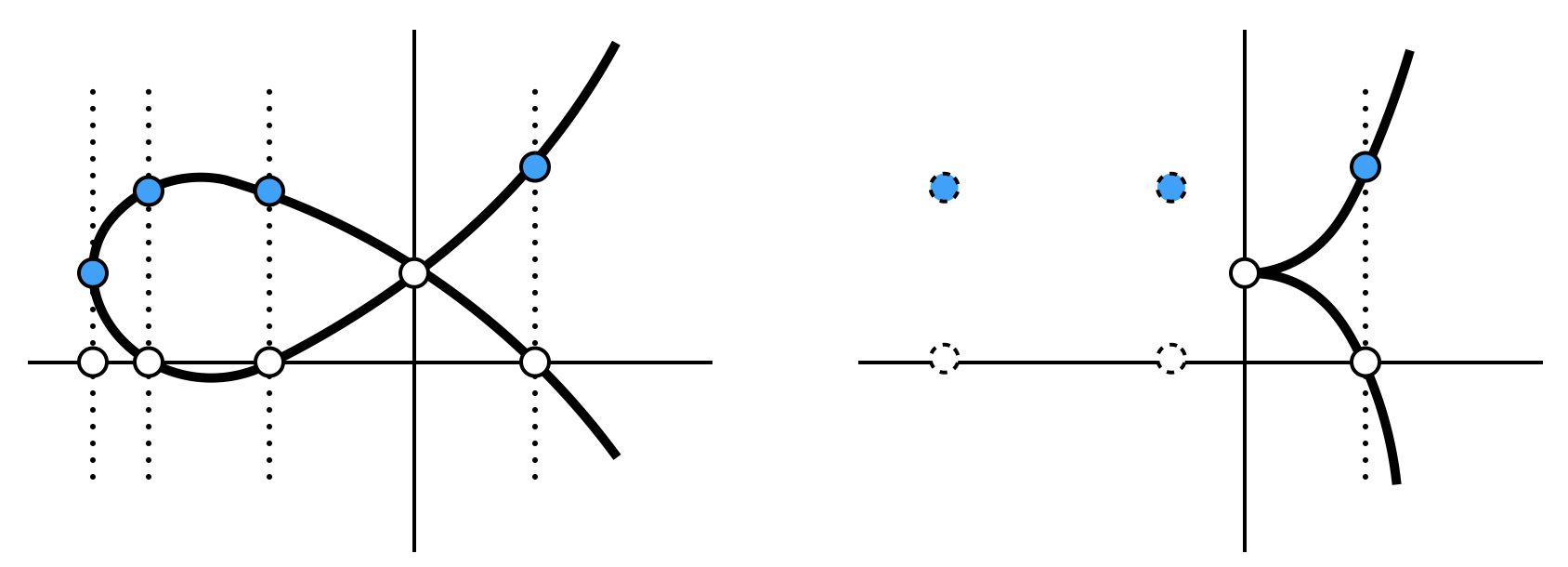}
        \caption{The very affine surface $X_z$ is the complement of a nodal cubic for generic $z$ (left). For $z = 0$, it is the complement of a cuspidal cubic (right).}
        \label{fig:nodalcubic}
    \end{figure}
    To make sense of the principal Landau determinant in this example, we think of $f$ as the graph polynomial ${\cal G}_G$ of a fictional diagram $G$. The polytope ${\rm Conv}(A)$ is the triangle with vertices $(0,0), (3,0), (0,2)$. To compute ${\rm PLD}_G({\cal E})$, we investigate the incidence varieties $Y_{G,Q}({\cal E})$ from \eqref{eq:incidence-variety} for the faces $Q$ of that triangle. One checks easily that, when $Q$ is a vertex, $Y_{G,Q}({\cal E})$ is empty. The same is true for the two-dimensional face $Q = {\rm Conv}(A)$. When $Q$ is one of the edges $(0,0)-(0,2)$, $(0,2)-(3,0)$, the equations for $Y_{G,Q}({\cal E})$ do not depend on $z$, so all its component are either empty or project dominantly to ${\cal E}$. Hence, also these faces do not contribute to ${\rm PLD}_G({\cal E})$. Finally, when $Q$ is the edge $(0,0)-(3,0)$, $Y_{G,Q}({\cal E})$ is defined by $1-(\alpha_1-z)\alpha_1^2 = -3\alpha_1^2+2z \alpha_1 = 0$. Eliminating $\alpha_1$ gives $4z^3+27 = 0$. In particular, we have $0 \notin {\rm PLD}_G({\cal E})$. 

    We have concluded that the PLD analysis does not detect the drop in the Euler characteristic for $z = 0$. Here is an ad hoc remedy for this. When $z = 0$, the nodal singularity at $(\alpha_1,\alpha_2) = (0,1)$ becomes a cusp. We now look at the incidence equations $f = \alpha_1 \frac{\partial f}{\partial \alpha_1} = \alpha_2 \frac{\partial f}{\partial \alpha_2} = 0$ in a partial compactification $\mathbb{C}^2 \supset (\mathbb{C}^*)^2$ which contains the boundary $\{\alpha_1 = 0\}$ containing the cusp. The ideal generated by these three equations in the ring $\mathbb{C}[\alpha_1,\alpha_2,z]$ has three primary components: 
    \begin{equation} \label{eq:PDcubic} \left \langle f, \alpha_1 \frac{\partial f}{\partial \alpha_1}, \alpha_2 \frac{\partial f}{\partial \alpha_2} \right \rangle \, = \, \langle \alpha_1,\alpha_2-1\rangle \cap \langle \alpha_2,3\alpha_1-2z,4z^3+27 \rangle \cap \langle z,\alpha_2-1,\alpha_1^3 \rangle. \end{equation}
    The first component projects dominantly to the parameter space ${\cal E}$, reflecting the fact that there is always (at least) a nodal singularity at $(0,1)$. The second component projects to the PLD given by $4z^3+27 = 0$. Finally, and most importantly, the third component projects to $z=0$. It is an embedded point supported on the first primary component $\langle \alpha_1, \alpha_2-1 \rangle$, indicating that for $z = 0$, the nodal singularity gets \emph{extra} singular. The same happens for $G = {\tt par}$, see Sec.~\ref{subsec:beyondstandard}. App.~\ref{sec:appendix2} presents a toric compactification which can be used systematically to detect such components. 
    \end{example}

\subsection{Different formulations}
Our Def. \ref{def:PLD} is based on a specialized GKZ analysis of the Feynman integral $I_G$. That integral is viewed as a generalized Euler integral of the type \eqref{eq:Eulerintegral}, with $n = \E$ the number of internal edges, $\ell = 1$ and $f = f_1 = {\cal U}_G + {\cal F}_G$. A different integral formula for $I_G$, called \emph{Feynman representation} establishes it as \eqref{eq:Eulerintegral} with $n = \E-1$, $\ell = 2$ and $f = (f_1,f_2) = (\overline{{\cal U}_G}, \overline{{\cal F}_G})$. Here $\overline{f}$ is the \emph{dehomogenization} of $f$, where we set $\alpha_\E = 1$. At the level of integrals, we have the following well-known result.
\begin{proposition}
Feynman integrals $I_G$ without numerators can be equivalently represented as integrals
\be
I_G = \int_{\R_+^{\E-1}} \frac{\d^{\E-1}\overline{\alpha}}{ \overline{\U_G}^{-\mu_1}\, \overline{\F_G}^{-\mu_2}}\prod_{i=1}^{\E-1} \overline{\alpha_i}^{\nu_i - 1} = \frac{\Gamma(-\mu_1 - \mu_2)}{\Gamma(-\mu_1)\Gamma(-\mu_2)} \int_{\R_+^{\E}} \frac{\d^{\E}\alpha}{(\U_G + \F_G)^{-\mu_1 - \mu_2}} \prod_{i=1}^{\E} \alpha_i^{\nu_i - 1}
\ee
up to an overall normalization $\Gamma(-\mu_2)$. Here, $\nu_i$ are the powers of propagators associated to every internal edge and $\mu_1 = -\mu_2 - \D/2$, $\mu_2 = \L \D/2 - \sum_{i=1}^{\E} \nu_i$, where $\D$ is the space-time dimension and $\L$ is the number of loops in the diagram.
\end{proposition}
\begin{proof}
The equality can be easily seen as an application of the identity
\be
\overline{\U_G}^{\mu_1} \overline{\F_G}^{\mu_2} = \frac{\Gamma(-\mu_1 -\mu_2)}{\Gamma(-\mu_1) \Gamma(-\mu_2)} \int_{\R_+} \d y\,   y^{-\mu_2 - 1} (\overline{\U_G} + y \overline{\F_G})^{\mu_1 + \mu_2}\, ,
\ee
followed by the change of variables
\be
y = \alpha_\E, \qquad \overline{\alpha_i} = \frac{\alpha_i}{\alpha_\E} \quad\text{for}\quad i=1,2,\ldots,\E-1\, .
\ee
To identify the factor $\alpha_\E^{\nu_\E-1}$, one needs to use $-\mu_2 - \L(\mu_1 + \mu_2) - \sum_{i=1}^{\E-1}\nu_i = \nu_\E$.
\end{proof}
In App.~\ref{sec:appendix}, we review how to also include numerator factors, which does not change any of the conclusions of our analysis.

The $\E$-dimensional very affine variety $(\mathbb{C}^*)^\E \setminus V_{(\mathbb{C}^*)^\E}({\cal G}_G)$ (see \eqref{eq:Xz}) is replaced with $(\mathbb{C}^*)^{\E-1} \setminus V_{(\mathbb{C}^*)^{\E-1}}(\overline{{\cal U}_G}\overline{{\cal F}_G})$. 
In this new setup, the GKZ paradigm would lead us to consider the principal $A$-determinant, where $A$ is the matrix \eqref{eq:cayley} with $\E+1$ rows formed from the exponents $A_1$ of $\overline{{\cal U}_G}$ and $A_2$ of $\overline{{\cal F}_G}$. In Sec.~\ref{subsec:definition}, one would replace the graph polynomial ${\cal G}_G$ by $y_1 \cdot \overline{{\cal U}_G} + y_2 \cdot \overline{{\cal F}_G}$, where $y_1$ and $y_2$ are new variables corresponding to the last two rows of $A$. We obtain $\E + 2$ critical point equations 
\[ y_1 \cdot \overline{{\cal U}_G} + y_2 \cdot \overline{{\cal F}_G} \, =  \, \partial_{\alpha_1, \ldots, \alpha_{\E-1},y_1,y_2} (y_1 \cdot \overline{{\cal U}_G} + y_2 \cdot \overline{{\cal F}_G}) \, = \,  0\]
on the torus $(\mathbb{C}^*)^{\E+1}$ with coordinates $\alpha_1, \ldots, \alpha_{\E-1}, y_1,y_2$. These are homogeneous in $y_1, y_2$, so we may dehomogenize and set $y_1 = 1, y_2 = y$. 
This section explains why this different approach would lead to the same definition. Here is the key observation.
\begin{lemma} \label{lem:varphi}
    Let ${\cal U}, {\cal F} \in \mathbb{C}[\alpha_1, \ldots, \alpha_\E]$ be homogeneous polynomials of degree $\L$ and $\L+1$ respectively. Let $\overline{\cal U}, \overline{\cal F}$ be their dehomogenizations after setting $\alpha_\E = 1$. The map
    \begin{equation} \label{eq:isomtori}
    \varphi: \, (\mathbb{C}^*)^\E \rightarrow (\mathbb{C}^*)^\E : \quad (\alpha_1, \ldots, \alpha_{\E-1}, y) \longmapsto (y\alpha_1, \ldots, y \alpha_{\E-1}, y) 
    \end{equation}
    is an isomorphism of tori which sends the hypersurface $\{ \overline{\cal U} + y \cdot \overline{\cal F} = 0 \}$ to $\{{\cal U} + {\cal F} = 0 \}$.
\end{lemma}
\begin{proof}
    The map $\varphi$ is an isomorphism because the matrix of exponents 
    \begin{equation} \label{eq:Q}
    {\cal Q} \, = \, \begin{pmatrix}
        1 & 0 & 0 & \cdots & 0 & 0 \\
        0 & 1 & 0 & \cdots & 0 & 0 \\
        0 & 0 & 1 & \cdots & 0 & 0 \\
        \vdots & \vdots & \vdots & \ddots & 1 & 0 \\
        1 & 1 & 1 & \cdots & 1 & 1
        \end{pmatrix}
    \end{equation}
    has determinant $1$. The second claim follows from the fact that the pullback $\varphi^*({\cal U } + {\cal F})$ of ${\cal U} + {\cal F}$ under the map $\varphi$ equals $y^\L \cdot (\overline{\cal U} + y \cdot \overline{\cal F})$.
\end{proof}
    Notice that, in \eqref{eq:isomtori}, the first torus $(\mathbb{C}^*)^\E$ has coordinates $\alpha_1, \cdots, \alpha_{\E-1}, y$, and the second has coordinates $\alpha_1, \ldots, \alpha_\E$. To avoid confusion, we will denote these tori by $T_1 \simeq T_2 \simeq (\mathbb{C}^*)^\E$, and $\varphi: T_1 \rightarrow T_2$. 

    The linear map ${\cal Q}: \mathbb{R}^\E \rightarrow \mathbb{R}^\E$ with ${\cal Q}$ as in \eqref{eq:Q} is an isomorphism which maps $\mathbb{Z}^\E \subset \mathbb{R}^\E$ bijectively onto itself. The Newton polytope $P_2 \subset \mathbb{R}^\E$ of ${\cal U} + {\cal F}$ is mapped to the Newton polytope $P_1 \subset \mathbb{R}^\E$ of $\overline{\cal U} + y \cdot \overline{\cal F}$, and ${\cal Q}$ defines a bijection between faces of $P_2$ and faces of $P_1$. For a face $Q_2 \subset P_2$, we write $Q_1 = {\cal Q}(Q_2)$ for the corresponding face of $P_1$. Note that $P_2$ is the polytope ${\rm Conv}(A)$ in Sec.~\ref{subsec:definition}.

    Let ${\cal H}_G = \overline{{\cal U}_G} + y \cdot \overline{{\cal F}_G}$. Similar to what we did for ${\cal G}_G$ above, for each face $Q_1 \subset P_1$ we let ${\cal H}_{G,Q_1}$ be the sum of the terms of ${\cal H}_G$ whose exponents lie in $Q_1$. For each face $Q_2$ of $P_2$, we consider the incidence varieties 
      \begin{small}
    \begin{align*}
Y_{G,Q_2}({\cal E}) \, &= \, \left \{ (\alpha, z) \in T_2 \times {\cal E} \, : \, {\cal G}_{G,Q_2}(\alpha;z) \,  = \, \partial_\alpha \, {\cal G}_{G,Q_2}(\alpha;z) \, = \, 0 \right \}, \\
Y_{G,Q_1}({\cal E}) \, &= \, \left \{ (\bar{\alpha}, z) \in T_1 \times {\cal E} \, : \, {\cal H}_{G,Q_1}(\bar{\alpha};z)  \, = \,  
\partial_{\bar{\alpha}} \, {\cal H}_{G,Q_1}(\bar{\alpha};z)  \, = \, 0 \right \}.
\end{align*}
\end{small}
Here $Y_{G,Q_2}({\cal E})$ is identical to the incidence variety seen in \eqref{eq:incidence-variety}, and the derivatives appearing in the definition of $Y_{G,Q_1}({\cal E})$ are with respect to $\bar{\alpha} = (\alpha_1, \ldots, \alpha_{\E-1}, y)$. 
\begin{proposition} \label{prop:equiv}
The restriction of the map $\varphi \times {\rm id}: T_1 \times {\cal E} \rightarrow T_2 \times {\cal E}$, with $\varphi$ as in \eqref{eq:isomtori}, to $Y_{G,Q_1}({\cal E})$ is an isomorphism $Y_{G,Q_1}({\cal E}) \rightarrow Y_{G,Q_2}({\cal E})$.
\end{proposition}
\begin{proof}
    This follows from the fact that whether a Laurent polynomial defines a singular hypersurface in a toric compactification does not depend on the choice of coordinates on the lattice. In more down to earth terms, the statement follows by checking that the substitution $\alpha = \varphi(\bar{\alpha}) = (y\alpha_1, \ldots, y\alpha_{\E-1}, y)$ in
    \begin{equation} \label{eq:G}
    {\cal G}_{G,Q_2}(\alpha;z)  \,=\, \frac{\partial {\cal G}_{G,Q_2}}{\partial \alpha_1}(\alpha;z)  \,=\,  \cdots   \,=\,\frac{\partial {\cal G}_{G,Q_2}}{\partial \alpha_\E}(\alpha;z) \, = \, 0
    \end{equation}
    leads to a new set of equations which is equivalent to 
    \begin{equation} \label{eq:H}
    {\cal H}_{G,Q_1}(\bar{\alpha};z)   \,=\,\frac{\partial {\cal H}_{G,Q_1}}{\partial \alpha_1}(\bar{\alpha};z)  \,=\,  \cdots   \,=\, \frac{\partial {\cal H}_{G,Q_1}}{\partial y}(\bar{\alpha};z) \, = \, 0.
    \end{equation}
    A point $(\bar{\alpha},z) \in T_1 \times {\cal E}$ satisfies \eqref{eq:H} if and only if $(\varphi(\bar{\alpha}),z) \in T_2 \times {\cal E}$ satisfies \eqref{eq:G}.
\end{proof}
It follows from Prop. \ref{prop:equiv} that using the Feynman representation would lead to the same definition of the principal Landau determinant, as the projections of our incidence varieties to the parameter space ${\cal E}$ are identical. 

A well-known observation is that passing from Lee--Pomeransky to Feynman representation preserves the Euler characteristic. We prove a more general result.
\begin{theorem}\label{thm:EulerChar}
    Let ${\cal F}$ and ${\cal U}$ be homogeneous polynomials in $\mathbb{C}[\alpha_1, \ldots, \alpha_\E]$ of degree $d_{\cal F}$ and $d_{\cal U}$ respectively, with $d_{\cal F} > d_{\cal U}$. Let $\overline{{\cal F}}, \overline{{\cal U}}$ be the dehomogenizations obtained by setting $\alpha_{\E} = 1$. We have 
    \[ \chi\left ( (\mathbb{C}^*)^{\E} \setminus V_{(\mathbb{C}^*)^{\E}}({\cal F} + {\cal U}) \right ) \,=\, (d_{\cal U} - d_{\cal F}) \cdot \chi \left ((\mathbb{C}^*)^{\E-1} \setminus V_{(\mathbb{C}^*)^{\E-1}}(\overline{{\cal F}}\cdot \overline{{\cal U}}) \right ). \]
\end{theorem}
\begin{proof}
    Let $T = (\mathbb{C}^*)^{\E}$ be the torus with coordinates $\alpha_1, \ldots, \alpha_{\E-1}, y$ and consider the hypersurface complement $T \setminus V_{T}(y \cdot \overline{{\cal F}} + \overline{{\cal U}})$. The first step is to show that this has Euler characteristic $\chi \left ((\mathbb{C}^*)^{\E-1} \setminus V_{(\mathbb{C}^*)^{\E-1}}(\overline{{\cal F}}\cdot \overline{{\cal U}}) \right )$. To this end, we decompose 
    \begin{align*} T \setminus V_{T}(y \cdot \overline{{\cal F}} + \overline{{\cal U}}) \, = & \,  \textcolor{mycolor1}{(T \setminus V_{T}(y \cdot \overline{{\cal F}} + \overline{{\cal U}})) \setminus V_{T}(\overline{{\cal F}}\cdot \overline{{\cal U}})} \\
    &\, \sqcup \, \textcolor{mycolor2}{V_{T}(\overline{{\cal F}}) \setminus V_{T}(\overline{{\cal U}})} \, \sqcup \,  \textcolor{mycolor5}{V_{T}(\overline{{\cal U}}) \setminus V_{T}(\overline{{\cal F}})}. 
    \end{align*}
    Forgetting the $y$-coordinate $(\alpha_1, \ldots, \alpha_{\E-1}, y) \mapsto (\alpha_1, \ldots, \alpha_{\E-1})$ gives three maps 
    \begin{align*}
        \textcolor{mycolor1}{(T \setminus V_{T}(y \cdot \overline{{\cal F}} + \overline{{\cal U}})) \setminus V_{T}(\overline{{\cal F}}\cdot \overline{{\cal U}})} & \, \longrightarrow \, (\mathbb{C}^*)^{\E-1} \setminus V_{(\mathbb{C}^*)^{\E-1}}(\overline{{\cal F}}\cdot \overline{{\cal U}}), \\
        \textcolor{mycolor2}{V_{T}(\overline{{\cal F}}) \setminus V_{T}(\overline{{\cal U}})} & \, \longrightarrow \, V_{(\mathbb{C}^*)^{\E-1}}(\overline{{\cal F}}) \setminus V_{(\mathbb{C}^*)^{\E-1}}(\overline{{\cal U}}), \\
        \textcolor{mycolor5}{V_{T}(\overline{{\cal U}}) \setminus V_{T}(\overline{{\cal F}})} & \, \longrightarrow  \, V_{(\mathbb{C}^*)^{\E-1}}(\overline{{\cal U}}) \setminus V_{(\mathbb{C}^*)^{\E-1}}(\overline{{\cal F}}). 
    \end{align*}
    Each of these maps is a fibration, with fibers isomorphic to $\mathbb{C}^* \setminus \{ y + 1 = 0 \}$, $\mathbb{C}^*$ and $\mathbb{C}^*$ respectively, with Euler characteristics \textcolor{mycolor1}{$-1$}, \textcolor{mycolor2}{$0$} and \textcolor{mycolor5}{$0$}. We conclude that 
    \begin{align*}
        \chi \left (  T \setminus V_{T}(y \cdot \overline{{\cal F}} + \overline{{\cal U}})\right) & = \chi \left ( \textcolor{mycolor1}{(T \setminus V_{T}(y \cdot \overline{{\cal F}} + \overline{{\cal U}})) \setminus V_{T}(\overline{{\cal F}}\cdot \overline{{\cal U}})} \right ) \\
        & ~~ + \chi \left (  \textcolor{mycolor2}{V_{T}(\overline{f}) \setminus V_{T}(\overline{{\cal U}})} \right ) + \chi \left ( \textcolor{mycolor5}{V_{T}(\overline{{\cal U}}) \setminus V_{T}(\overline{{\cal F}})} \right ) \\
        & = \textcolor{mycolor1}{-1} \cdot \chi \left ( (\mathbb{C}^*)^{\E-1} \setminus V_{(\mathbb{C}^*)^{\E-1}}(\overline{{\cal F}}\cdot \overline{{\cal U}}) \right ) \\
        & ~~ + \textcolor{mycolor2}{0} \cdot \chi( V_{(\mathbb{C}^*)^{\E-1}}(\overline{{\cal F}}) \setminus V_{(\mathbb{C}^*)^{\E-1}}(\overline{{\cal U}})) + \textcolor{mycolor5}{0} \cdot \chi( V_{(\mathbb{C}^*)^{\E-1}}(\overline{{\cal U}}) \setminus V_{(\mathbb{C}^*)^{\E-1}}(\overline{{\cal F}})) \\
        & = - \chi \left ( (\mathbb{C}^*)^{\E-1} \setminus V_{(\mathbb{C}^*)^{\E-1}}(\overline{{\cal F}}\cdot \overline{{\cal U}}) \right ).
    \end{align*} 
    We now consider a different torus $\tilde{T} = (\mathbb{C}^*)^{\E}$ with coordinates $\alpha_1, \ldots, \alpha_\E$. The map $\tilde{T} \rightarrow T$ given by 
    \[ (\alpha_1, \ldots, \alpha_\E) \longmapsto \left( \frac{\alpha_1}{\alpha_\E}, \ldots, \frac{\alpha_{\E-1}}{\alpha_\E}, \alpha_\E^{d_{\cal F}-d_{\cal U}} \right ) \]
    sends $\tilde{T} \setminus V_{\tilde{T}}({\cal F} + {\cal U})$ to $T \setminus V_{T}(y \cdot \overline{{\cal F}} +  \overline{{\cal U}})$. Fibers consist of $d_{\cal F} - d_{\cal U}$ points. Hence, 
    \begin{align}
    \chi(\tilde{T} \setminus V_{\tilde{T}}({\cal F} + {\cal U})) \, &= \, (d_{\cal F}-d_{\cal U}) \cdot \chi(T \setminus V_{T}(y \cdot \overline{{\cal F}} +  \overline{{\cal U}})) \\ &= \, (d_{\cal U}-d_{\cal F}) \cdot \chi \left ( (\mathbb{C}^*)^{\E-1} \setminus V_{(\mathbb{C}^*)^{\E-1}}(\overline{{\cal F}}\cdot \overline{{\cal U}}) \right )\, .\qedhere
    \end{align} 
\end{proof}
Note that, when the polynomials ${\cal U}$ and ${\cal F}$ depend on parameters $z$, the Euler discriminants of the hypersurfaces defined by ${\cal U}(\alpha;z) + {\cal F}(\alpha;z)$ and $\overline{{\cal U}}(\alpha;z) \cdot \overline{{\cal F}}(\alpha;z)$ coincide. We point out that \cite[Lem.~48]{Bitoun:2017nre} is a special instance of Thm.~\ref{thm:EulerChar}. 

\subsection{Beyond the standard classification} \label{subsec:beyondstandard}

In this section, we make a comparison between the principal Landau determinant and a textbook formulation of Landau equations. In particular, we explain how our classification of singularities is different from that usually employed in the literature.

Let us first consider the case with all internal massive edges, $\m_i \neq 0$, which closely matches with the standard formulation \cite{Eden:1966dnq,nakanishi1971graph}. Recall that for any connected subdiagram $\gamma \subset G$, the result of substituting $\alpha_e \to \epsilon \alpha_e$ for every $e \in \gamma$ is 
\begin{equation}\label{eq:UF-limits}
\begin{matrix}
\U_G\big|_{\alpha_\gamma \rightarrow \epsilon \alpha_\gamma}  \,=\, \epsilon^{\L_\gamma}\, \U_\gamma\, \U_{G/\gamma} \,+\, {\cal O}(\epsilon^{\L_\gamma + 1})\, , \quad \F_{G}\big|_{\alpha_\gamma \rightarrow \epsilon \alpha_\gamma} \,=\, \epsilon^{\L_\gamma}\, \U_{\gamma}\, \F_{G/\gamma} \,+\, {\cal O}(\epsilon^{\L_\gamma + 1})\, ,\smallskip\\
{\cal G}_G\big|_{\alpha_\gamma \rightarrow \epsilon \alpha_\gamma} \, = \, \epsilon^{\L_\gamma}\, \U_{\gamma}\,{\cal G}_{G/\gamma} \,+\, {\cal O}(\epsilon^{\L_\gamma + 1})
\end{matrix}
\end{equation}
where $G/\gamma$ denotes the \emph{reduced diagram} obtained from $G$ by contracting all the edges in $\gamma$ and identifying all the vertices in $\gamma$.
The above assumption on the masses implies that the right-hand sides have at least one non-vanishing monomial at order $\epsilon^{\L_\gamma}$.
The proof is standard, see, e.g., \cite[Prop.~4]{Mizera:2021icv}.  This result allows us to label facets by subgraphs $\gamma$. Let $\mathbf{w}_\gamma$ be the weight vector whose entries are $1$ for every edge $i \in \gamma$ and $0$ otherwise:
\be
\mathbf{w}_\gamma \,=\, \sum_{i \in \gamma} e_i\, ,
\ee
where $e_i$ is the $i$-th basis vector in $\R^{\E}$. We think of these as vectors in the normal fan of the Newton polytope of ${\cal G}_G$. They select a face of that Newton polytope by minimizing the scalar product. The face $Q$ corresponding to $\mathbf{w}$ is the Newton polytope of the corresponding \emph{initial form}, denoted by ${\rm in}_{\bf w}({\cal G}_G) = {\cal G}_{G,Q}$. We will consider the cases $\mathbf{w} = 0, \mathbf{w}_\gamma$, $-\mathbf{w}_\gamma$, and $-\mathbf{w}_G$ to match the types of singularities studied in the literature.

\paragraph{Leading second-type singularities.}
Firstly, the dense face (the interior of the polytope) corresponds to $\mathbf{w} = 0$.  The incidence variety \eqref{eq:incidence-variety}, is defined by the equations
\be
\G_G(\alpha;z) \,=\, \partial_{\alpha_e} \G_G(\alpha;z) \,=\, 0 \text{ for all } e\in G
\ee
for $\alpha \in (\C^\ast)^\E$.
The corresponding components in ${\rm PLD}_G$ are known in the literature as \emph{leading} (also called \emph{pure}) \emph{second-type singularities} \cite{Cutkosky:1960sp,doi:10.1063/1.1724262}.

\paragraph{Leading first-type singularities.}
The simplest nonzero weight vector is $-\mathbf{w}_G = (-1, \ldots, -1)$. Since the homogeneity degree of $\F_G$ in the $\alpha$'s is one higher than that of $\U_G$, only the monomials in the first polynomial survive in the initial form:
\be
\mathrm{in}_{-\mathbf{w}_G}(\G_G) \,=\, \F_{G}\, .
\ee
The corresponding incidence variety is defined by the equations
\be\label{eq:leading-first-type}
\partial_{\alpha_e} \F_G(\alpha;z) \,=\, 0 \text{ for all } e\in G\, 
\ee
for $\alpha \in (\C^\ast)^\E$.
Since $\F_G$ is homogeneous in $\alpha$, the equation $\F_G = 0$ would be redundant and hence does not need to be written down. If we additionally impose the inequality $\U_G \neq 0$, these would give what are known in the literature as \emph{leading singularities} (of the first type) \cite{Bjorken:1959fd,Landau:1959fi,10.1143/PTP.22.128}, later formalized as the \emph{Landau discriminant} \cite{Mizera:2021icv}.

\paragraph{Subleading second-type singularities.}

For the weights $\mathbf{w}_\gamma$, applying the factorization properties \eqref{eq:UF-limits} gives
\be\label{eq:subleading-second-type}
\mathrm{in}_{\mathbf{w}_\gamma}(\G_G) \,=\, \U_\gamma\, \G_{G/\gamma}\, .
\ee
Recall that the variables $\alpha$ in $\gamma$ and $G/\gamma$ are disjoint. There are several components and we first consider the case $\G_{G/\gamma} \neq 0$. The system of equations defining the incidence variety involves $
\partial_{\alpha_e} \U_{\gamma} = 0 \text{ for all }e\in \gamma $
for $\alpha \in (\C^\ast)^\E$ (once again, $\U_\gamma = 0$ is redundant). Since none of the equations depends on the kinematic variables $z$, it either gives no solutions or a dominant component that we discard from the PLD. Hence we need $\G_{G/\gamma} = 0$ for non-trivial solutions. In this case, the system is
\begin{equation}
\partial_{\alpha_e} \G_{G/\gamma}(\alpha;z) \,=\, 0 \text{ for all }e\in G/\gamma 
\end{equation}
for $\alpha \in (\C^\ast)^\E$. Note that the variables $\alpha_e$ for $e \in \gamma$ do not appear and hence are unconstrained. This is the same system of equations as for the dense face, but for the diagram $G/\gamma$ instead of $G$. In the literature, these are referred to as \emph{subleading singularities of the second type} (also called \emph{mixed second type}) \cite{Drummond1963,Boyling1968}.  

\paragraph{Subleading first-type singularities.}
Finally, let us consider the weights $-\mathbf{w}_\gamma$. Using homogeneity properties of $\U_G$ and $\F_G$, we find
\be
\mathrm{in}_{-\mathbf{w}_\gamma}(\G_G)\,=\, \U_{\gamma^c} \F_{G/\gamma^c}\,.
\ee
Here, $\gamma^c$ denotes the complement of $\gamma$ in $G$. The analysis is entirely analogous to the case \eqref{eq:subleading-second-type}. The solutions with $\F_{G/\gamma^c} \neq 0$ can be discarded. We are hence left with
\be
\partial_{\alpha_e} \F_{G/\gamma^c} (\alpha;z) \,=\, 0 \text{ for all } e \in G/\gamma^c
\ee
for $\alpha \in (\C^\ast)^\E$. This is the same system of equations as \eqref{eq:leading-first-type}, except with $G/\gamma^c$ instead of $G$. Solutions of such equations are known as \emph{subleading singularities of the first kind}.\\

\label{subsec:parachute}

One of the simplest examples of Landau singularities is associated to the parachute diagram $G = \mathtt{par}$ illustrated in Fig.~\ref{fig:diagrams}k. In order to make it more interesting and conform to the above assumptions, we will make all the masses distinct and non-zero. The kinematic space $\cal K$ is therefore parametrized by $(s, \M_3, \M_4, \m_1, \m_2, \m_3, \m_4) \in \C^{7}$. The graph polynomial is $\G_{\tt par} = \U_{\tt par} + \F_{\tt par}$ with
\begin{align}
\U_{\tt par} &\,=\, (\alpha_1+\alpha_2)(\alpha_3 + \alpha_4) + \alpha_3 \alpha_4,\\
\F_{\tt par} &\,=\, s \alpha_1 \alpha_2 (\alpha_3 + \alpha_4) + \M_3 \alpha_1 \alpha_3 \alpha_4 + \M_4 \alpha_2 \alpha_3 \alpha_4 - (\m_1 \alpha_1 + \m_2 \alpha_2 + \m_3 \alpha_3 + \m_4 \alpha_4) \U_{\tt par}.\nn
\end{align}
The rays of the normal fan of $\Newt(\G_{\tt par})$ index its facets. They are given by
\begin{gather}
\big\{
(-1, -1, -1, -1), (1, 0, 0, 0), (0, 1, 0, 0), (0, 0, 1, 0), (0, 0, 0, 1), \\
(0, 0, 1, 1), (1, 1, 1, 0), (1, 1, 0, 1), (1, 1, 1, 1)  \big\}.
\end{gather}
The $f$-vector is $(15, 33, 27, 9)$. Hence, together with the dense face, there are in total $85$ systems of equations to solve. Note that this is much larger than the naive counting $2^4 = 16$ (each edge collapsed or not) of reduced diagrams. Let us first discuss a few interesting cases that lead to discriminants with degree larger than $1$. We will keep using the notation $\m_i = m_i^2$ whenever convenient.

\paragraph*{Weight $\mathbf{(-1,-1,0,0)}$.}
The weight vector $(-1,-1,0,0)$ lies in the relative interior of the 2-dimensional cone generated by $(-1,-1,-1,-1)$ and $(0,0,1,1)$. Its initial form is
\be
\mathrm{in}_{(-1,-1,0,0)}(\G_{\tt par}) \,=\, (\alpha_3 + \alpha_4) [ (s - \m_1 - \m_2) \alpha_1 \alpha_2 - \m_1 \alpha_1^2 - \m_2 \alpha_2^2 ]\, . 
\ee
It gives rise to the incidence variety carved out by the following set of equations:
\begin{subequations}
\begin{align}
(\alpha_3 + \alpha_4) [ (s - \m_1 - \m_2) \alpha_1 \alpha_2 - \m_1 \alpha_1^2 - \m_2 \alpha_2^2 ] &\,=\, 0, \\
(s - \m_1 - \m_2) \alpha_1 \alpha_2 - \m_1 \alpha_1^2 - \m_2 \alpha_2^2 &\,=\, 0, \label{eq:mmpp-2}\\
(\alpha_3 + \alpha_4)[(s - \m_1 - \m_2) \alpha_2 - 2\m_1 \alpha_1] &\,=\, 0,\\
(\alpha_3 + \alpha_4)[(s - \m_1 - \m_2) \alpha_1 - 2\m_1 \alpha_2] &\,=\, 0.
\end{align}
\end{subequations}
The incidence variety has $2$ components, both of dimension $9$. The first one has equation $\alpha_3 + \alpha_4 = 0$ and \eqref{eq:mmpp-2} and hence projects dominantly to the kinematic space. It has the physical interpretation of a UV sub-divergence associated to shrinking the bubble subdiagram. The second component is seen by setting $\alpha_3 + \alpha_4 \neq 0$. It projects to codimension $1$ in the kinematic space, giving the discriminant:
\be\label{eq:Epar-1}
E_{\mathtt{par},1}({\mathcal K}) \,=\, \lambda(s,\m_1,\m_2) \,=\, [s - (m_1 + m_2)^2] [s - (m_1 - m_2)^2]\, ,
\ee
where $\lambda$ is the K\"all\'en function \eqref{eq:Kallen}. It corresponds to a $3$-dimensional fiber with $(\alpha_1 : \alpha_2) = (1/m_1 : \pm 1/m_2)$ and any $\alpha_3, \alpha_4 \in \C^\ast$, see \cite[Sec.~2.6]{Mizera:2021icv}. These components are called normal and pseudo-normal thresholds in the $s$-channel. 

In \cite[Sec.~6.4]{Berghoff:2022mqu}, the authors identify a component \cite[Eq.~(6.15)]{Berghoff:2022mqu} of what they call the \emph{Landau variety} (see also \cite{Landshoff1966,doi:10.1063/1.1724262}). In our notation, this component is
\begin{equation} \label{eq:panzerberghoff} 
F_{\mathtt{par}} \,=\, s(\M_4-\m_1)(\M_3-\m_2) - (\m_1 \M_3 - \m_2 \M_4)(\m_2-\m_1 + \M_4-\M_3) \, = \, 0.
\end{equation}
Unfortunately, this is \emph{not} part of the principal Landau determinant. The reason is similar to what we saw in Ex.~\ref{ex:embcomp}. To analyze this in more detail, we made the simplification $(\M_3,\M_4,\m_1,\m_2,\m_3,\m_4) = (1, 1, 2, 3, 1, 2)$. We modify the incidence equations as follows:  
\be
\alpha_2^{-2} \alpha_4 \cdot \G_{\tt par}(\alpha;z) = \alpha_2^{-2} \alpha_4 \cdot \alpha_e \cdot \partial_{\alpha_e} \G_{\tt par}(\alpha;z) = 0 \text{ for all } e\in {\tt par}.
\ee
This does not change the solutions in $(\mathbb{C}^*)^{\E}$. Then, we apply the invertible change of coordinates $(\alpha_1,\alpha_2,\alpha_3,\alpha_4) = (y_1y_4^{-1}, y_4^{-1}, y_2y_3y_4^{-1}, y_2 y_4^{-1})$. This leads to five polynomials in $\mathbb{C}[s][y_1,y_2,y_3,y_4]$.
For the reader who is familiar with toric geometry, we are expressing the incidence equations in coordinates on a copy of $\mathbb{C}^4$ inside the toric variety associated to ${\rm Newt}({\cal G}_{\tt par})$. More precisely, $y_1,\ldots,y_4$ are coordinates on the affine piece of this projective toric variety corresponding to the smooth vertex $(0,2,0,1)$ \cite[Sec.~3.5]{telen2022introduction}.

Our two-dimensional cone coming from weight $(-1,-1,0,0)$ now corresponds to the coordinate subspace $\{y_2 = y_4 = 0 \}$. The ideal generated by our equations in $\mathbb{C}[s][y_1,y_2,y_3,y_4]$ has six primary components, of which we display the following two: 
\begin{align*}
    P_1 &\,=\, \langle \, y_3 + 1, y_2, 2y_1^2 - y_1y_4 - sy_1 + 5y_1 - y_4 + 3 \,\rangle, \\
    P_2 &\,=\, \langle 2s + 1, y_4, (y_3+1)^2, 2y_2 + 5y_3 + 5, y_1 - y_3 + 1 \rangle. 
\end{align*}
The first component $P_1$ was identified above: it projects dominantly to $s$-space and is contained in $\alpha_3 + \alpha_4 = 0$ (in our old coordinates). The second primary component $P_2$ is an embedded component of $P_1$, which is contained in $\{y_2 = y_4 = 0 \}$ and projects to $\{2s + 1 = 0\}$. One checks that this is the component \eqref{eq:panzerberghoff} identified by Berghoff and Panzer, specialized to our choice of masses. Like in Ex. \ref{ex:embcomp}, we analyzed primary components of the incidence equations, extended to a partial compactification containing the locus of the singularity, in this case $\{y_2 = y_4 = 0 \}$. 

\paragraph*{Weights $\mathbf{(-1,-1,1,0)}$ and $\mathbf{(-1,-1,0,1)}$.}
Similarly, one obtains the contributions coming from the three-dimensional cones constructed by adding either the ray $(0,0,1,0)$ or $(0,0,0,1)$ to the above two-dimensional cone. Due to symmetry, it is enough to look at the first case, which has weight $(-1,-1,1,0)$ and the initial form:
\be
\mathrm{in}_{(-1,-1,1,0)}(\G_{\tt par}) \,=\, \alpha_4 [ (s - \m_1 - \m_2) \alpha_1 \alpha_2 - \m_1 \alpha_1^2 - \m_2 \alpha_2^2 ]\, . 
\ee
The discussion is parallel to the previous case and hence these faces also lead to the discriminant $E_{\mathtt{par},1}({\cal K})$ from \eqref{eq:Epar-1}. However, physically, it comes from the Schwinger proper times expanding and contracting at different relative rates according to the weights $(-1,-1,1,0)$. This phenomenon was previously observed in \cite{Landshoff1966} in a more ad-hoc manner.

\paragraph*{Weights $\mathbf{(-1,0,-1,-1)}$ and $\mathbf{(0,-1,-1,-1)}$.}
The two-dimensional cone generated by $-(1,1,1,1)$ and $(0,1,0,0)$ contains the weight $-(1,0,1,1)$. The initial form~is
\be
\mathrm{in}_{(-1,0,-1,-1)}(\G_{\tt par}) \,=\, \M_4 \alpha_1 \alpha_3 \alpha_4 - (\m_1 \alpha_1 + \m_3 \alpha_3 + \m_4 \alpha_4)(\alpha_1 \alpha_3 + \alpha_1 \alpha_4 + \alpha_3 \alpha_4)\, .
\ee
It corresponds to the reduced diagram obtained by shrinking the edge with the Schwinger parameter $\alpha_2$. Hence, the equations are analogous to those appearing in Ex.~\ref{ex:banana-generic} for the banana diagram ${\tt B}_3$ with modified kinematics. It has two components that contribute to the principal Landau determinant:
\begin{align}
E_{\mathtt{par},2}({\mathcal K}) \,=\, &\,[\M_4 - (m_1 + m_3 + m_4)^2][\M_4 - (m_1 + m_3 - m_4)^2]\\
&\cdot[\M_4 - (m_1 - m_3 + m_4)^2][\M_4 - (m_1 - m_3 - m_4)^2],
\end{align}
and $\{ \M_4 = 0 \}$. The former comes from the one-dimensional fiber with $(\alpha_1 : \alpha_3 : \alpha_4) = (1/m_1 : \pm 1/m_3 : \pm 1/m_4)$ and $\alpha_2 \in \C^\ast$ and is the subleading first-type singularity, also known as the normal and pseudo-normal thresholds in the $\M_3$-channel. The latter is associated to a one-dimensional fiber determined by $\{ \m_1 \alpha_1 + \m_3 \alpha_3 + \m_4 \alpha_4 = 0 \} \cap \{\alpha_1 \alpha_3 + \alpha_1 \alpha_4 + \alpha_3 \alpha_4 = 0\}$ and is the subleading second-type singularity.

By symmetry, an analogous computation for the weight $(0,-1,-1,-1)$ leads to 
\be
E_{\mathtt{par},3}({\mathcal K}) \,=\, E_{\mathtt{par},2}({\mathcal K}) \Big|_{\m_1 \to \m_2, \M_4 \to \M_3}
\quad \text{and} \quad \M_3 = 0.
\ee

\paragraph*{Weight $\mathbf{(-1,-1,-1,-1)}$.}
The one-dimensional cone given by the ray $(-1,-1,-1,-1)$ has the initial form
\be
\mathrm{in}_{(-1,-1,-1,-1)}(\G_{\tt par}) \,=\, \F_{\tt par}.
\ee
This is the leading singularity. One of its components is an irreducible variety of degree $6$ in $\cal K$ given by:
\begin{align}
E_{\mathtt{par},4}({\mathcal K}) \,=\, \big(&m_1^2 \M_3 s+m_3^2 \M_3 s+m_4^2 \M_3 s+2 m_3 m_4 \M_3 s+m_2^2 \M_4 s+m_3^2 \M_4 s+m_4^2 \M_4
   s\\
   &+2 m_3 m_4 \M_4 s+m_1^4 + \M_3 m_1^2 \M_3^2+m_2^2 m_1^2 \M_3+m_3^2 m_1^2
   \M_3+m_4^2 m_1^2 \M_3\\
   &+2 m_3 m_4 m_1^2 \M_3+m_2^2 m_1^2 \M_4-m_3^2 m_1^2 \M_4-m_4^2 m_1^2
   \M_4-2 m_3 m_4 m_1^2 \M_4 \\
   &+m_1^2 \M_3 \M_4-m_2^2 \M_4^2-m_2^2 m_3^2 \M_3-m_2^2 m_4^2 \M_3-2
   m_2^2 m_3 m_4 \M_3-m_2^4 \M_4 \\
   &+m_2^2 m_3^2 \M_4+m_2^2 m_4^2 \M_4+2 m_2^2 m_3 m_4
   \M_4+m_2^2 \M_3 \M_4-m_3^2 s^2-m_4^2 s^2\\
   &-2 m_3 m_4 s^2-m_2^2 m_1^2 s+m_3^2 m_1^2
   s+m_4^2 m_1^2 s+2 m_3 m_4 m_1^2 s-m_3^4 s-m_4^4 s\\
   &-4 m_3 m_4^3 s+m_2^2 m_3^2 s+m_2^2
   m_4^2 s-6 m_3^2 m_4^2 s-4 m_3^3 m_4 s+2 m_2^2 m_3 m_4 s\\
   &-\M_3 \M_4 s \big)( m_3 \leftrightarrow m_4 ),\label{eq:par-leading-disc}
\end{align}
where the factor $( m_3 \leftrightarrow m_4 )$ is given by switching the variable $m_3$ with $m_4$ in the first factor. It has zero-dimensional fibers, i.e., the Schwinger parameters are localized to points.
As before, we used $\m_i = m_i^2$, in terms of which the discriminant factors into two irreducible components. Another component is of degree $2$ and given by
\be\label{eq:par-leading-disc2}
E_{\mathtt{par},5}({\mathcal K}) \,=\, \lambda(s,\M_3,\M_4),
\ee
which also comes from a $0$-dimensional fiber.

\paragraph*{Weight $\mathbf{(0,0,0,0)}$.}
Finally, the contribution from the dense face has the initial form equal to the graph polynomial itself:
\be
\mathrm{in}_{(0,0,0,0)}(\G_{\tt par}) \,=\, \G_{\tt par}.
\ee
It has a one-dimensional fiber which also projects down to the component $E_{\mathtt{par},5}({\mathcal K})$. It is the leading second-type singularity.

\paragraph*{Principal Landau determinant.}
The remaining discriminants can be analyzed in an analogous fashion and yield only degree-$1$ or empty components. As a result of this:
\be
E_{\mathtt{par}}({\mathcal K}) \,=\, \m_1 \m_2 \m_3 \m_4 \M_3 \M_4 s \prod_{i=1}^{5} E_{\mathtt{par},i}({\mathcal K}).
\ee
Likewise, after including the component \eqref{eq:panzerberghoff}, the Euler discriminant is given by
\be
\nabla_{\chi}({\cal K}) \,=\, \{ E_{\text{par}}({\cal K}) F_{\text{par}} \,=\, 0 \} \, .
\ee
We verified that \eqref{eq:panzerberghoff} is the only extra factor by running \texttt{cgReduction} \cite{Panzer:2014caa} and collecting all candidate components on which the signed Euler characteristic drops, see App.~\ref{sec:appendixHyperInt} for details.
Subsets and special cases of these Landau singularities were also found in \cite{Landshoff1966,doi:10.1063/1.1724262,Lairez:2022zkj,Berghoff:2022mqu,Hannesdottir:2022xki}.

For diagrams involving massless particles, the face structure of $\mathrm{Newt}(\G_G)$ may be drastically different, because the factorization implied by \eqref{eq:UF-limits} changes in such cases. More specifically, for a given $\gamma$, we might encounter $\F_{G/\gamma} = 0$ and $\mathcal{O}(\epsilon^{\L_\gamma+1})$ terms are needed to understand the leading behavior. Physically, such situations are associated with infrared (IR) divergences \cite{Arkani-Hamed:2022cqe}. Let us illustrate it on the parachute example from the previous subsection.
This time, we consider the kinematic subspace
\be
{\cal E} \,=\, {\cal K} \cap \{ s = \m_1 = \m_2 = 0 \}.
\ee
It means that $\G_{\tt par}({\cal E}) \,=\, \U_{\tt par} + \F_{\tt par}({\cal E})$ with
\be
\F_{\tt par}({\cal E}) \,=\, \M_3 \alpha_1 \alpha_3 \alpha_4 \,+\, \M_4 \alpha_2 \alpha_3 \alpha_4 \,-\, (\m_3 \alpha_3 + \m_4 \alpha_4) \, \U_{\tt par}.
\ee
The resulting polytope $\Newt(\G_{\tt par}({\cal E}))$ has fewer faces compared to the generic-mass case and its $f$-vector is $(11,23,19,7)$. The principal Landau determinant reads
\be
E_{\tt par}({\cal E}) \,=\, \m_3 \m_4 \M_3 \M_4 (\M_3 - \M_4) \lambda(\M_3, \m_3, \m_4) \lambda(\M_4, \m_3, \m_4).
\ee
Dominant components have been filtered out according to Def.~\ref{def:PLD}. We verified that $\nabla_{\chi}(\mathcal{E}) = \{ E_{\text{par}}(\mathcal{E}) = 0\}$, i.e., PLD does not miss any components in this case.

\section{One-loop and banana diagrams}\label{sec:4}

In this section, we consider the application to the simplest examples of one-loop diagrams with $n$ external legs and banana diagrams with $\E$ internal edges. Singularities of Feynman integrals belonging to these families are well-known, see, e.g., \cite{nakanishi1971graph}. The purpose of the forthcoming discussion is to demonstrate how to phrase this analysis in terms of principal Landau determinants.

\subsection{One-loop diagrams}

For the family of one-loop diagrams ${\tt A}_n$ with $n$ external legs, $\E$ internal edges with $n=\E$, and generic masses, illustrated in \cite[Fig.~1a]{Mizera:2021icv}, the Symanzik polynomials are
\begin{equation}\label{eq:UFoneloop} 
{\cal U}_{{\tt A}_n} \,=\, \alpha_1+\dots+\alpha_n, \quad {\cal F}_{{\tt A}_n} \,=\, \sum_{i<j} (s_{i,i+1,\dots,j-1}-\m_i - \m_j) \, \alpha_i \alpha_j - \sum_{i=1}^n \m_i \alpha_i^2.
\end{equation}
In the special case $j=i+1$, $s_{i}=\M_i$. The $s = n+\binom{n+1}{2}$ coefficients of the polynomial ${\cal G}_{{\tt A}_n} \coloneqq {\cal U}_\An + {\cal F}_\An$ are either constants, or linear functions of $s_{i,i+1,\dots,j-1},\m_i$ that parameterize the kinematic space ${\cal K}\subset\C^s$. We will write $A_n$ for the integer matrix of size $n\times s$ with columns given by the exponents of ${\cal G}_\An$. As explained in Sec.~\ref{sec:2}, the polynomial ${\cal G}_\An$ defines a hypersurface $V_{A_n,(s_{i,i+1,\dots,j-1},\m_i)}$ in the algebraic torus $(\C^*)^n$ when fixing the coefficients $s_{i,i+1,\dots,j-1},\m_i$. To simplify the notation, we will drop the subscript $(s_{i,i+1,\dots,j-1},\m_i)$ when referring to this variety.

In what follows, we study the principal Landau determinant $E_{{\tt A}_n}({\cal E})$ associated to one-loop diagrams ${\tt A}_n$, when restricting to the subspaces ${\cal E}$ of the kinematic space ${\cal K}\subset \C^s$ introduced in Sec.~\ref{sec:2}. 
In particular, we show that none of the factors of the principal $A$-determinant vanishes identically when substituting parameters in $\{{\cal K},{\cal E}^{(\M_i,0)}\}$. Furthermore, we conjecture that the same statement holds for ${\cal E}^{(0,0)}$. Geometrically, this means that intersecting the principal $A$-determinant variety with the kinematic space results in a proper subvariety of both. In accordance with Thm.~\ref{thm:eulercharvolume}, this behaviour is predicted from the computations of the Euler characteristic of the variety $V_{\An}({\cal E})$ and the normalized volume of the polytope $\text{Conv}(A_n({\cal E}))$ shown in Tab.~\ref{tab:ngon}.
\begin{table}[t]
\begin{center}
\begin{tabular}{ c | c | c | c | c}
   $\quad$ &  ${\cal K}$ & ${\cal E}^{(\M_i,0)}$ & ${\cal E}^{(0,\m_e)}$ & ${\cal E}^{(0,0)}$ \\
  \hline
$(-1)^n\cdot\chi(V_{\An}({\cal E}))$ & $2^n-1$ & $2^n-1-n$ & $\bm{2^n-1-n}$ & $\bm{2^n-1+n-n^2}$\\
$\text{vol}(A_n({\cal E}))$ & $2^n-1$ & $2^n-1-n$ & $2^n-1$ & $\bm{2^n-1+n-n^2}$
\end{tabular}
\caption{Values of signed Euler characteristic $(-1)^n\cdot\chi(V_{\An}({\cal E}))$ and volume $\text{vol}(A_n({\cal E}))$ for one-loop diagrams $\An$ with $n\geq 2$ external legs. We thank Hjalte Frellesvig for providing the formulas in the case ${\cal E}^{(0,0)}$. To our knowledge, bold entries are conjectured, others are proved. 
}
\label{tab:ngon}
\end{center}
\end{table}

For the subspace ${\cal E}^{(0,0)}$, we verified the equality between volume and Euler characteristic for $n=2,\dots,10$. However, understanding the face structure of the polytope $\conv(A_n({\cal E}^{(0,0)}))$ turns out to be quite challenging. Therefore, we could not prove the formula in the last column of Tab.~\ref{tab:ngon} in full generality. 

According to Def.~\ref{eq:EA}, computing the principal $A$-determinant for a one-loop diagram boils down to two steps: (i) understanding the faces of the polytope $\text{Conv}(A_n({\cal E}))$; and (ii) computing the discriminant associated to each face. 

Since ${\cal G}_\An$ has degree 2, in step (ii) we can make use of well-known descriptions of the discriminant of a quadratic form in terms of its symmetric matrix, see \cite[Ex.~1.3~(b)]{gelfand2008discriminants}. 
A general quadratic form in the variables $\alpha=(\alpha_1, \ldots, \alpha_n)$ supported on $A_n$ is
\begin{equation}\label{eq:generalPolynomial}
     f_n(\alpha;z) \,=\, \sum_{ 0 \leq i \leq  j \leq n }\!\! z_{ij} \, \alpha_i \, \alpha_j  \,=\, \frac{1}{2}  (1~ \alpha_1 ~ \cdots ~ \alpha_n) \begin{pmatrix} 0 & z_{01} & \cdots &  z_{0n} \\
z_{01} & 2z_{11} & \cdots &  z_{1n} \\
\vdots & \vdots & \ddots & \vdots \\ 
z_{0n} & z_{1n} & \cdots & 2z_{nn}
\end{pmatrix} \begin{pmatrix} 1 \\ \alpha_1 \\ \vdots \\ \alpha_n \end{pmatrix},
\end{equation}
where, in the first expression, we set $\alpha_0 = 1$ and $z_{00} = 0$ for convenience. We denote the coefficient matrix of size $n+1$ displayed above by $Z$. The principal $A_n$-determinant is a polynomial in the coefficients $z_{ij}$. Restricting to the subspaces of the kinematic space listed above corresponds to setting some of the entries of the matrix $Z$ to zero. It will be convenient to use the symmetric matrix $Z$ to describe the factors of the principal Landau determinant $E_\An({\cal E})$.

\subsubsection{Generic masses}
 We begin the study of the principal Landau determinant $E_{{\tt A}_n}({\cal K})$ by understanding in detail the facet description of the Newton polytope of the polynomial ${\cal G}_{\An}$.  
 The Newton polytope is a truncation of a dilated standard simplex, and the associated toric variety is the blow-up of $\mathbb{P}^n$ at one of its torus invariant points. To make this precise, we introduce the following notation. Let $[n] = \{ 1, \ldots, n \}$ and write $e_i$ for the $i$-th standard basis vector of $\mathbb{R}^n$. For any subset $I \subset [n]$, we write 
\[ T(I) \, = \, {\rm Conv}(\,e_i, \,2 \cdot e_i \, : \,  i \in I\,).\]
With this notation, we have ${\rm Newt}({\cal G}_{{\tt A}_n}) = T([n])$. The polytopes $T([2])$, $T([3])$ and a Schlegel diagram for $T([4])$ are shown in Fig.~\ref{fig:Tpolytopes}.
\begin{figure}
	\centering

\begin{tikzpicture}[x  = {(1cm,0cm)},
                    y  = {(0cm,1cm)},
                    z  = {(0cm,0cm)},
                    scale = 1,
                    color = {lightgray}]

  \coordinate (v0_Q) at (1, 0);
  \coordinate (v1_Q) at (0, 1);
  \coordinate (v2_Q) at (2, 0);
  \coordinate (v3_Q) at (0, 2);

  \definecolor{vertexcolor_Q}{rgb}{ 0.2 0.6 1 }

  \tikzstyle{vertexstyle_Q} = [circle, scale=0.375pt, fill=vertexcolor_Q,]

  \definecolor{facetcolor_Q}{rgb}{ 0.6 0.8 1 }

  \definecolor{edgecolor_Q}{rgb}{ 0 0 0 }
  \tikzstyle{facetstyle_Q} = [fill=facetcolor_Q, fill opacity=0.85, draw=edgecolor_Q, line width=0.8 pt, line cap=round, line join=round]

  \draw[facetstyle_Q] (v1_Q) -- (v0_Q) -- (v2_Q) -- (v3_Q) -- (v1_Q) -- cycle;

  \foreach \i in {0,1,2,3} {
    \node at (v\i_Q) [vertexstyle_Q] {};
  }

\end{tikzpicture}\hspace{1.5cm}
	\begin{tikzpicture}[x  = {(0.9cm,-0.076cm)},
                    y  = {(-0.06cm,0.95cm)},
                    z  = {(-0.44cm,-0.29cm)},
                    scale = 1,
                    color = {lightgray}]

  \coordinate (v0_R) at (1, 0, 0);
  \coordinate (v1_R) at (0, 1, 0);
  \coordinate (v2_R) at (0, 0, 1);
  \coordinate (v3_R) at (2, 0, 0);
  \coordinate (v4_R) at (0, 2, 0);
  \coordinate (v5_R) at (0, 0, 2);

  \definecolor{vertexcolor_R}{rgb}{ 0.2 0.6 1 }

  \tikzstyle{vertexstyle_R} = [circle, scale=0.375pt, fill=vertexcolor_R,]

  \definecolor{facetcolor_R}{rgb}{ 0.6 0.8 1 }

  \definecolor{edgecolor_R}{rgb}{ 0 0 0 }
  \tikzstyle{facetstyle_R} = [fill=facetcolor_R, fill opacity=0.75, draw=edgecolor_R, line width=0.8 pt, line cap=round, line join=round]

  \draw[facetstyle_R] (v0_R) -- (v3_R) -- (v5_R) -- (v2_R) -- (v0_R) -- cycle;
  \draw[facetstyle_R] (v4_R) -- (v3_R) -- (v0_R) -- (v1_R) -- (v4_R) -- cycle;
  \draw[facetstyle_R] (v0_R) -- (v2_R) -- (v1_R) -- (v0_R) -- cycle;
  \draw[facetstyle_R] (v2_R) -- (v5_R) -- (v4_R) -- (v1_R) -- (v2_R) -- cycle;

  \foreach \i in {2,1,0} {
    \node at (v\i_R) [vertexstyle_R] {};
  }

  \draw[facetstyle_R] (v5_R) -- (v3_R) -- (v4_R) -- (v5_R) -- cycle;

  \foreach \i in {5,4,3} {
    \node at (v\i_R) [vertexstyle_R] {};
  }

\end{tikzpicture}\hspace{1.5cm}
	\begin{tikzpicture}[x  = {(0.9cm,-0.076cm)},
	y  = {(-0.06cm,0.95cm)},
	z  = {(-0.44cm,-0.29cm)},
	scale = 1,
	color = {lightgray}]

	\coordinate (v0_P) at (1, 0, 0);
	\coordinate (v1_P) at (0, 1, 0);
	\coordinate (v2_P) at (0.371429, 0.371429, 0.371429);
	\coordinate (v3_P) at (0, 0, 1);
	\coordinate (v4_P) at (2, 0, 0);
	\coordinate (v5_P) at (0, 2, 0);
	\coordinate (v6_P) at (0.604651, 0.67, 0.604651);
	\coordinate (v7_P) at (0, 0, 2);

	\definecolor{vertexcolor_P}{rgb}{ 0.2 0.6 1 }
	
	\tikzstyle{vertexstyle_P} = [circle, scale=0.325pt, fill=vertexcolor_P,]
	
	\definecolor{edgecolor_P_0}{rgb}{ 0 0 0 }
	\definecolor{edgecolor_P_1}{rgb}{ 0.2 0.6 1 }
	\tikzstyle{edgestyle_P_1_0} = [line width=0.75pt,, color=edgecolor_P_0,]
	\tikzstyle{edgestyle_P_2_0} = [line width=0.75pt,, color=edgecolor_P_1,]
	\tikzstyle{edgestyle_P_2_1} = [line width=0.75pt,, color=edgecolor_P_1,]
	\tikzstyle{edgestyle_P_3_0} = [line width=0.75pt,, color=edgecolor_P_0,]
	\tikzstyle{edgestyle_P_3_1} = [line width=0.75pt,, color=edgecolor_P_0,]
	\tikzstyle{edgestyle_P_3_2} = [line width=0.75pt,, color=edgecolor_P_1,]
	\tikzstyle{edgestyle_P_4_0} = [line width=0.75pt,, color=edgecolor_P_0,]
	\tikzstyle{edgestyle_P_5_1} = [line width=0.75pt,, color=edgecolor_P_0,]
	\tikzstyle{edgestyle_P_5_4} = [line width=0.75pt,, color=edgecolor_P_0,]
	\tikzstyle{edgestyle_P_6_2} = [line width=0.75pt,, color=edgecolor_P_1,]
	\tikzstyle{edgestyle_P_6_4} = [line width=0.75pt,, color=edgecolor_P_1,]
	\tikzstyle{edgestyle_P_6_5} = [line width=0.75pt,, color=edgecolor_P_1,]
	\tikzstyle{edgestyle_P_7_3} = [line width=0.75pt,, color=edgecolor_P_0,]
	\tikzstyle{edgestyle_P_7_4} = [line width=0.75pt,, color=edgecolor_P_0,]
	\tikzstyle{edgestyle_P_7_5} = [line width=0.75pt,, color=edgecolor_P_0,]
	\tikzstyle{edgestyle_P_7_6} = [line width=0.75pt,, color=edgecolor_P_1,]
	
	
	\foreach \i/\k in {1/0,2/0,2/1,3/0,3/1,3/2,4/0,5/1,5/4,6/2,6/4,6/5,7/3,7/4,7/5,7/6} {
		\draw[edgestyle_P_\i_\k] (v\i_P) -- (v\k_P);
	}

	\foreach \i in {7,3,6,2,0,1,4,5} {
		\node at (v\i_P) [vertexstyle_P] {};
	}

\end{tikzpicture}
	\caption{The truncated simplices $T([2])$, $T([3])$ and $T([4])$.} 
 \label{fig:Tpolytopes}
\end{figure}
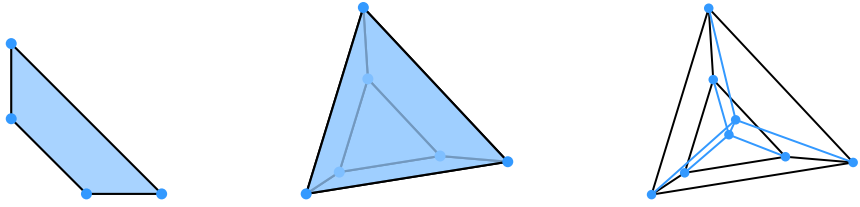
To describe the faces of the truncated simplices $T([n])$, it is convenient to introduce the notation 
\[ S(I) \,=\, {\rm Conv}(e_i \, : \,  i \in I), \quad D(I) \,=\, {\rm Conv}(2 \cdot e_i \, : \,  i \in I), \quad I \subset [n],\]
for the $(|I|-1)$-dimensional simplices $(S)$ and dilated simplices $(D)$ associated to the index set $I$. The following lemma gives a complete description of the face structure of the polytope $T([n])$.

\begin{lemma}\label{lem:facesT}
    The polytope $T([n])$ has $2n$ vertices given by $e_1, \ldots, e_n, 2\cdot e_1, \ldots, 2 \cdot e_n$. Its faces of dimension $1 \leq k \leq n -1$ consist of 
    \begin{enumerate}
        \item $\binom{n}{k+1}$ simplices $\{S(I) \, : \,  I \subset [n], \,  |I| = k+1 \}$, 
        \item $\binom{n}{k+1}$ dilated simplices $\{D(I) \, : \,  I \subset [n], \, |I| = k+1 \}$, and 
        \item $\binom{n}{k}$ truncated simplices $\{T(I)\, : \,  I \subset [n], \, |I| = k \}$.
    \end{enumerate}
    In particular, the f-vector of $T([n])$ is given by $\left (2n, 2\binom{n}{2}+n, 2 \binom{n}{3} + \binom{n}{2}, \ldots, 2  + n \right )$.
\end{lemma}
\begin{proof}
We sketch the proof and leave the details to the reader. Let $\Sigma$ be the normal fan of the standard simplex ${\rm Conv}(0,e_1, \ldots, e_n) \subset \mathbb{R}^n$ and let $\sigma = \mathbb{R}_{\geq 0} \cdot \{ e_1, \ldots, e_n \} \subset \mathbb{R}^n$ be the positive orthant. The normal fan of $T([n])$ is the \emph{star subdivision} $\Sigma^*(\sigma)$ of $\Sigma$ along $\sigma$. The $2n$ vertices of $T([n])$ correspond to its $2n$ full-dimensional cones. The $(n-k)$-dimensional cones come in three types. In terms of ray generators, these types are described by:
\begin{enumerate}
    \item the ray $(1,1,\ldots, 1)$ and $\{e_i \, : \, i \in [n] \setminus I \}$, for $|I| = k+1$,
    \item the ray $(-1,-1,\ldots,-1)$ and $\{e_i \, : \, i \in [n] \setminus I \}$, for $|I| = k+1$,
    \item the rays $\{e_i \, :\, i \in [n] \setminus I \}$, for $|I| = k$.
\end{enumerate}
The corresponding $k$-dimensional faces are those listed in the lemma. 
\end{proof}
Next, we investigate the principal $A_n$-determinant $E_{A_n}$ corresponding to our truncated simplices $T([n])$. Notice that the columns of the matrix $A_n$ are precisely the lattice points in $T([n]) \cap \mathbb{Z}^n$. We will express $E_{A_n}$ as a product of discriminants $\Delta_F$, where $F$ runs over the faces of $T([n])$ listed in Lem.~\ref{lem:facesT}, as described in \eqref{eq:EA}. For this purpose, we index the rows and columns of the matrix $Z$ in \eqref{eq:generalPolynomial} by $\{0,1, \ldots, n\}$ and, for any $I \subset \{0,1, \ldots, n\}$, we write $Z_I$ for the square submatrix with rows and columns indexed by $I$. 

\begin{lemma}\label{lem:discrformula}
The face discriminants of $A_n = T([n]) \cap \mathbb{Z}^n$ are given by the following formulae: 
    \begin{enumerate}
        \item $\Delta_{S(I)} = 1$ for $I \subset [n]$ and $2 \leq |I| \leq n$, 
        \item $\Delta_{D(I)} = \det(Z_I)$ for $I \subset [n]$ and $2 \leq |I| \leq n$, and
        \item $\Delta_{T(I)} = \det(Z_{\{ 0\} \cup I})$ for $I \subset [n]$ and $2 \leq |I| \leq n$. When $|I| = 1$, we have $\Delta_{T(I)} = 1$.
    \end{enumerate}
    In particular, the $A_n$-discriminant equals $\det(Z)$.
\end{lemma}

\begin{proof}
    Point 1 follows from the fact that the $A$-discriminant of a standard simplex equals $1$, and point 2 is a well-known formula for the discriminant of a quadratic form in terms of its symmetric matrix, see \cite[Ex.~1.3 (b)]{gelfand2008discriminants}. For point three, it suffices to show the case $I = [n]$. First, note that $\det(Z) \neq 0$, as plugging in $z_{01} = z_{11} = \cdots = z_{nn} = 1/2$ and $0$ for all other $i,j$ gives $1/4$. It is also clear that $\Delta_{T([n])}$ divides $\det(Z)$: if $f_n$ has a singularity in the torus, the determinant must vanish, as it is the discriminant of the corresponding quadric (see point 2). To show equality, it suffices to check that the degree formula in \cite[Chpt.~9, Thm.~2.8]{gelfand2008discriminants} gives $\deg(\Delta_{T([n])}) = n+1$.
\end{proof}

We have the following immediate corollary of Lem.~\ref{lem:facesT} and \ref{lem:discrformula}.
\begin{theorem}\label{thm:detformula}
    The principal $A_n$-determinant 
    corresponding to the polytope $T([n])$ is
    \begin{equation} \label{eq:EAn} E_{A_n} \, = \, \textcolor{mycolor1}{\left (\prod_{i=1}^n z_{0i}z_{ii} \right )} \cdot 
    \textcolor{mycolor2}{\left ( \prod_{\substack{I \subset [n] \\ |I| = 2}} \Delta_{D(I)} \right )} \cdot
    \textcolor{mycolor5}{ \prod_{k=2}^{n-1} \left ( \prod_{\substack{I \subset [n] \\ |I| = k+1}}  \Delta_{D(I)}  \cdot \prod_{\substack{I \subset [n] \\ |I| = k} } \Delta_{T(I)} \right )} \cdot \textcolor{mycolor4}{\Delta_{T([n])}}.
    \end{equation}
    Here, the factors are sorted by increasing dimension of the corresponding face of $T([n])$.  Equivalently, this polynomial is the product of $\left (\prod_{i=1}^n z_{0i} \right)$ with all principal minors of $Z_{[n]}$, and all principal $(k+1)$-minors of $Z$ involving the index $0$, with $k \geq 2$. The degree~is 
    \[ \textcolor{mycolor1}{2n} + \textcolor{mycolor2}{2 \cdot \binom{n}{2}} + \textcolor{mycolor5}{\sum_{k=2}^{n-1} \left(\binom{n}{k+1} + \binom{n}{k} \right ) \cdot (k+1)} + \textcolor{mycolor4}{(n+1)} \, = \,  (n+1) \cdot (2^n-1).\]
\end{theorem}

Note that the right hand side in the last formula is $(n+1) \cdot{\rm vol}(T([n]))$, consistently with Tab.~\ref{tab:ngon}, which is the degree of the $A_n$-resultant. The next question is what happens when we substitute the coefficients of ${\cal G}_\An$ from \eqref{eq:UFoneloop} in the principal $A_n$-determinant. In what follows, we will indicate this substitution with a tilde: for any face $Q$ of $T([n])$, 
\[ \widetilde{\Delta}_{Q} \, = \, (\Delta_Q)_{|z_{0i} \gets 1, \,  z_{ii} \gets -\m_i, \, z_{ij} \gets s_{i,i+1,\ldots,j-1} - \m_i - \m_j, i \neq j},\]
and
\[Z({\cal E}) = Z_{|z_{0i} \gets 1, \,  z_{ii} \gets -\m_i, \, z_{ij} \gets s_{i,i+1,\ldots,j-1} - \m_i - \m_j, i \neq j}. \]
These are polynomials in the variables $s_{i,i+1,\ldots,j-1}$, $\m_i$. Recall that, when $j = i+1$, we set $s_{i} = \M_i$.
\begin{lemma}
    For all $I \subset [n]$ with $2 \leq |I| \leq n$, we have $\widetilde{\Delta}_{D(I)} \neq 0$ and $\widetilde{\Delta}_{T(I)} \neq 0$. Moreover, these are homogeneous polynomials with $\deg(\widetilde{\Delta}_{D(I)}) = \deg(\widetilde{\Delta}_{T(I)}) + 1 = |I|$.
\end{lemma}
\begin{proof}
For the sake of symplicity, we denote $Z({\cal E}) = \widetilde{Z}$. To show that $\widetilde{\Delta}_{D(I)} \neq 0$, by Lem.~\ref{lem:discrformula} it suffices to observe that $\widetilde{Z}_I$ is invertible for the choices $\m_i = 0$ and $s_{ij} = 1$ for all $i,j$. The statement $\deg(\widetilde{\Delta}_{D(I)}) = |I|$ follows easily from the fact that $\widetilde{Z}_I$ only involves coefficients of ${\cal F}_{{\tt A}_n}$, which are homogeneous of degree 1 in the parameters. \\
For $\widetilde{\Delta}_{T(I)}$, the same choice $\m_i = 0$ and $s_{ij} = 1$ works to show that $\widetilde{\Delta}_{T(I)} = \det(\widetilde{Z}_{\{0\}\cup I}) \neq 0$ and $\deg(\widetilde{\Delta}_{T(I)}) = |I|-1$ follows from the fact that $\det(\widetilde{Z}_{\{0\}\cup I})$ is homogeneous of degree 2 in the coefficients of ${\cal U}_{{\tt A}_n}$, and homogeneous of degree ${|I|} -1$ in those of ${\cal F}_{{\tt A}_n}$.
\end{proof}
\begin{theorem}
    Substituting the coefficients of ${\cal G}_\An$ in the principal $A_n$-determinant gives a nonzero polynomial in the kinematic variables of degree
\[ \textcolor{mycolor1}{n} + \textcolor{mycolor2}{2 \cdot \binom{n}{2}} + \textcolor{mycolor5}{\sum_{k=2}^{n-1} \left( \binom{n}{k+1} \cdot (k+1) + \binom{n}{k} \cdot (k-1) \right )} + \textcolor{mycolor4}{(n-1)}
= (n-1)\cdot 2^{n} + 1.\] 
Its square-free part defines the principal Landau determinant $E_{\An}({\cal K})$.
\end{theorem}

\begin{remark}
In the physics literature, $\Delta_{D(I)} = 0$ with $|I|=1$ are called mass singularities, those with $|I|=2$ are the normal and pseudo-normal thresholds, and those with $2 < |I| < n$ and $|I| = n$ are the subleading and leading Landau singularities (of the first type). Similarly, $\Delta_{T(I)} = 0$ with $1 < |I| < n$ and $|I|=n$ are the subleading and leading Landau singularities of the second type, respectively.
\end{remark}

\begin{example}\label{ex:3loop}
For $n=2$, we have
\be
Z({\cal K}) \,=\, \begin{pmatrix} 0 & 1 & 1 \\
1 & -2 \m_1 & s - \m_1 - \m_2 \\
1 & s - \m_1 - \m_2 & -2 \m_2
\end{pmatrix},
\ee
where $s := \M_1$. The principal Landau determinant $E_{{\tt A}_2}({\cal E})$ is given by the vanishing locus of the degree $5$ polynomial
\be
E_{{\tt A}_2}({\cal K})\, =\, \textcolor{mycolor1}{\m_1 \m_2} \textcolor{mycolor2}{\lambda(\m_1,\m_2,s)} \textcolor{mycolor4}{s},
\ee
where $\lambda$ is the K\"all\'en function \eqref{eq:Kallen}. 
\end{example}

\begin{example}
For $n=3$, we have
\be
Z({\cal K}) \,=\, \begin{pmatrix} 0 & 1 & 1 & 1\\
1 & -2 \m_1 & \M_1 - \m_1 - \m_2 & \M_3 - \m_1 - \m_3 \\
1 & \M_1 - \m_1 - \m_2 & -2 \m_2 & \M_2 - \m_2 - \m_3 \\
1 & \M_3 - \m_1 - \m_3 & \M_2 - \m_2 - \m_3 & -2\m_3
\end{pmatrix}.
\ee
The principal Landau determinant $E_{{\tt A}_3}({\cal E})$ is given by the vanishing locus of
\begin{align}
E_{{\tt A}_3}({\cal K}) &\,=\, \textcolor{mycolor1}{\m_1 \m_2 \m_3} \textcolor{mycolor2}{\prod_{i=1}^{3}\lambda(\m_i, \m_{i+1},\M_i)}\textcolor{mycolor5}{\M_1 \M_2 \M_3(\m_1^2 \M_2+\m_1 \M_2^2}\\
&\textcolor{mycolor5}{+\m_2 \m_1 \M_1- 
   \m_3 \m_1 \M_1-\m_2
   \m_1 \M_2-\m_3 \m_1 \M_2-\m_1 \M_1 \M_2-\m_2 \m_1
   \M_3+\m_3 \m_1 \M_3}\nn\\
&\textcolor{mycolor5}{-\m_1 \M_2 
   \M_3+\m_3 \M_1^2+\m_2
   \M_3^2+\m_3^2 \M_1-\m_2 \m_3 \M_1+\m_2 \m_3 \M_2-\m_3
   \M_1 \M_2+\m_2^2 \M_3}\nn\\
&\textcolor{mycolor5}{-\m_2 \m_3 \M_3- 
   \m_2 \M_1
   \M_3-\m_3 \M_1 \M_3-\m_2 \M_2 \M_3+\M_1 \M_2 \M_3)}\nn\textcolor{mycolor4}{\lambda(\M_1,\M_2,\M_3)},\nn
\end{align}
where the subscripts are taken modulo $3$. The above polynomial has degree $17$.
\end{example}

\begin{example}
For $n=4$, we further specialize to the equal-mass subspace ${\cal E}^{(\M,\m)}$ of ${\cal K}$ given by $\M_i = \M$, $\m_i = \m$, for which
\be
Z({\cal E}^{(\M,\m)}) \,=\, \begin{pmatrix} 0 & 1 & 1 & 1 & 1\\
1 & -2 \m & \M - 2\m & s - 2\m & \M - 2\m \\
1 & \M - 2\m & -2\m & \M - 2\m & t - 2\m \\
1 & s-2\m & \M - 2\m & -2\m & \M - 2\m \\
1 & \M - 2\m & t - 2\m & \M - 2\m & -2\m
\end{pmatrix}.
\ee
The principal Landau determinant on this subspace is given by the square-free part of the polynomial
\begin{align}
& \textcolor{mycolor1}{\m^4} \textcolor{mycolor2}{(\M - 4\m)^4 \M^4 (s-4\m)s (t - 4\m) t}
 \textcolor{mycolor5}{s^2\left(4 \m \M-\m s-\M^2\right)^2}\\
&\textcolor{mycolor5}{
t^2\left(4\m \M-\m t-\M^2\right)^2 \M^4 s t \cdot s t \left(16 \m \M-4 \m s-4 \m t-4 \M^2+s t\right)} \nn\\
&\textcolor{mycolor5}{\cdot s^2 (4 \M-s)^2 t^2(4 \M-t)^2}\textcolor{mycolor4}{s t (4\M - s - t)}.\nn
\end{align}
The degree is $49$.
\end{example}

\begin{remark}\label{rmk:0externalMasses}
Notice that, restricting to the subspace ${\cal E}^{(0,\m_e)}\subset{\cal K}$ where the external masses $\M_i$ vanish, does not change the monomial support of the Symanzik polynomials. However, the principal $A_n$-determinant identically vanishes when substituting the kinematics parameters in ${\cal E}^{(0,\m_e)}$, e.g., the minor \textcolor{mycolor4}{$\Delta_{T(\{1,2\})}$} is zero for all $n\geq 3$. This is consistent with the Euler characteristic being smaller than the volume in the third column of Tab.~\ref{tab:ngon}. 
\end{remark}

\subsubsection{Zero internal masses}
In this section, we compute the principal Landau determinant when restricting to the subspace ${\cal E}^{(\M_i,0)} \subset{\cal K}$ where the internal masses $\m_e = 0$ vanish. This assumption does not change the polynomial ${\cal U}_{{\tt A}_n}$, while the second Symanzik polynomial becomes
\begin{equation} \label{eq:UFoneloop0} 
{\cal F}_{{\tt A}_n}({\cal E}^{(\M_i,0)})\,=\, \sum_{i<j} s_{i,i+1,\dots,j-1} \, \alpha_i \alpha_j,
\end{equation}
with subscripts taken modulo $n$.
The Newton polytope of the polynomial ${\cal G}_{\An}({\cal E}^{(\M_i,0)})$ is the $(n-1)$-dimensional \textit{hypersimplex} in $\R^n$:
\begin{align*}
     \Delta_{n,2} \, &= \, \text{Conv}(\,e_i,\,e_i+e_j \; : \; i,j\in[n],\, i<j\,)\\
     &= \, \{(\alpha_1\dots,\alpha_n) \in [0,1]^n \, : \, 1\leq \alpha_1+\dots + \alpha_n \leq 2\}.
\end{align*}
Its $f$-vector is described in \cite[Cor.~4]{hibi2015face}. We write it explicitly for completeness:      
\begin{equation}\label{eq:fvec}
\Biggl( \binom{n+1}{2},\, \binom{n+1}{2}\cdot(n-1),\, \binom{n+1}{k+1}\cdot(n-k+1) \;\;\text{for}\,\, 2\leq k \leq n\Biggr).
\end{equation}
Notice that the hypersimplex $\Delta_{n,2}$ can be thought of as a slice of the unit hypercube $[0,1]^n$ by the hyperplanes $\sum \alpha_i = 1$ and $\sum \alpha_i = 2$. It is well-know that the normalized volume of the hypersimplex $\Delta_{n,2}$ equals
the Eulerian number $A_{n,1}$ \cite{lam2007alcoved}, consistently with the computation in Tab.~\ref{tab:ngon}. In what follows, we will compute the principal $A$-determinant for the hypersimplices $\Delta_{n,2}$, where the columns of the matrix $A_n({\cal E}^{(\M_i,0)})$ are precisely the vectors $\{\,e_i,\,e_i+e_j \, : \, i,j\in[n],\, i<j\}$. A general polynomial in the variables $\alpha = (\alpha_1,\dots,\alpha_n)$ supported on $A_n({\cal E}^{(\M_i,0)})$ is given by 
\begin{equation}\label{eq:f}
f = \sum _{0\leq i < j \leq n} z_{ij}\alpha_i\alpha_j,
\end{equation}
where we set $\alpha_0 = 1$ and the coefficients $z_{ij}$ can be seen as the entries of the symmetric matrix $Z$ in \eqref{eq:generalPolynomial} where all the diagonal entries are assumed to be zero. The following lemma describes the face structure of the hypersimplex $\Delta_{n,2}$.

\begin{lemma}\label{lem:hypersimplex}
  The zero- and one-dimensional faces of the polytope $\Delta_{n,2}$ are simplices. These account for the first two entries in the $f$-vector \eqref{eq:fvec}. The faces of dimension $2\leq k\leq n$ are:
  \begin{enumerate}
      \item $\binom{n+1}{k+1}$ hypersimplices of dimension $k$, and
      \item $\binom{n+1}{k+1}\cdot(n-k)$ simplices of dimension $k$.
  \end{enumerate}
\end{lemma}

\begin{proof}
    Our strategy consists in listing the faces of the polytope $\Delta_{n,2}$ in each dimension, accounting for numbers in the $f$-vector in \eqref{eq:fvec}. Since it does not change the face structure of the polytope, we work with homogeneous coordinates.
    We denote $\Delta'_{n,2} = \conv(e_i+e_j\,:\, i,j\in\{0,\dots,n\}, i<j),$ where $e_i\in \{0,1\}^{n+1}$ is the $i$-th basis vectors of $\Z^{n+1}$. Note that $\Delta_{n,2}'$ is the Newton polytope of $f$ in \eqref{eq:f}, when $\alpha_0$ is \emph{not} set to 1.
    Each face of $\Delta'_{n,2}$ is uniquely determined by the corresponding cone in the normal fan $\Sigma_{n,2}\subset \R^{n+1}$, and the weight vectors in each of the cones are considered modulo the linearity space spanned by $(1,1,\dots,1)$.\\
    The vertices are determined by the weights in $\Sigma_{n,2}$ of type $-(e_i+e_j)$ for $i,j\in\{0,\dots,n\}$ and $i<j$. Edges correspond to weights of type ${\bf w}_{j,I} = - e_j +\sum_{i\in I} e_i$, where $j\in \{0,\dots,n\}$, and $I\subset \{0,\dots,n\}\setminus\{j\}$ with $|I| = n-2$.

     For $2\leq k\leq n$, the $(n+1-k)-$dimensional cones come in two types. In terms of ray generators, these types are described by the vectors:
    \begin{enumerate}
    \item ${\bf w}_I = \sum_{i\in I} e_i$, where $I\subset \{0,\dots ,n\}$ with $|I| = n-k$. There are precisely $\binom{n+1}{k+1}$ many of such vectors. We have 
    $$\hbox{in}_{{\bf w}_I}(f)\,=\,\sum_{i,j\in J, i<j}z_{ij}\alpha_i\alpha_j,$$
    where $J = \{0,\dots,n\}\setminus I$. Its Newton polytope is a $k$-dimensional hypersimplex;
    \item ${\bf w}_{j,I}$ as above, with $|I| = n-k-1$. There are $\binom{n+1}{k+1}(n-k)$ such vectors. We have     $$\hbox{in}_{{\bf w}_{j,I}}(f) \, = \, \alpha_j \cdot \sum_{i\in \{0,\dots,n\}\setminus I} z_{ij}\alpha_i,$$
    whose Newton polytope is a $k$-dimensional simplex. \qedhere
    \end{enumerate}
\end{proof}

The following theorem follows from Lem.~\ref{lem:hypersimplex} and Thm.~1.2 in \cite[Chpt.~10]{gelfand2008discriminants}:

\begin{theorem}
    The principal $A_n({\cal E}^{(\M_i,0)})$-determinant
    corresponding to $\Delta_{2,n}$~is
    \begin{equation}\label{eq:Adiscr0}
    E_{A_n({\cal E}^{(\M_i,0)})} \, = \,\prod_{\substack{I\subset\{0,1,\dots,n\}\\3\leq|I|\leq n+1}}\det(Z_I) \,=\, \prod_{0\leq i < j \leq n} z_{ij}^{n-1} \cdot\prod_{\substack{I\subset\{0,1,\dots,n\}\\4\leq|I|\leq n+1}}\det(Z_I).
    \end{equation}
    Its degree is $(n-1)\cdot\binom{n+1}{2}\,+\,\sum_{k=4}^{n+1}k\cdot\binom{n+1}{k} \,=\, (n+1)\cdot(2^n-1-n)$.
\end{theorem}
\begin{proof}
The discriminant of a face hypersimplex determined by a subset $I\subset\{0,\dots,n\}$ of size $4\leq |I|\leq n+1$, as described in the proof of Lem.~\ref{lem:hypersimplex}, is given by $\det(Z_I)$, see \cite{helmer2018nearest}. The exponents $n-1$ of the factors $z_{ij}$ equal the subdiagram volume of any vertex of the hypersimplex, see \cite[Prop. 4.7]{helmer2018nearest}. Finally, the degree count shows that the degree of the $A_n({\cal E}^{(\M_i,0)})$-resultant agrees with $(n+1)\cdot \text{vol}(\Delta_{n,2})$, consistently with Tab. \ref{tab:ngon}. 
\end{proof}
We now investigate what happens when substituting the coefficients of the Symanzik polynomials in \eqref{eq:UFoneloop0} in the principal $A_n({\cal E}^{(\M_i,0)})$-determinant. The matrix $Z$ after substituting the parameters in the subspace ${\cal E}^{(\M_i,0)}$ is given by 
\[Z({\cal E}^{(\M_i,0)})\, = \, Z_{|z_{0i} \gets 1, \, z_{ij} \gets s_{i,i+1,\ldots,j-1}, i \neq j}, \]
where no substitution is required for the entries $z_{ii}$ since they are set to zero.
\begin{theorem}
Substituting the coefficients of ${\cal G}_\An({\cal E}^{(\M_i,0)})$ in the $A_n({\cal E}^{(\M_i,0)})$-determinant gives a polynomial in the kinematic variables of degree
\begin{align}
&(n-1)\cdot\binom{n}{2} \,\,+\,\, \sum_{k=4}^n k\cdot \binom{n}{k}\,\,+\,\, \sum_{k=3}^n (k-1)\cdot \binom{n}{k} \, \\
&= \, \frac{1}{2} [n(n-1)^2\,+\,(n-2)\cdot(2^n-1-n)\,+\,n(2^n-n^2+n-2)].\label{eq:degree}
\end{align}
 Its square-free part defines the principal Landau determinant variety ${\rm PLD}_{\An}({\cal E}^{(\M_i,0)})$.
\end{theorem}

\begin{proof}
Let $I\subset \{0,1,\dots,n\}$ with $3\leq |I|\leq n+1$. For simplicity, we denote $\widetilde{Z} = Z({\cal E}_\An^{(\M_i,0)})$. It is enough to observe that $\det (\tilde{Z}_I)\neq 0$ for the choices $s_{i,i+1,\dots,j-1} = 1$. To compute the degree we first observe that the contribution from the principal minors of type $\tilde{Z}_I$, where $0\notin I$ have degree $|I|$. Their contribution accounts for the second summand in \eqref{eq:degree}. The first summand instead comes from the minors of size 3. Finally, the last summand comes from the $\det(\tilde{Z}_{\{0\}\cup I})$, where $I\subset [n]$ and $3\leq|I|\leq n$. In particular, in these cases the degree of $\det(\tilde{Z}_{\{0\}\cup I})$ equals $|I|-1$.
\end{proof}

\begin{example}
For $n=4$, the principal Landau determinant of the graph ${\tt A}_4$ with coefficient matrix 
\be
Z({\cal E}^{(\M_i,0)})\, =\,
\begin{pmatrix} 0 & 1 & 1 & 1 & 1\\
1 & 0 & \M_1 & s & \M_4 \\
1 & \M_1 & 0 & \M_2 & t \\
1 & s & \M_2 & 0 & \M_3 \\
1 & \M_4 & t & \M_3 & 0
\end{pmatrix}
\ee
is given by the square-free part of the polynomial
\begin{align*}
& s^3 t^3 \M_1^3 \M_2^3 \M_3^3 \M_4^3\cdot \lambda(\M_3,\M_4,s)\cdot \lambda(\M_2,\M_3,t)\cdot\lambda(\M_1,\M_4,t)\cdot\lambda(\M_1,\M_2,s)\\
&\cdot(\M_1^2\M_3-\M_1\M_2\M_3+\M_1\M_3^2-\M_1\M_2\M_4+\M_2^2\M_4-\M_1\M_3\M_4-\M_2\M_3\M_4\\
&+\M_2\M_4^2+\M_1\M_2s-\M_1\M_3s-\M_2\M_4s+\M_3\M_4s-\M_1\M_3t+\M_2\M_3t\\
&+\M_1\M_4t-\M_2\M_4t-\M_1st-\M_2st-\M_3st-\M_4st+s^2t+st^2)\\
&\cdot\lambda(\M_1\M_3,\M_2\M_4,st).
\end{align*}
The degree is 33.
\end{example}

\subsubsection{Zero internal and external masses}
In this subsection, we further specialize the family of one-loop diagrams to the subspace ${\cal E}^{(0,0)}$ where both internal and external masses equal zero. In this case, the second Symanzik polynomial is given by
$${\cal F}_{{\tt A}_n}({\cal E}^{(0,0)})\, = \, \sum_{\substack{i<j,\\ j\neq i+1\\ j \neq i-1(\text{mod} \,n)}} s_{i,i+1,\dots, j-1} \alpha_i \alpha_j.$$
The Newton polytope of ${\cal G}_{\An}({\cal E}^{(0,0)})$ can be described as
$$P_n \, = \, \text{Conv}(\, e_i, e_i+e_j \,\, : \,\, (i,j)\in [n],\, i<j, \,j\neq i + 1  ,\, j\neq i-1(\text{mod}\, n)).$$
For $n\geq 4$, $P_n$ is a full-dimensional polytope in $\R^n$ with $\binom{n}{2}$ vertices. Furthermore, we extrapolated from computations for small $n$ that it is defined by the inequalities $$\alpha_i\geq 0, \;\;\, \alpha_i+\alpha_{i+1} \leq 1,\;\;\, \sum_{i\in[n]} \alpha_i \geq 1, \;\;\,\sum_{i\in [n]}\alpha_i \leq 2, $$
where the indices are considered modulo $n$. Notice that, in particular, the last hyperplane just appears for $n\geq 5$, accounting for $2n+1$ facets. We denote $A_n({\cal E}^{(0,0)})$ the matrix whose columns are the lattice points in $P_n\cap \Z^n$.
However, the complexity of the combinatorics of this polytope makes more complicated to develop an analogous discussion to determine the principal $A_n({\cal E}^{(0,0)})$-determinant corresponding to the polytope $P_n$. A general polynomial in the variables $\alpha = (\alpha_1,\dots,\alpha_n)$ with support in $A_n({\cal E}^{(0,0)})$ can be written as in \eqref{eq:f} where the matrix $Z$ must have all diagonal entries and all entries $z_{ij}$ with $j=i+1 (\text{mod} \, n)$ set to zero. We will write $\hat{Z}$ for such a matrix.
\begin{conjecture}\label{conj:A00}
The principal $A_n({\cal E}^{(0,0)})$-determinant variety corresponding to the polytope $P_n$ is the vanishing locus of the square-free polynomial
\begin{equation}\label{eq:A00}
    E_{A_n({\cal E}^{(0,0)})} \, = \, \prod_{\substack{I\subset\{0,1,\dots,n\}\\4\leq|I|\leq n+1}} \det(\hat{Z}_I).
\end{equation}
The intersection of the principal $A_n({\cal E}^{(0,0)})$-determinant variety with the subspace ${\cal E}^{(0,0)}$ is a hypersurface in ${\cal E}^{(0,0)}$. Its defining equation is the principal Landau determinant $E_\An({\cal E}^{(0,0)})$ and it can be attained by substituting the coefficients of the polynomials ${\cal G}_\An({\cal E}^{(0,0)})$ into \eqref{eq:A00}.
\end{conjecture}

Conj. \ref{conj:A00} presents a number of challenges, the main one being that the combinatorics of the polytope $P_n$ is hard to understand. A degree check indicates that we should expect nontrivial integer exponents, as in \eqref{eq:EA}, also for discriminants corresponding to faces of dimension greater than one. Computing such exponents would prove the formula for the number of master integrals in the last column of Tab.~\ref{tab:ngon}.

\subsection{Banana diagrams}
Let ${\tt B}_\E$ denote the banana diagram with $\E\geq 2$ internal edges, see Fig.~\ref{fig:schlegel} (left) for an example. We denote $\alpha=(\alpha_1,\dots,\alpha_\E)$, then the Symanzik polynomials are given by
\be
{\cal U}_{{\tt B}_\E}\, = \, \sigma_{\E-1}(\alpha),\qquad {\cal F}_{{\tt B}_\E} \,=\, s\prod_{e=1}^\E\alpha_e-\biggl(\sum_{e=1}^\E\m_e\alpha_e\biggr){\cal U}_{{\tt B}_\E},
\ee
where $\sigma_{\E-1}(\alpha)$ denotes the elementary symmetric polynomial of degree $\E-1$ in $\E$ variables. For example, $\U_{{\tt B}_3} = \alpha_1 \alpha_2 + \alpha_2 \alpha_3 + \alpha_1 \alpha_3$. We will denote $A_\E$ the matrix whose columns are the exponent vectors of the graph polynomial ${\cal G}_{{\tt B}_\E}.$ When choosing generic coefficients for the parameters $\m_e,s$ in the kinematic space ${\cal K}\subset\C^{\E+1}$, the computation of the Euler characteristic and the volume return different values. More precisely, we verified using \texttt{Julia} that for $2\leq \E\leq 10$ we have 
\be
|\chi(V_{{\tt B}_\E}({\cal K}))|\, = \, 2^\E-1, \qquad \hbox{vol}(A_\E({\cal K})) = \binom{2\E-1}{\E}.
\ee
These values for the Euler characteristic are proven, see \cite{Kalmykov:2016lxx,Bitoun:2017nre}.
When restricting to the subspace ${\cal E}^{(0,0)}\subset{\cal K}$, i.e., setting the internal and external masses to zero, the graph polynomial ${\cal G}_{{\tt B}_\E}({\cal E}^{(0,0)})$ is supported on the vertices of an $\E$-dimensional simplex. Therefore, we have $\hbox{vol}(A_\E({\cal E}^{(0,0)})) = 1.$
As mentioned in Ex.~\ref{ex:banana-nongeneric}, the signed Euler characteristic of a smooth very affine variety coincides with its maximum likelihood degree, see \cite{franecki2000gauss,huh2013maximum}. Furthermore, varieties with maximum likelihood degree one were geometrically characterized in \cite{huh2014varieties}. We will use that characterization to prove that the signed Euler characteristic $(-1)^\E\cdot\chi(V_{A_\E}({\cal E}^{(0,0)}))$ equals one for all $\E\geq 2$.
\begin{proposition}
If $s\neq 0$, the signed Euler characteristic of $V_{A_\E}({\cal E}^{(0,0)})$ is one: $(-1)^\E\cdot\chi(V_{A_\E}({\cal E}^{(0,0)}))= 1$. Moreover, the likelihood function $L = \alpha_1^{\nu_1}\cdots \alpha_\E^{\nu_\E} {\cal G}^\mu_{{\tt B}_\E}({\cal E}^{(0,0)})$ has one unique critical point given by 
\begin{equation}
    ((\E-1)\cdot\mu+\nu_1+\cdots+\nu_\E) \cdot \left(-\frac{1}{s(\mu+\nu_1)}, \, \ldots \, , -\frac{1}{s(\mu+\nu_E)}  \right).
\end{equation}
\end{proposition}

\begin{proof}
The proof is an application of \cite[Thm.~1 and Thm.~2]{huh2014varieties}. We let  
\be
A \, = \, \begin{pmatrix} 1 & 1  &\cdots & 1 & 1 \end{pmatrix} \in \mathbb{Z}^{1 \times (\E+2)}, \quad B \, = \, \begin{pmatrix}
1  & 1 & \dots & 1 & \E-1\\
-1 & 0 & \dots & 0 & -1 \\
0 & -1 & \dots & 0 & -1 \\
\vdots& \vdots & \ddots & \vdots & \vdots\\
0 & 0 & \dots & -1 & -1\\
0 & 0 & \dots & 0 & 1
\end{pmatrix} \in \Z^{(\E+2)\times(\E+1)}
\ee
and we choose the vector ${\bf d}=(1/s,\dots,1/s,1/s^{\E-1})$. Plugging these data into \cite[Thm.~2]{huh2014varieties}, which uses the same notation, proves that $(-1)^\E\cdot\chi(V_{A_\E}({\cal E}^{(0,0)}))= 1$.
The critical points of the function $L$ are determined via the map in \cite[Thm.~1]{huh2014varieties}.
\end{proof}

\section{Computing principal Landau determinants} \label{sec:algorithm}

Our definition of the principal Landau determinant in Sec.~\ref{subsec:definition} hints at an algorithm for computing it via elimination of variables. 
Standard methods for this are based on Gr\"obner bases. They are implemented, for instance, in the software package \texttt{Oscar.jl} \cite{OSCAR}. The advantage of such methods is that they return the \emph{exact} answer. That is, they are guaranteed to return the principal Landau determinant with exact, integer coefficients. In large examples, such as those illustrated in Fig.~\ref{fig:diagrams}, it is not feasible to use these methods. We then resort to a numerical sampling algorithm that attempts to reconstruct the principal Landau determinant using homotopy continuation methods. This goes a long way in tackling such larger cases, and gives reliable answers in practice. It is a generalization of the strategy used to compute Landau discriminants in \cite{Mizera:2021icv}.

These algorithms are implemented in a \texttt{Julia} package \texttt{PLD.jl}, publicly available at
\begin{center}
\PLDwebsite.
\end{center}
The website contains a database of diagrams, the source code, and a tutorial exemplifying its use.
For instance, for the parachute diagram, the principal Landau determinant is computed as follows: 
\begin{minted}{julia}
edges = [[3,1],[1,2],[2,3],[2,3]];
nodes = [1,1,2,3];
getPLD(edges, nodes, internal_masses = :generic,
                     external_masses = :generic)
\end{minted}
Here, \texttt{edges} and \texttt{nodes} encode the diagram in the same format as in \cite{Mizera:2021icv}: each vertex is assigned a number; \texttt{edges} is the list of pairs of vertices that are connected by internal edges; \texttt{nodes} is the list of vertices to which we attach external momenta $p_1, p_2, \ldots$. The internal masses $m_1, m_2, \ldots$ and Schwinger parameters $\alpha_1, \alpha_2, \ldots$ are assigned in the same order as they appear in \texttt{edges}. For example, above the vertex $1$ has the momentum $p_1 + p_2$, vertex $2$ has $p_3$, and vertex $3$ has $p_4$. There are internal edges connecting $3$ to $1$ with mass $m_1$, $1$ to $2$ with mass $m_2$, and $2$ to $3$ twice with masses $m_3$ and $m_4$. The resulting diagram is $G = \texttt{par}$ from Fig.~\ref{subfig:par}.

When masses are set to \texttt{:generic}, the function \texttt{getPLD} automatically assigns distinct variables to the internal masses squared $\texttt{m}_{\texttt{1}}, \texttt{m}_{\texttt{2}}, \ldots$, and external masses squared $\texttt{M}_{\texttt{1}}, \texttt{M}_{\texttt{2}}, \ldots$, as well as Mandelstam invariants $\texttt{sij\ldots} = (p_i + p_j + \ldots)^2$ in a cyclic basis (for $n=4$, the basis consists of $\texttt{s} = (p_1 + p_2)^2$ and $\texttt{t} = (p_2 + p_3)^2$). Other options of \texttt{getPLD} are explained in detail in Ex.~\ref{ex:getPLD} and in the tutorial at \PLDwebsite.

The output of the above command gives the results described in Sec.~\ref{subsec:beyondstandard}. For example, a few lines of the output are
\begin{minted}{julia}
codim: 3, face: 1/33, weights: [-1, 0, 0, 0], discriminant: m?$_1$?
codim: 3, face: 2/33, weights: [-1, 0, -1, 0], discriminant: 1
codim: 3, face: 3/33, weights: [0, 1, 1, 2], discriminant: 1
codim: 3, face: 4/33, weights: [-1, -1, 0, 1], discriminant: m?$_1$?^2 - 2*m?$_1$?*m?$_2$?
                                             - 2*m?$_1$?*s + m?$_2$?^2 - 2*m?$_2$?*s + s^2
\end{minted}
A verbose version of the output also prints other information, for example, whether a given face has dominant components (UV/IR divergences), $\alpha$-positive solutions, etc.

\subsection{Symbolic elimination} \label{subsec:symbolic}

Formally, the problem of elimination of variables can be phrased as follows. 
\begin{center}
    \textit{Given a set of generators $f_i(\alpha,z), i = 1, \ldots, m$ for an ideal $I \subset \mathbb{Q}[\alpha_1, \ldots, \alpha_n,z_1, \ldots, z_s]$, compute a set of generators of the \emph{elimination ideal} $I_z = I \cap \mathbb{Q}[z_1, \ldots, z_s]$.}
\end{center}
This is the algebraic version of \emph{coordinate projection}. More precisely, let $Y = V(I)$ be the affine variety defined by $I$ in $\mathbb{C}^n \times \mathbb{C}^s$, and let $\pi_z(Y) \subset \mathbb{C}^s$ be its image under the coordinate projection $(\alpha, z) \mapsto z$. The variety of the elimination ideal $I_z$ is the Zariski closure of $\pi_z(Y)$ \cite[Chpt. 3, \S 2]{cox2013ideals}. Notice that, although we are interested in complex algebraic varieties, we assume in this section that the equations are defined over $\mathbb{Q}$. This is necessary in order to manipulate the generators symbolically. A standard symbolic tool for elimination is a \emph{Gr\"obner basis} of $I$ \cite[Chpt. 2]{cox2013ideals}. Many computer algebra systems, including \texttt{Oscar.jl} \cite{OSCAR}, offer an implementation.

For our purposes, the variety $Y$ is the incidence variety in \eqref{eq:incidence-variety}. Recall that ${\cal E} \simeq \mathbb{C}^s$ is a linear subspace of the kinematic space ${\cal K}$ associated to a diagram $G$. Our variety lives in $(\mathbb{C}^*)^\E \times \mathbb{C}^s$, where $\mathbb{C}^* = \mathbb{C} \setminus \{0\}$. To enforce nonzero $\alpha$-coordinates, we add a new variable $y$ and impose the condition $y\cdot \alpha_1 \cdots \alpha_\E - 1 = 0$. This gives an ideal $I_{G,Q} \subset \mathbb{Q}[\alpha,y,z]$ generated by $\E + 2$ equations: the graph polynomial ${\cal G}_{G,Q}$ restricted to a face $Q$, its partial derivatives, and this extra equation. Explicitly,
\begin{equation} \label{eq:IGQ} I_{G,Q} \, = \, \langle {\cal G}_{G,Q}, \, \partial_\alpha {\cal G}_{G,Q}, \, y\cdot \alpha_1 \cdots \alpha_{\E} - 1 \rangle. \end{equation}
The variety $V(I_{G,Q}) \subset \mathbb{C}^{\E+1} \times {\cal E}$ is isomorphic to $Y_{G,Q}({\cal E})$, and it has the same projection onto ${\cal E}$. Rather than computing this projection, we need to do this for each irreducible component $Y^{(i)}_{G,Q}({\cal E})$ (see Sec.~\ref{subsec:definition}). Since we are only interested in the variety of $I_{G,Q}$, not in any non-reduced scheme structure, the Nullstellensatz allows us to replace $I_{G,Q}$ with its radical \cite[Chpt.~4]{cox2013ideals}. We will therefore assume that $I_{G,Q}$ is a radical ideal from now on. The ideals of the individual components $Y^{(i)}_{G,Q}({\cal E})$ are found from $I_{G, Q}$ using \emph{primary decomposition} \cite[Chpt.~4, \S 8]{cox2013ideals}. This too is implemented in most computer algebra software systems, and returns a list of ideals $I^{(j)}_{G,Q}$ such that 
\begin{equation} \label{eq:primarydecomp} I_{G,Q} \, = \, \bigcap_{i \in \mathbb{I}(G,Q)} I^{(i)}_{G,Q},
\end{equation}
and the ideals $I^{(j)}_{G,Q}$ are \emph{prime}. In particular, their varieties are irreducible. Note that \eqref{eq:primarydecomp} is the algebraic counterpart of the irreducible decomposition \eqref{eq:irreddecomp}.

We are now ready to eliminate: we compute the elimination ideals $(I_{G, Q}^{(i)})_z$ for $i \in \mathbb{I}(G,Q)$. Among the results, the ideals that have codimension 1 are the nonzero principal ideals. We add their generator to the list of factors $\Delta_{G,Q}^{(i)}({\cal E})$ of $E_{G}({\cal E})$. 

This method serves a more general purpose than computing principal Landau determinants. Namely, we solve the following elimination problem. Let $Y^{(i)}$, $i \in \mathbb{I}$ be the irreducible components of 
\begin{equation} \label{eq:incidenceY}
    Y \, = \, \{ (\alpha, z) \in \mathbb{C}^n \times \mathbb{C}^s \, : \, f_1(\alpha,z) = \cdots = f_m(\alpha,z) = 0 \}.
\end{equation}
The map $\pi_z: (x,z) \mapsto z$ projects onto $z$-space. We compute the unique (up to scale) defining equation $\Delta^{(i)} \in \mathbb{Q}[z]$ of the projection $\pi_z(Y^{(i)})$, for all $i \in \mathbb{I}$ such that this projection has codimension 1. In analogy with our notation in Sec.~\ref{subsec:definition}, we will denote this subset of indices by $\mathbb{I}_1 \subset \mathbb{I}$. We demonstrate this on a running example. 

\begin{example}[$n = 1, s = 2$ and $m = 4$] \label{ex:running}
We consider four polynomials in $\mathbb{Q}[\alpha,z_1,z_2]$:
\begin{align}
     f_1 \, = \, & (z_1^2 + z_2^2 - 3 \alpha)   (z_1^2 \alpha - 2 z_1 z_2^2 + 4 z_1 z_2 \alpha + 2 z_1 z_2 + z_1 \alpha^2  - 3 z_1 \alpha - 4 z_2^3 - z_2^2 \alpha \\ 
     & + 6 z_2^2 + 2 z_2 \alpha^2 - z_2 \alpha - 2 z_2 - 2 \alpha^2 + 2 \alpha)   (z_1^2 - z_2 + 2), \\ 
f_2 \, = \, & (z_1^2 + z_2^2 - 3 \alpha)   z_2   (z_1 + z_2 + \alpha - 1)   (z_1^2 - z_2 + 2)   (z_1 - z_2 + \alpha - 1), \\ 
f_3 \,  = \, & (z_1^2 + z_2^2 - 3 \alpha)   (z_1 + z_2 + \alpha - 1)   (z_1 - 2)   (z_1^2 - z_2 + 2)   (z_1 - z_2 + \alpha - 1),\\ 
f_4 \, = \, & (z_1^2 + z_2^2 - 3 \alpha)   (2 z_1^2 \alpha^2 - z_1^2 \alpha - 4 z_1^2 + 6 z_1 z_2^2 - 4 z_1 z_2 \alpha - 10 z_1 z_2 - 3 z_1 \alpha^2 - z_1 \alpha \\
& + 12 z_1 + 2 z_2^2 \alpha^2 + z_2^2 \alpha - 6 z_2^2 - 4 z_2 \alpha^2 + z_2 \alpha + 14 z_2 - 2 \alpha^2 + 6 \alpha - 8)   (z_1^2 - z_2 + 2).
\end{align}
The real part of the incidence variety $Y = \{f_1 = f_2 = f_3 = f_4 = 0\}$ is shown in Fig. \ref{fig:runningex}. The ideal is constructed in \texttt{Oscar.jl} as follows:
\begin{minted}{julia}
R, vrs = PolynomialRing(QQ,["z1";"z2";"α"]); z1, z2 = vrs[1:2]; α = vrs[3];
Q1 = z1^2 + z2^2 - 3*α; Q2 = z1^2 - z2 + 2;
f1 = Q1*(z1^2*α - 2*z1*z2^2 + 4*z1*z2*α + 2*z1*z2 + z1*α^2 - ... + 2*α)*Q2
f2 = Q1*z2*(z1 + z2 + α - 1)*Q2*(z1 - z2 + α - 1)
f3 = Q1*(z1 + z2 + α - 1)*(z1 - 2)*Q2*(z1 - z2 + α - 1)
f4 = Q1*(2*z1^2*α^2 - z1^2*α - 4*z1^2 + 6*z1*z2^2 - 4*z1*z2*α - ... - 8)*Q2
I = ideal(R,[f1;f2;f3;f4])
\end{minted}
It turns out that this ideal is radical. This can be verified via
\begin{minted}{julia}
I == radical(I)
\end{minted}
We now compute its primary decomposition. This is done with the command
\begin{minted}{julia}
PD = primary_decomposition(I)
\end{minted}
We find that $I = I_1 \cap I_2 \cap I_3 \cap I_4 \cap I_5$, with 
\[ \begin{matrix} 
I_1 = \langle \alpha + z_1 + z_2 -1, z_1^2 + z_2^2 -1 \rangle, \quad I_2 = \langle z_1-2, z_2 \rangle, \quad I_3 = \langle 3\alpha -z_1^2-z_2^2 \rangle, \\ 
I_4 = \langle \alpha + z_1 - z_2-1, (z_1-1)^2+(z_2-1)^2 -1 \rangle, \quad I_5 = \langle z_2-z_1^2-2 \rangle.
\end{matrix}
\]
Each of the ideals $I_j$ contributes an irreducible component of the variety $Y$ in Fig. \ref{fig:runningex}. We eliminate $\alpha$ from each $I_j$ and record the generators of the nonzero principal ideals: 
\begin{minted}{julia}
PLD = []
for i = 1:length(PD)
      E = eliminate(PD[i][1],[α])
      if length(gens(E)) == 1 && gens(E)[1] !=0
            push!(PLD, gens(E)[1])
      end
end
\end{minted}
The result consists of three equations $\Delta^{(1)}, \Delta^{(4)}, \Delta^{(5)}$: 
\[ \Delta^{(1)} \, = \, z_1^2 + z_2^2 -1, \quad \Delta^{(4)} \, = \, (z_1-1)^2 + (z_2-1)^2 - 1, \quad \Delta^{(5)} \, = \, z_2-z_1^2-2.\]
Here $\Delta^{(i)}$ is the defining equation of the projection of $Y^{(i)} = V(I_i)$ onto $(z_1,z_2)$-space, see Fig. \ref{fig:runningex}. The components $Y^{(2)}$ and $Y^{(3)}$ do not contribute, because their projections have codimension 2 and 0, respectively. 
\end{example}

\begin{figure}
\centering
\includegraphics[scale=1]{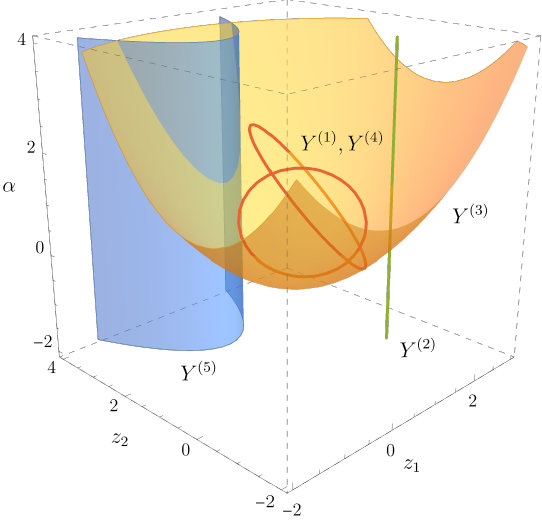}
\caption{The incidence variety from Ex. \ref{ex:running}.}
\label{fig:runningex}
\end{figure}

The above operations of computing the radical of $I_{G, Q}$, then its primary decomposition, and then performing elimination of variables for each prime component tend to be quite costly in practice. In larger examples, it is necessary to come up with a numerical alternative. The resulting algorithm will compute the factors $\Delta_{G,Q}^{(i)}({\cal E})$ without ever computing the generators of $I^{(j)}_{G,Q}$. The idea is to use \emph{sampling} and \emph{interpolation} instead.

\subsection{Numerical elimination} \label{sec:numerical}

We will present our numerical strategy in the general setup introduced in the previous section. That is, $Y \subset \mathbb{C}^n \times \mathbb{C}^s$ is defined by $m$ equations $f_i(\alpha,z) \in \mathbb{Q}[\alpha,z]$. It is then straightforward to specialize the discussion to principal Landau determinants. Like before, our aim is to compute the unique (up to scale) defining equation $\Delta^{(i)} \in \mathbb{Q}[z]$ of the projection $\pi_z(Y^{(i)})$, for all $i \in \mathbb{I}_1$. In \texttt{PLD.jl}, this computation is executed by the function \texttt{project\_codim1}. Since this method uses numerical computations, the input is now in the format of the package \texttt{HomotopyContinuation.jl} \cite{10.1007/978-3-319-96418-8_54}, rather than \texttt{Oscar.jl} \cite{OSCAR}. Here is an example code snippet based on our running Ex. \ref{ex:running}.
\begin{minted}{julia}
using HomotopyContinuation; @var z1 z2 α;
Q1 = z1^2 + z2^2 - 3*α; Q2 = z1^2 - z2 + 2;
f1 = ... # The code for generating f1, f2, f3, f4 is the same as above.
f = [f1;f2;f3;f4]
Δ, samples, gaps = project_codim1(f,[z1;z2], α; homogeneous = false)
\end{minted}
In general, the inputs are \texttt{f}, the list of $m$ equations $f_i$, and two seperate lists of variables, the second of which should be eliminated. We often call \texttt{z} the \emph{parameters}, and \texttt{α} the \emph{variables}. The option \texttt{homogeneous = false} indicates that we do not expect the factors of the PLD to be homogeneous polynomials in this example. It plays a role in numerical interpolation, which is part of the algorithm (see below). 

 For ease of exposition, we assume that the vanishing ideal of each irreducible component $Y^{(i)}$ with $i \in \mathbb{I}_1$, is a prime component of $I = \langle f_1, \ldots, f_m \rangle \subset \mathbb{C}[\alpha,z]$. If this is not the case, $I$ has higher multiplicity along the component $Y^{(i)}$, and our current implementation will not find the defining equation $\Delta^{(i)}$ of $\pi_z(Y^{(i)})$. The difficulty with such non-reduced components is essentially that it is hard to compute points on them numerically. To overcome this issue, new methodology in numerical algebraic geometry is needed, which is beyond the scope of our article. In theory, and sometimes in practice, our assumption can be realized by replacing the ideal $I$ with its radical. The strategy of \texttt{project_codim1} can be summarized as follows: 
\begin{enumerate}
    \item[(i)] Compute a set of sample points $\tilde{S} \subset \mathbb{C}^n \times \mathbb{C}^s$ on the incidence variety $Y$ from \eqref{eq:incidenceY}, so that all points in $\tilde{S}$ lie on a component $Y^{(i)}$ with ${\rm codim}(\pi_z(Y^{(i)})) \leq 1$.
    \item[(ii)] Filter out points that lie on a component which projects dominantly to $z$-space, i.e., ${\rm codim}(\pi_z(Y^{(i)})) = 0$. Call the remaining set of samples $S$.
    \item[(iii)] Divide $S$ into groups according to which component $Y^{(i)}$ they lie on. Let $S^{(i)} \subset S$ be the set of samples on $Y^{(i)}$.
    \item[(iv)] Deduce the degree of $\pi_z(Y^{(i)})$ from the corresponding group of samples $\pi_z(S^{(i)})$. 
    \item[(v)] Sample sufficiently many more points on each component $\pi_z(Y^{(i)})$ to find unique interpolating polynomials $\Delta^{(i)}$. 
\end{enumerate}
We now explain each of these steps in more detail.

\paragraph{Step (i)} We choose a general line $L \subset \mathbb{C}^s$ in $z$-space. It is clear that $\pi_z^{-1}(L) \cap Y^{(i)} \neq \emptyset$ if and only if ${\rm codim}(\pi_z(Y^{(i)}) )\leq 1$. 
If $\pi_z(Y^{(i)})$ has codimension 1 in $\mathbb{C}^s$ and the restriction of $\pi_z$ to $Y^{(i)}$ has finite fibres, i.e., $i \in \mathbb{I}_1$, this nonempty intersection consists of finitely many points. These are regular solutions to the system of equations
\begin{equation} \label{eq:addline}
f_1(\alpha;z) \, = \, \cdots \, = \, f_m(\alpha;z) \, = \, 0 \quad \text{and} \quad z \in L. 
\end{equation}
Here \emph{regular} means that the Jacobian matrix of this system of equations, evaluated at a point in $\pi_z^{-1}(L) \cap Y^{(i)}$, has rank $n + s$.
Components which project dominantly to $\mathbb{C}^s$ under $\pi_z$ contribute positive dimensional solution sets to \eqref{eq:addline}. At each point on these positive dimensional components, the Jacobian matrix has rank at most $n + s -1$. We conclude that all regular solutions to \eqref{eq:addline} lie on a component $Y^{(i)}$ for which $i \in \mathbb{I}_1$. We solve \eqref{eq:addline} using \texttt{HomotopyContinuation.jl}, and add all regular solutions to $\tilde{S}$. The presence of positive dimensional components will lead to singular solutions in the~output. 

However, it might happen that not all indices $i \in \mathbb{I}_1$ are accounted for. For instance, in the situation of Fig. \ref{fig:runningex}, the preimage $\pi_z^{-1}(L)$ of a general line $L$ does not intersect $Y^{(2)}$, as desired, and it makes a one-dimensional intersection with $Y^{(3)}$. The intersections with $Y^{(1)}$ and $Y^{(4)}$ consist of 2 points each, and these are regular solutions to \eqref{eq:addline}. This gives four points in $\tilde{S}$. However, in that example, $\pi_z(Y^{(5)})$ also has codimension 1, and $\pi_z^{-1}(L) \cap Y^{(5)}$ is positive dimensional. The reason is that the restriction of $\pi_z$ to $Y^{(5)}$ has 1-dimensional fibres. We need to make sure that such components are not missed.

Our algorithm proceeds as follows. If \eqref{eq:addline} has singular solutions, let $h_1(\alpha) = c_0 + c_1 \alpha_1 + \cdots + c_n \alpha_n$ be a random affine linear polynomial in the $\alpha$ variables. We cut down the dimension of the fibres of $\pi_z$ by adding $h_1$ to the system: 
\begin{equation} \label{eq:addline-addh}
f_1(\alpha;z) \, = \, \cdots \, = \, f_m(\alpha;z) \, = \, h_1(\alpha) \, = \, 0 \quad \text{and} \quad z \in L. 
\end{equation}
The regular solutions to this system may lie on two types of componentes $Y^{(i)}$. Either $i \in \mathbb{I}_1$ and $\pi_z$ has 1-dimensional fibres restricted to $Y^{(i)}$, or $Y^{(i)}$ projects dominantly to $\mathbb{C}^s$, with zero-dimensional fibres. Clearly, the first group of samples is the one we want to keep. Filtering out spurious samples is taken care of in step (ii). For now, we add all regular solutions to \eqref{eq:addline-addh} to $\tilde{S}$. If \eqref{eq:addline-addh} has singular solutions, there might be indices $i \in \mathbb{I}_1$ for which $Y^{(i)} \rightarrow \pi_z(Y^{(i)})$ has 2-dimensional fibres, so we add another random equation $h_2(\alpha) = 0$ to \eqref{eq:addline-addh}. We solve, retain regular solutions, and proceed in this manner until no singular solutions are found.  

\paragraph{Step (ii)} To check whether $p = (\alpha^*,z^*) \in \tilde{S} \subset Y$ lies on a component which projects dominantly to $\mathbb{C}^s$, it is enough to check whether its tangent space $T_pY$ projects dominantly to $\mathbb{C}^s$. This is the kernel of the Jacobian matrix of $f_1 = \cdots = f_m = 0$, evaluated at $p$. Computing the dimension of the projection of this linear space is an elementary task from linear algebra. We apply this test to all points in $\tilde{S}$, and discard those points which lie on a dominant component. The remaining set is $S$.

\paragraph{Step (iii)} Decomposing $S$ into groups $S^{(i)}$ according to the irreducible components of $Y$ can be done using monodromy loops in \texttt{HomotopyContinuation.jl}. This was applied for Landau discriminants in \cite[Sec.~3.2]{Mizera:2021icv}, to which we refer for more details. 

\paragraph{Step (iv)} Notice that a general line $L \subset \mathbb{C}^s$ intersects $\nabla^{(i)} = \overline{\pi(Y^{(i)})}$ in $\deg (\nabla^{(i)})$-many points, for $i \in \mathbb{I}_1$. Hence, the degree of this hypersurface is the number of distinct points in $\pi_z(S^{(i)})$. For simplicity, we write $d_i = \deg(\nabla^{(i)}) = \deg(\Delta^{(i)})$. 

\paragraph{Step (v)} Finally, we gather more samples on each component $\nabla^{(i)}$, $i \in \mathbb{I}_1$. The minimal number of samples needed is the number of monomials that may appear in the equation $\Delta^{(i)}$, minus one. Indeed, vanishing on a sample imposes one linear condition on the polynomial $\Delta^{(i)}$, and the total amount of parameters is the dimension of the space of polynomials in $z_1, \ldots, z_m$ of degree $d_i$. The formula is $\binom{m-1+d_i}{d_i}$ if \texttt{homogeneous = true}, and $\binom{m+d_i}{d_i}$ if \texttt{homogeneous = false}. The way we collect more samples on $Y^{(i)}$ is as follows. We pick a new line $L' \subset \mathbb{C}^s$, and move our initial line $L$ continuously towards $L'$. Along the way, we use homotopy continuation to track the points in $S^{(i)}$ to a new set of samples $(S^{(i)})'$. In such a homotopy, the sample points cannot leave their irreducible component $Y^{(i)}$. We repeat this many times, until the necessary amount of samples is achieved. For more details and numerical considerations, see \cite[Sec.~3.2]{Mizera:2021icv}. 

After we have used numerical interpolation to obtain floating point approximations of the coefficients of $\Delta^{(i)}$, we may use \emph{rationalization} to recover the exact coefficients in $\mathbb{Q}$. This assumes the original equations $f_i$ were defined over $\mathbb{Q}$, as is the case for principal Landau determinants. We used the command \texttt{rationalize} in \texttt{Julia} for this.

The function \texttt{project\_codim1} returns a list of polynomials $\Delta$ with rational coefficients. The second output contains, for each component $\Delta^{(i)}$, the list of samples used to interpolate it. The third output, \texttt{gaps}, gives an indication of the quality of the numerical interpolation for each component. More precisely, for each component, it records the ratio between the second smallest and the smallest singular value of the linear interpolation problem. As a rule of thumb, when $\mathtt{gap} = 10^e$ for component $\Delta^{(i)}$, one can expect that its coefficients were approximated with $e$ accurate digits. As a tolerance for rationalization, we use $10^{-8}$ as a default. This means that if $e<8$, the answer should probably not be trusted. This was never the case in our experiments.

\subsection{Overall algorithm}
This subsection presents a pseudocode that illustrates the key steps of the algorithm behind the function \texttt{getPLD} in \texttt{PLD.jl}, which was shown in the introduction to this section. The definition of the principal Landau determinant requires to perform elimination on the ideal $I_{G,Q}$ from \eqref{eq:IGQ}, for each face $Q$ of ${\rm Newt}({\cal G}_G)$. Alg. \ref{algorithm} takes this into account, and it summarizes all the steps explained in the previous subsections.  

\begin{algorithm}[t] 
\caption{The algorithm to compute the principal Landau determinant}\label{algorithm}
\begin{algorithmic}[1] 
\Input The vectors of nodes and edges encoding a Feynman diagram $G$, a subspace $\cal{E}$ of the kinematic space $\cal{K}$, a choice of method (\texttt{:sym} for symbolic, \texttt{:num} for numerical).
\Output A vector \texttt{discs} of specialized discriminants $\Delta_{G,Q}^{(i)}({\cal E})$ that constitute the factors of the principal Landau determinant $E_G({\cal E})$.
\Statex
\For{each face $Q$ of $\text{Newt}({\cal G}_G({\cal E}))$}
    \State Compute the initial form ${\cal G}_{G,Q}({\cal E})$;
    \State Define $Y_{G,Q}$ by \texttt{eqs} $:=\{ \mathcal{G}_{G,Q} =\partial_{\alpha} \mathcal{G}_{G,Q} = y\cdot\alpha_1\cdots\alpha_\E-1=0\}$;
    \State \texttt{discs} $:= \texttt{[\,\,]}$;
    \If{\textcolor{mycolor1}{\texttt{method == :sym}}}
        \State $I_{G,Q} := \,$\texttt{radical}(\texttt{ideal}(\texttt{eqs})); 
        \State $\texttt{PD} := \,$\texttt{primary_decomposition}$(I_{G,Q})$;
        \For{$I_{G,Q}^{(i)}$ in  $\texttt{PD}$}
            \State $J := \,$\texttt{eliminate}$(I_{G,Q}^{(i)},[\alpha,y])$;
            \ForAll{principal, codimension 1, non-zero ideals $J^{(i)}$ in $J$}
                    \State Add generator of $J^{(i)}$ to \texttt{discs};
            \EndFor
        \EndFor
    \ElsIf{\textcolor{mycolor1}{\texttt{method == :num}}}
     \State $\Delta_{G,Q}, \,\texttt{_}\,, \texttt{gaps} := \texttt{project_codim1}(\texttt{eqs},\texttt{pars},[\alpha,y])$;
        \ForAll{components $\Delta_{G,Q}^{(i)}$ in $\Delta_{G,Q}$ with large $\texttt{gaps}^{(i)}$}
            \State Add $\Delta_{G,Q}^{(i)}$ to \texttt{discs};
        \EndFor
    \EndIf
\EndFor
\State \textbf{return} \texttt{discs};
\end{algorithmic}
\end{algorithm}

A tutorial on how to use \texttt{PLD.jl} can be found at \PLDwebsite. We here present an example to illustrate some of the optional inputs of the main function, which makes it very easy to use. 

\begin{example}\label{ex:getPLD}
The following code snippet runs \texttt{getPLD} for the diagram $G = $ \texttt{outer-dbox}:
\begin{minted}{julia}
@var m2;
edges = [[1,2],[2,3],[3,4],[4,5],[5,6],[6,1],[3,6]];
nodes = [1,2,4,5];
getPLD(edges, nodes, internal_masses = [m2,m2,m2,m2,m2,m2,0],
                     external_masses = :zero,
                     method = :sym,
                     high_prec = false,
                     codim_start = -1,
                     face_start = 1,
                     single_face = false,
                     single_weight = nothing,
                     verbose = true,
                     homogeneous = true,
                     save_output = "",
                     load_output = "")
\end{minted}
Only \texttt{edges} and \texttt{nodes} are required arguments.
The options \texttt{internal_masses} and \texttt{external_masses} can be set to either \texttt{:generic}, \texttt{:equal}, \texttt{:zero} (the default), or a custom list of variables. In the first three cases, the variables are assigned automatically.
Above, we set all the internal masses to be equal to \texttt{m}, except for the edge \texttt{[3,6]}, which is massless, and all the external masses are set to zero.

The \texttt{method} can be either \texttt{:sym} (the default) or \texttt{:num}.
Here, we use symbolic elimination. The optional input \texttt{high\_prec} is only meaningful when using \texttt{method = :num}. It tells the program to do the numerical interpolation in higher precision. This is slower, but it might be necessary for exact rational reconstruction. It was not needed for the examples we computed, listed at \PLDwebsite, but it was crucial to find some of the larger discriminants in \cite{Mizera:2021icv}. The option \texttt{verbose = true} makes sure that intermediate results are printed, and \texttt{homogeneous = true} means we expect all factors of the PLD to be homogeneous polynomials. This helps for interpolation, see Sec.~\ref{sec:numerical}.

The next important options are \texttt{codim_start} and \texttt{face\_start}. Our program runs through the faces of ${\rm Newt}({\cal G}_G)$ in order of decreasing codimension. That is, it runs through all vertices first, then all edges, etc. If a computation was interrupted, or one is for other reasons only interested in the result for a subset of all faces, it is convenient to skip low dimensional faces. The faces in each codimension are numbered according to the output of the function \texttt{getWeights}, which returns, for each codimension, a list of weight vectors revealing the faces. This is illustrated in Ex.~\ref{ex:weights}. The function \texttt{getPLD} will start from face number \texttt{face_start} in codimension \texttt{codim_start}. If \texttt{single_face} is true, it will do the computation only for that face. Similarly, if \texttt{single_weight} is set to a weight list, the computation will happen only for the face with this specific weight. If \texttt{codim_start = -1} (the default option), the program will run over all faces. The output for each face can be saved by setting the variable \texttt{save_output} to the file name. Likewise, it can be loaded by setting \texttt{load_output} to this file name, in which case the discriminants will be loaded instead of computed from scratch. This option is useful for quickly converting an output of the computation to \texttt{Oscar} format.
\end{example}

\begin{example}\label{ex:weights}
The package \texttt{PLD.jl} provides the functions \texttt{getWeights} and \texttt{getIF}. These compute weight vectors lying in the interior of each cone in the normal fan of the Newton polytope of a polynomial (\texttt{getWeights}) and the initial form of a polynomial for a fixed weight vector (\texttt{getIF}). Here is an example for $G = \texttt{par}$: 
\begin{minted}{julia}
edges = [[3,1],[1,2],[2,3],[2,3]];
nodes = [1,1,2,3];
U, F = getUF(edges, nodes, internal_masses = :generic,
                           external_masses = :generic);
weights = getWeights(U+F);
getIF(U+F, weights[2])
\end{minted}
The following code lines compute the weights for the example of Sec.~\ref{subsec:parachute}, as well as the initial forms of the faces of the Newton polytope on $\cal{G}_{\texttt{par}}$ of codimension $1$ (recall that \texttt{Julia} indexes lists from $1$, but our codimensions start from $0$).

For example, \texttt{weights[2][1]} gives the weight list \texttt{[-1,-1,-1,-1]} corresponding to the first-type leading singularity. If we wanted to compute discriminants only for this case, it is enough to run
\begin{minted}{julia}
getPLD(edges, nodes, internal_masses = :generic, external_masses = :generic,
                     method = :num, single_weight = weights[2][1])
\end{minted}
The result is the discriminant with two components from \eqref{eq:par-leading-disc} and \eqref{eq:par-leading-disc2}.
\end{example}

\begin{example}
Even though our main focus is on Feynman integrals, the package \texttt{PLD.jl} can be also used to analyze singularities of other classes of integrals, for instance those appearing in applications to cosmological wavefunctions, gravitational-wave physics, or energy correlators. As a simple example, consider the integral
\begin{align}\label{eq:example15}
\int_{\R^3_+} &\frac{z_1^\varepsilon z_2^\varepsilon z_3^\varepsilon\, \d^3 z}{(X_1 + z_1 + X_2 + z_2 + X_3 + z_3)(X_1 + z_1 + Y_1)(X_3 + z_3 + Y_2)}\\
&\times \frac{(X_1 + z_1 + X_3 + z_3 + 2X_2 + 2z_2 + Y_1 + Y_2)}{(X_2 + z_2 + Y_1 + Y_2)(X_1 + z_1 + X_2 + z_2 + Y_2)(X_2 + z_2 + X_3 + z_3 + Y_1)}
\end{align}
with $\varepsilon\in \C$. Furthermore, all the symbols with capital letters are parameters and $z_1, z_2, z_3$ are integration variables. 
This integral computes the three-chain graph wavefunction, see, e.g., \cite{Hillman:2019wgh}. Its singularities can be obtained with the specialized principal $A$-determinant applied to the Cayley configuration $\sum_{i=1}^{6} \alpha_i P_i$, where $P_i$'s are the polynomial factors in the denominator of \eqref{eq:example15} and the $\alpha_i$'s are adjoined to the list of variables. We can dehomogenize it by setting, say $\alpha_6 = 1$. In \texttt{PLD.jl}, it amounts to running
\begin{minted}{julia}
R, pars = PolynomialRing(QQ, ["X1", "X2", "X3", "Y1", "Y2"])
S, vars = LaurentPolynomialRing(R, ["z1", "z2", "z3",
                                    "α1", "α2", "α3", "α4", "α5"])
(X1, X2, X3, Y1, Y2) = pars
(z1, z2, z3, α1, α2, α3, α4, α5) = vars

P1 = X1 + z1 + X2 + z2 + X3 + z3
P2 = X1 + z1 + Y1 
P3 = X3 + z3 + Y2
P4 = X2 + z2 + Y1 + Y2
P5 = X1 + z1 + X2 + z2 + Y2
P6 = X2 + z2 + X3 + z3 + Y1
Cayley = α1*P1 + α2*P2 + α3*P3 + α4*P4 + α5*P5 + P6

getSpecializedPAD(Cayley, pars, vars)
\end{minted}
The output gives the discriminants
\begin{gather}
X_1 + X_2 + X_3,\; X_1 + X_2 + Y_2,\; X_1 + X_2 - Y_2,\; X_1 + X_3 - Y_1 - Y_2,\; X_1 + Y_1,\\
X_1 - X_3 - Y_1 + Y_2,\; X_1 - Y_1,\; X_1 - Y_1 + 2Y_2,\; X_1 - Y_1 - 2Y_2,\; X_2 + X_3 + Y_1,\\
X_2 + X_3 - Y_1,\; X_2 + Y_1 + Y_2,\; X_2 + Y_1 - Y_2,\; X_2 - Y_1 + Y_2,\; X_2 - Y_1 - Y_2,\\
X_3 + 2Y_1 - Y_2,\; X_3 + Y_2,\; X_3 - 2Y_1 - Y_2,\; X_3 - Y_2,\; Y_1,\; Y_1 + Y_2,\; Y_1 - Y_2,\; Y_2\, .
\end{gather}
We verified that the singularities of the integral \eqref{eq:example15} evaluated in \cite[Sec.~4.3]{Hillman:2019wgh} and \cite[App.~B]{De:2023} are strictly contained in the above set.

\end{example}

\subsection{Standard Model examples}
\label{subsec:standard-model}

We applied the above algorithm to all diagrams in Fig.~\ref{fig:diagrams}. For each of them, we ran the symbolic elimination described in Sec.~\ref{subsec:symbolic} until the problems became too lengthy or did not terminate. For example, for $G = \texttt{outer-dbox}$, the discriminants for all faces of codimension-$2$ or higher were computed symbolically. Beyond this point, we apply the numerical algorithm from Sec.~\ref{sec:numerical}.\footnote{For comparison, the full computation for one of the simplest diagrams, $G = \texttt{inner-dbox}$, took $9.3$ minutes, while the most difficult diagram, $G = \texttt{npl-dpent}$, took $63.7$ hours on two Intel Xeon E5-2695 v4 CPUs with 18 cores each.} In addition, we also ran \texttt{HyperInt} on the same set of diagrams using the method outlined in App.~\ref{sec:appendixHyperInt}. In contrast with \texttt{PLD.jl}, \texttt{HyperInt} terminated only for diagrams (a-e) and (g-h). For (b-e) and (g-h), \texttt{HyperInt} found additional components of the Euler discriminant not found by \texttt{PLD.jl}. Out of the 114 diagrams we tested, \texttt{HyperInt} terminated for 64. For 25 diagrams among these 64, the Euler discriminant contains at least one more component than the PLD.

Tab.~\ref{tab:diagrams} summarizes the basic information of each diagram $G$: the number of kinematic variables it depends on ($\dim \mathcal{E}$), the $f$-vector of the Newton polytope $\Newt(\G_G(\mathcal{E}))$, as well as the degrees of the discriminants found by either method. The notation follows \cite[Tab.~1]{Mizera:2021icv} where $[\ldots, a^b, \ldots]_1$ means that there are $b$ components with degree $a$ and the codimension in $\mathcal{E}$ is $1$.
All the detailed results are collected at \PLDwebsite, where we also provide results for all the diagrams in \cite[Fig.~1]{Mizera:2021icv} with various assignments of internal and external masses.\footnote{To ensure stability of the rationalization step (v), we do not attempt to interpolate components with $>10000$ samples. For example, $G=\texttt{pentb}$ from \cite[Fig.~1]{Mizera:2021icv} with all massless particles depends on $m=5$ kinematic parameters, which means we can reconstruct its discriminants up to degree $d_i \leqslant 19$, while the case with all generic internal and external masses has $m=18$, which gives $d_i \leqslant 4$.}

An example output file is organized as follows. The header contains all the information about the diagram:
\begin{minted}{julia}
################################
# Diagram information
################################

name = "outer-dbox"

edges = [[1, 2], [2, 5], [3, 5], [3, 4], [4, 6], [1, 6], [5, 6]]
nodes = [1, 2, 3, 4]
internal_masses = [m2, m2, m2, m2, m2, m2, 0]
external_masses = [0, 0, 0, 0]

U = x[1]*x[3] + x[1]*x[4] + x[1]*x[5] + x[1]*x[7] + # more terms
F = -m2*x[1]^2*x[3] - m2*x[1]^2*x[4] - m2*x[1]^2*x[5] + # more terms
parameters = [m2, s, t]
variables = [x[1], x[2], x[3], x[4], x[5], x[6], x[7]]

χ_generic = 64
f_vector = [39, 153, 271, 272, 165, 60, 12]
\end{minted}
Here, \texttt{χ_generic} is the generic signed Euler characteristic that can be found in Tab.~\ref{tab:volVSEuler}, while \texttt{f_vector} is the $f$-vector from Tab.~\ref{tab:diagrams}. This header is followed by a list of individual components (in this case, $8$ of them), for example:
\begin{minted}{julia}
################################
# Component 4
################################

D[4] = m2*s + m2*t - 1//4*s*t
χ[4] = 55
weights[4] = [[-1, -1, 0, -1, -1, 0, -2], # more weights ]
computed_with[4] = ["PLD_sym", "PLD_num", "HyperInt"]
\end{minted}
For each component, \texttt{D[i]} is the defining polynomial, $\texttt{χ[i]}$ is the signed Euler characteristic evaluated on $\texttt{D[i]} = 0$ (computed using reliable numerical techniques, see Rmk.~\ref{rmk:1}), $\texttt{weights[i]}$ is the list of weights for which the component was found using either method (no weights are recorded for \texttt{HyperInt}). Finally, $\texttt{computed_with[i]}$ summarizes which method was used to find the corresponding component.

\begin{table}[t]
\centering
\footnotesize
    \begin{tabular}{c|c|c|c}
        Diagram $G$ & $\dim \mathcal{E}$ & $f$-vector of $\Newt(\G_G(\mathcal{E}))$ & Degrees of $\nabla_\chi(\mathcal{E})$ \\
        \hline
        $\texttt{inner-dbox}$ & 3 &  $(37, 156, 294, 310, 195, 72, 14)$ & $[1^7, 2]_1$ \\
        $\texttt{outer-dbox}$ & 3 & $(39, 153, 271, 272, 165, 60, 12)$ & $[1^6, 2, 3]_1$ \\
        $\texttt{Hj-npl-dbox}$ & 4 & $(40, 174, 333, 350, 215, 76, 14)$ & $[1^{13}, 2^6, 3^3]_1$ \\
        $\texttt{Bhabha-dbox}$ & 4 & $(32, 162, 347, 393, 252, 90, 16)$ & $[1^6, 2^3]_1$ \\
        $\texttt{Bhabha2-dbox}$ & 4 & $(37, 187, 394, 435, 272, 96, 17)$ & $[1^7, 2^3, 3]_1$ \\
        $\texttt{Bhabha-npl-dbox}$ & 4 & $(35, 184, 402, 457, 291, 103, 18)$ & $[1^8, 2^3]_1$ \\
        $\texttt{kite}$ & 6 & $(24, 66, 73, 39, 10)$ & $[1^6, 2^4, 3, 4^2]_1$ \\
        $\texttt{par}$ & 7 & $(15, 33, 27, 9)$ & $[1^7, 2^2, 3, 4^2, 6]_1$ \\
        $\texttt{Hj-npl-pentb}$ & 7 & $(56, 294, 681, 884, 699, 343, 101, 16)$ & $[1^{31}, 2^{16}, 3^{12}, 4^{7}, 6^{3}, 12^{2}]_1$ \\
        $\texttt{dpent}$ & 9 & $(64, 528, 1770, 3158, 3336, 2171, 867, 202, 24)$ & $[1^{30}, 2^{11}, 4^6, 5, 6]_1$ \\
        $\texttt{npl-dpent}$ & 9 & $(64, 597, 2117, 3852, 4058, 2606, 1029, 239, 28)$ & $[1^{52}, 2^{22}, 3^6, 4^{16}, 5^7, 6^3]_1$ \\
        $\texttt{npl-dpent2}$ & 9 & $(63, 562, 1969, 3591, 3820, 2482, 988, 230, 27)$ & $[1^{46}, 2^{21}, 3^2, 4^{12}, 5^4, 6^3]_1$\\
    \end{tabular}
\caption{\label{tab:diagrams} Summary table of the results for each diagram $G$: dimension of the kinematic subspace $\mathcal{E}$, $f$-vector of the corresponding Newton polytope $\Newt(\G_G(\mathcal{E}))$, and degrees of the components of $\nabla_\chi(\mathcal{E})$ we found, see the main text for details.}
\end{table}

In the above computations, we made heavy use of parallelization implemented in \texttt{HomotopyContinuation.jl} to speed-up computations.
On multi-core machines, the number of threads is controlled by the variable \texttt{JULIA_NUM_THREADS}, e.g., one can set it with \texttt{export JULIA_NUM_THREADS = 64} before running \texttt{PLD.jl}.

\section{\label{sec:conclusion}Conclusion and outlook}

In this work, we revisited the Landau analysis of singular loci of Feynman integrals using tools from computational algebraic geometry. With a practical view towards explicit computations, we defined the principal Landau determinant (PLD) $E_G({\cal E})$ of a Feynman diagram as a subvariety of the parameter space ${\cal E}$. It estimates the Euler discriminant variety $\nabla_{\chi}({\cal E})$, which is the locus of kinematic parameters for which the signed Euler characteristic of the hypersurface defined by the graph polynomial ${\cal G}_G$ drops compared to its generic value.  

Algorithms for computing the PLD are implemented in our open-source \texttt{Julia} package \texttt{PLD.jl} available at \url{https://mathrepo.mis.mpg.de/PLD/} together with a tutorial explaining its features and a database containing the output of our algorithm on 114 examples.
Our definitions and algorithms are inspired by the principal $A$-determinant from the work of Gelfand, Kapranov, and Zelevinsky (GKZ), which is known to be the singular locus of a $D$-module, called $A$-hypergeometric system. Our examples show that a careful adaptation of the GKZ framework is necessary, but it results in an effective method for computing singularities of Feynman integrals.
Sec.~\ref{sec:4} presented the only cases we found in which the kinematic space is not properly contained in the principal $A$-determinant. We leave the case of vanishing internal and external masses with a conjectural formula (see Conj.~\ref{conj:A00}) for the principal $A$-determinant. The main obstacle in proving this formula is the complexity of the face structure of the Newton polytope.

We compared our results with those of \texttt{HyperInt} for computing an upper bound on the singularity locus of Feynman integrals (see App.~\ref{sec:appendixHyperInt}). The \texttt{HyperInt} computation is based on Pham and Brown's \emph{Landau variety} \cite{pham2011singularities,Brown:2009ta}. In future research, it would be desirable to understand the relation between all these geometric objects in detail. Concretely, which inclusions hold between the Landau variety, the Euler discriminant variety, the principal Landau determinant variety, and the singular locus of the $D$-module annihilating the Feynman integral? Our belief about the relation between the varieties defined by the Euler discriminant and the PLD is stated in Conj.~\ref{conj:pldvschi}. 

Furthermore, our algorithms for computing the PLD can be improved. For instance, one could exploit \emph{sparsity}, meaning the fact that its factors only have a few nonzero coefficients. Also, to what extent can we shrink the gap between Euler discriminant and PLD by (partially) compactifying Schwinger parameter space in our computations (see Ex.~\ref{ex:embcomp} and App.~\ref{sec:appendix2})?

\acknowledgments

We thank Marko Berghoff, Marie Brandenburg, Francis Brown, Hofie Hannesdottir, Erik Panzer, and Bernd Sturmfels for useful discussions. C.F. has received funding from the
European Union’s Horizon 2020 research and innovation programme under the Marie Sklodowska-Curie grant agreement No 101034255.
S.M. gratefully acknowledges funding provided by the Sivian Fund and the Roger Dashen Member Fund at the Institute for Advanced Study.
This material is based upon work supported by the U.S. Department of Energy, Office of Science, Office of High Energy Physics under Award Number DE-SC0009988. We are grateful to anonymous referees for their insightful comments on an earlier version of this paper. 

\appendix

\section{\label{sec:appendixHyperInt} Bounding the Landau variety with \texttt{HyperInt}}

In this appendix, we explain how to use \texttt{HyperInt} \cite{Panzer:2014caa} for computing an upper bound on the singularity locus of Feynman integrals (the Landau variety in the sense of \cite{Brown:2009ta}).\footnote{We thank Erik Panzer for suggesting to compare the output of \texttt{HyperInt} with \texttt{PLD.jl}, and for clarifying the use of the command \texttt{cgReduction}.} We then use the Euler discriminant to filter out spurious components and compare it with the output of \texttt{PLD.jl}. As an illustrative example, we consider the double-box diagram with an outer massive loop, $G = \texttt{outer-dbox}$, from Fig.~\ref{fig:diagrams}c.

\texttt{HyperInt} implements the compatibility graph method of polynomial reduction based on \cite{Brown:2009ta}. A Feynman diagram does not need to integrate to polylogarithms for this algorithm to terminate.
For the specific example at hand, the \texttt{Maple} code needed to find it is 
\begin{minted}{julia}
read "HyperInt.mpl":

edges := [{1,2}, {2,5}, {5,3}, {3,4}, {4,6}, {6,1}, {5,6}]:
nodes := [[1,0], [2,0], [3,0], [4,0]]:
internal_masses := [m2, m2, m2, m2, m2, m2, 0]:

U := graphPolynomial(edges):
F := subs({s12 = -s, s14 = -t},
          secondPolynomial(edges, nodes, internal_masses)):

S := irreducibles({U,F}):
L[{}] := [S, combinat[choose](S, 2)]:

cgReduction(L, {s,t,m2}, 4):
candidates := {s,t,m2} union L[{seq(x[i], i=1..7)}][1];
\end{minted}
The notation matches that of \texttt{PLD.jl} as much as possible (the additional minus signs in \texttt{-s} and \texttt{-t} are needed because \texttt{HyperInt} uses Euclidean conventions). Line $14$ computes the upper bound on the Landau variety, defined by the zero locus of polynomials in \texttt{s}, \texttt{t}, and \texttt{m2}. In general, one can provide an additional parameter \texttt{d} to the command \texttt{cgReduction} of \texttt{HyperInt} (line 14), see \cite{Panzer:2014caa}. The default value is 1, and here it is chosen to be 4. This parameter controls the maximum degree of the consecutive coordinate projections of the algorithm. The result $\mathtt{candidates}$ is a list of $17$ polynomials:
\begin{minted}{julia}
candidates := {m2,
               s,
               t,
               m2 - s,
               m2 - 1/2*s,
               m2 - 1/4*s,
               m2 - 1/4*t,
               s + t,
               s + 2*t,
               m2*s + m2*t - 1/4*s*t,
               m2*s + 2*m2*t - 1/2*s*t,
               m2*t - 1/4*s^2 - 1/4*s*t,
               m2*t + 1/4*s^2 - 1/4*s*t,
               m2*t + s^2,
               m2*s^2 + 4*m2*s*t + 4*m2*t^2 - s*t^2,
               m2*s^2 + 2*m2*s*t + m2*t^2 - s^2*t,
               m2^2*s - 2*m2*s*t - 4*m2*t^2 + s*t^2}
\end{minted}
They are candidates for the polynomials defining components of the Landau variety. However, not all of them correspond to actual singularities. To detect which ones do, we can use the Euler characteristic check. This can be achieved with the tools provided in \texttt{PLD.jl}. After defining \texttt{edges}, \texttt{nodes}, \texttt{internal_masses} as above, it amounts to running the following \texttt{Julia} code:
\begin{minted}[escapeinside=||]{julia}
U, F, pars, vars = getUF(edges, nodes;
                         internal_masses = internal_masses,
                         external_masses = :zero);
(m2, s, t) = pars;
candidates = # the above list with / |\(\rightarrow\)| //
EulerDiscriminantQ(U+F, pars, vars, candidates)
\end{minted}
The function ${\tt EulerDiscriminantQ}$ checks whether the provided candidate components belong to the Euler discriminant by computing their signed Euler characteristics and comparing them to the generic signed Euler characteristic $\chi^*$. Moreover, to increase reliability of the results, each Euler characteristic is computed ten times and the maximum of the computed and certified values is selected.
In particular, in the code snippet above, line $6$ scans over the list of \texttt{candidates} and finds that only $8$ out of the $17$ components lie in the Euler discriminant. They are
\be
\{ m^2,\; s,\; t,\; m^2 - s/4,\; m^2 - t/4,\; s+t,\; m^2(s+t) - st/4,\; m^4 s - 2 m^2 t (s + 2t)  + s t^2\}\, ,
\ee
where $m^2 = \mathtt{m2}$.
They have Euler characteristics $12, 6, 33, 48, 63, 45, 55, 62,$ respectively, compared to the generic $64$. This list agrees with the singularities of this topology of Feynman integrals that can be determined from the differential equations in \cite[App.~A]{Caron-Huot:2014lda}.

This output can be compared with the result of running \texttt{PLD.jl}. The latter finds the first $7$ out of the above $8$ components. These results are summarized in {\PLDwebsite} in the conventions explained in Sec.~\ref{subsec:standard-model}. 
Note that, in order for \texttt{HyperInt} to return a non-empty list of components, one might have to increase the optional parameter \texttt{d} of \texttt{cgReduction} (see above). In practice, one needs to strike a balance: too low of a bound leads to no components, and computations with too high of a bound might not terminate or consume too much memory. In our experiments, we tried several values of \texttt{d} in an unstructured way. Using this parameter in a more informed way, it might be that \texttt{HyperInt} can find an upper bound for the Landau variety of a few more diagrams. 
For larger examples (f) and (i-l) from Fig.~\ref{fig:diagrams}, we found that \texttt{HyperInt} does not terminate, which means we could not make a direct comparison with \texttt{PLD.jl}.

\section{\label{sec:appendix}From loop momentum to Schwinger parameters}

In this appendix, we informally review the steps that bring a Feynman integral in the loop-momentum representation into an integral over Schwinger parameters. This derivation also defines the Symanzik polynomials for integrals involving numerators.

The starting point is a family of scalar Feynman integrals in $\D$ space-time dimensions with $\n$ external legs and $\L$ loops:
\be\label{eq:Feynman-integral}
I_{\nu_1, \nu_2, \ldots, \nu_m} := \frac{1}{(i\pi^{\D/2})^\L} \int \frac{\d^\D \ell_1\, \d^\D \ell_2 \cdots \d^\D \ell_\L}{P_1^{\nu_1} P_2^{\nu_2} \cdots P_m^{\nu_m}}\, .
\ee
The integration variables are the $\L$ loop momenta $\ell_a$ with $a=1,2,\ldots,\L$, which are vectors in the Minkowski space $\R^{1,\D-1}$. Here, each inverse propagator $P_j$ is a quadratic polynomial in $\ell_a$'s, the external momenta $p_i$'s for $i=1,2,\ldots,\n$, and the masses $m_i$. The external momenta satisfy the momentum conservation $\sum_{i=1}^{n} p_i = 0$. Finally, $i=\sqrt{-1}$, and $\nu_j$'s are integer exponents. Their number $m$ can be taken to be equal to the number $\E$ of internal edges in the diagram or larger. In many applications, one extends this set to span a basis of the kinematic invariants $\{ \ell_a \cdot \ell_b, \ell_a \cdot p_i\}$ with $a,b = 1,2,\ldots,\L$ and $i = 1,2,\ldots,\min(\D,\n-1)$. The number of such invariants is
\be
m = \frac{\L (\L + 1)}{2} + \L \min(\D, \n-1)
\ee
if we assume that all $\ell_a$'s are unconstrained (it is enough that $\D \geq \frac{\L+1}{2}$). The most interesting case is $\nu_j \geq 0$ for $j=1,2,\ldots,\E$ for the set of propagators, and $\nu_j \leq 0$ for the remaining $j=\E+1,\ldots,m$ called irreducible scalar products (ISP's). However, in the context of master integrals for differential equations, one often treats both as arbitrary integers. Whenever \eqref{eq:Feynman-integral} diverges, one treats it as a function of $\D$ in a procedure called dimensional regularization.

We will derive two Schwinger-parametric representations of Feynman integrals \eqref{eq:Feynman-integral}, depending on whether the powers $\nu_j$ for $j=\E+1,\ldots,m$ are arbitrary or constrained to be non-negative.

\subsection{Arbitrary powers of ISP's}

We first treat the case in which all $\nu_j$'s are arbitrary. We introduce Schwinger parameters $\alpha_j$ for $j=1,2,\ldots,m$ by writing
\be\label{eq:Schwinger-trick}
\frac{1}{P_j^{\nu_j}} \,=\, \frac{(-i)^{\nu_j}}{\Gamma(\nu_j)} \int_0^{\infty}  \alpha_j^{\nu_j - 1}\, \e^{i P_j \alpha_j}\d \alpha_j\, .
\ee
Applying this identity $m$ times, we obtain%
\footnote{Strictly speaking, we should use the identity \eqref{eq:Schwinger-trick} with $\frac{1}{(P_j/\mu^2)^{\nu_j}}$ with some mass scale $\mu$ so that the exponent in the integrand is dimensionless, but here we work in the units where $\mu = 1$.}
\be
I_{\nu_1, \nu_2, \ldots, \nu_m} \,=\, \frac{1}{(i\pi^{\D/2})^\L} \frac{(-i)^{\sum_{j=1}^{m}\nu_j}}{\prod_{j=1}^{m} \Gamma(\nu_j)} \int \d^\D \ell_1\, \d^\D \ell_2 \cdots \d^\D \ell_\L \int_{\R_+^m} \prod_{j=1}^{m} \alpha_j^{\nu_j - 1}\, \e^{i \sum_{j=1}^{m} \alpha_j P_j}\d^{m}\alpha \, .
\ee
Using the fact that each $P_j$ is quadratic in the loop momenta, the sum in the exponent can be further rewritten as
\be
\sum_{j=1}^{m} \alpha_j P_j \,=:\, \sum_{a,b=1}^{\L} \ell_a \cdot \ell_b\, \mathbf{Q}_{ab} + 2 \sum_{a=1}^{\L} \ell_a \cdot \mathbf{L}_a + c\, ,
\ee
thus defining $\mathbf{Q}$ which is an $\L \times \L$ matrix, $\mathbf{L}$ is an $\L$-vector whose entries are Minkowski vectors, and $c$ is a scalar. Each of $\mathbf{Q}, \mathbf{L}, c$ is linear in the Schwinger parameters. At this stage, we can complete the square in the loop momenta as follows:
\be
\sum_{i=1}^{m} \alpha_i P_i \,=\, \left[ \bm{\ell} + \mathbf{Q}^{-1} \mathbf{L} \right]^\intercal \cdot \mathbf{Q} \left[ \bm{\ell} + \mathbf{Q}^{-1} \mathbf{L} \right] - \mathbf{L}^\intercal \cdot \mathbf{Q}^{-1} \mathbf{L} + c\, ,
\ee
where for brevity we used matrix notation with the vector $[\bm{\ell}]_a = \ell_a$. At this stage, the integrals over the loop momenta are Gaussian and we can simply do them (Wick rotation of the contour in the time component of $\mathbb{R}^{1,\D-1}$ leads to an additional factor of $i^\L$, see \cite[App.~A]{Hannesdottir:2022bmo}):
\be
I_{\nu_1, \nu_2, \ldots, \nu_m} \,=\, i^{\L\D/2}\frac{(-i)^{\sum_{j=1}^{m}\nu_j}}{\prod_{j=1}^{m} \Gamma(\nu_j)} \int_{\R_+^m} \frac{1}{\Ue^{\D/2}} \prod_{j=1}^{m} \alpha_j^{\nu_j - 1}\, \e^{i \Fe/\Ue}\d^{m}\alpha\, .
\ee
Here, we have defined the generalized \emph{Symanzik polynomials}
\be
\Ue := \det \mathbf{Q}, \qquad \Fe := \left( - \mathbf{L}^\intercal \cdot \mathbf{Q}^{-1} \mathbf{L} + c \right) \Ue\, .
\ee
Compared to the main text, we omit the subscripts ${}_G$ for clarity.
They are homogeneous polynomials with degrees $\L$ and $\L{+}1$ in the Schwinger parameters respectively.

One can further massage this expression into the form used in the main part of the paper. For example, without changing the value of the integral, we can multiply it by a dummy propagator with the Schwinger parameter $\alpha_0$:
\be
\frac{1}{1^{\nu_0}} \,=\, \frac{(-i)^{\nu_0}}{\Gamma(\nu_0)} \int_0^{\infty} \d \alpha_0\, \alpha_0^{\nu_0 - 1}\, \e^{i \alpha_0}\, ,
\ee
where the choice $\nu_0 = (\L+1)\D/2 - \sum_{j=1}^{m} \nu_j$ will turn out to be convenient. The integral now takes the form
\be
I_{\nu_1, \nu_2, \ldots, \nu_m} \,=\, i^{\L\D/2} \frac{(-i)^{\sum_{j=0}^{m}\nu_j}}{\prod_{j=0}^{m} \Gamma(\nu_j)} \int_{\R_+^{m+1}} \frac{1}{\Ue^{\D/2}} \prod_{j=0}^{m} \alpha_j^{\nu_j - 1}\, \e^{i (\Fe + \alpha_0 \Ue)/\Ue}\d^{m+1}\alpha \, .
\ee
Finally, we can integrate out the overall scale by a change of variables
\be
(\alpha_0, \alpha_1, \ldots, \alpha_m) \,\to\, \lambda(1,\alpha_1, \ldots, \alpha_m)\, .
\ee
where $\lambda > 0$ and the measure transforms as $\d^{m+1} \alpha \to \lambda^{m} \d\lambda\, \d^{m} \alpha$. Using homogeneity properties of the Symanzik polynomials, the result is
\begin{align}
I_{\nu_1, \nu_2, \ldots, \nu_m} &= i^{\L\D/2} \frac{(-i)^{\sum_{j=0}^{m}\nu_j}}{\prod_{j=0}^{m} \Gamma(\nu_j)} \int_{\R_+} \lambda^{\D/2 - 1}\d \lambda\!\! \int_{\R_+^{m}} \frac{1}{\Ue^{\D/2}} \prod_{j=1}^{m} \alpha_j^{\nu_j - 1}\, \e^{i \lambda (\Fe + \Ue)/\Ue}\d^{m}\alpha \\
&= \frac{\Gamma(\D/2)}{\prod_{j=0}^{m} \Gamma(\nu_j)} \int_{\R_+^{m}} \frac{1}{(\Fe + \Ue)^{\D/2}} \prod_{j=1}^{m} \alpha_j^{\nu_j - 1}\d^{m}\alpha\, .\label{eq:FU-parametrization}
\end{align}
In the final step we integrated out $\lambda$ using the same identity as in \eqref{eq:Schwinger-trick}. This is the form of the Feynman integrals \cite{Lee:2013hzt} used in the main text.

\subsection{Non-positive powers of ISP's}

Let us now discuss the case in which the powers $\nu_j \leq 0$ are non-positive integers for $j=\E+1,\ldots,m$. The difference will be that now we can use a different Schwinger parametrization:
\be
\frac{1}{P_j^{\nu_j}} = \frac{(-i)^{-\nu_j}\Gamma(1{-}\nu_j) }{2\pi i}\oint_{|\alpha_j| = \eps} \alpha_j^{\nu_j - 1} \e^{i P_j \alpha_j}\, \d \alpha_j\, .
\ee
Apart from the prefactors, the difference to \eqref{eq:Schwinger-trick} is that we integrate $\alpha_j$ over a small anti-clockwise circle with radius $\eps$ around the origin.

Since the integrand is identical to the one encountered in the previous subsection, we can carry out the same manipulations up to \eqref{eq:FU-parametrization}. Taking care of the overall normalization, we get
\be
I_{\nu_1, \nu_2, \ldots, \nu_m} =  \frac{\Gamma(\D/2) \prod_{j=\E+1}^m (-i)^{-2\nu_j}\Gamma(1-\nu_j)}{(2\pi i)^{m-\E}\prod_{j=0}^{m} \Gamma(\nu_j)} \int_{\R_+^{\E} \times T^{m-\E}} \frac{1}{(\Fe + \Ue)^{\D/2}} \prod_{j=1}^{m} \alpha_j^{\nu_j - 1}\d^{m}\alpha\, .
\ee
The idea is to finish by carrying out the $m-\E$ integrations over the torus $T^{m-\E}$ given by the product of circles $\{ |\alpha_j| = \eps \}$ for $j=\E+1,\ldots,m$. To this end, let us define
\be
\U := \Ue \big|_{\alpha_j = 0},\qquad
\F := \Fe \big|_{\alpha_j = 0} \qquad \text{for}\qquad j=\E{+}1, \ldots, m\, .
\ee
These are the Symanzik polynomials of the Feynman diagram without any numerators, as defined in \cite[Defs.~1--2]{Mizera:2021icv}. Recall that $\U$ and $\F$ can be computed combinatorially in terms of sums over spanning trees and $2$-trees. Let us write $\Ue = \U + \U'$ and $\Fe = \F + \F'$, where the primed polynomials contain all the dependence on the Schwinger parameters we want to integrate out.
At this stage, we need to compute
\begin{align}
\frac{\prod_{j=\E+1}^{m} (-1)^{-2\nu_j}\Gamma(1-\nu_j)}{(2\pi i)^{m-\E}}\oint_{T^{m-\E}} &\frac{1}{(\F + \U + \F' + \U')^{\D/2}} \prod_{j=\E+1}^{m} \alpha_j^{\nu_j - 1} \d^{m-\E} \alpha\\
&=: \frac{\mathcal{N}}{(\F + \U)^{\D/2 - \sum_{j=\E+1}^{m} \nu_j }}\, ,
\end{align}
which defines the polynomial $\mathcal{N}$. The normalization is chosen such that it has degree $-\sum_{j=\E+1}^{m} \nu_j \geq 0$ in the remaining Schwinger parameters $\alpha_j$ with $j=1,2,\ldots,\E$.

Putting everything together, we find that the Feynman integral with numerators can be written as
\be
I_{\nu_1, \nu_2, \ldots,\nu_\E | \nu_{\E+1},\ldots, \nu_m} \,=\, \frac{\Gamma(\D/2)}{\prod_{j=0}^{m} \Gamma(\nu_j)} \int_{\R_+^{\E}} \frac{\mathcal{N}}{(\F + \U)^{\D/2 - \sum_{j=\E+1}^{m} \nu_j}} \prod_{j=1}^{\E} \alpha_j^{\nu_j - 1}\d^{\E}\alpha\, .
\ee
In particular, since $\mathcal{N}$ is polynomial, it can be expanded in terms of monomials
\be
\mathcal{N} = \sum_{k} c_k \prod_{j=1}^{\E} \alpha_j^{\rho_{jk}}\, .
\ee
Here, the coefficients $c_k$ are independent of the $\alpha_j$'s and all $\rho_{jk} \in \Z_{\geq 0}$.
As a consequence, any Feynman integral with non-positive powers $\nu_j$ of ISP's, $j=\E+1,\ldots, m$ can be expressed as a linear combination of those with no ISP powers and shifted $\D$. Its singularity analysis is then identical to the case without numerators, which leads to the following result.

\begin{proposition}
The PLD associated to the integral $I_{\nu_1, \nu_2, \ldots,\nu_\E | \nu_{\E+1},\ldots, \nu_m}$ with arbitrary $\nu_{j\leq E}$ and $\nu_{j>E} \in \Z_{\leq 0}$ coincides with that of $I_{\nu_1, \nu_2, \ldots,\nu_\E | 0,\ldots, 0}$ for arbitrary $\nu_{j\leq E}$.
\end{proposition}

In other words, numerators cannot introduce new components of the PLD.

\section{\label{sec:appendix2}A toric view on principal Landau determinants}

Our Def.~\ref{def:PLD} of the principal Landau determinant $E_G({\cal E})$ uses the incidence varieties $Y_{G,\Gamma}({\cal E})$ for each face $Q$ of the polytope ${\rm Conv}(A)$. In this appendix, we define a global object ${\cal Y}_G$ which captures all these incidence varieties at once. Let $P = {\rm Newt}({\cal G}_G)$ and write $X_P$ for the corresponding projective toric variety. We consider the equations 
\[ {\cal G}_G \, = \, \alpha_1 \cdot \frac{\partial {\cal G}_G}{\partial \alpha_1} \,= \, \cdots \, = \, \alpha_\E \cdot \frac{\partial {\cal G}_G}{\partial \alpha_\E} \, = \, 0. \]
The Newton polytope of each of these equations is contained in $P$. Hence, we can regard them as $\E+1$ global sections of the line bundle ${\cal O}_{X_P}(D_P) \otimes {\cal L}$ on $X_P \times {\cal E}$, where $D_P$ is the divisor on $X_P$ associated to $P$ (cf.~\cite[Ex. 5.30]{telen2022introduction}), and ${\cal L}$ is the trivial bundle on ${\cal E}$ with sections $\mathbb{C}[{\cal E}]$. The intersection of their zero loci on $X_P \times {\cal E}$ is ${\cal Y}_G$. 

The toric variety $X_P$ is stratified by torus orbits, one for each face $Q \subset P$. We denote these orbits by $O_Q$, and have that $O_Q \simeq (\mathbb{C}^*)^{\dim Q}$. In particular, for the full-dimensional face $Q = P$, we have that $O_P \simeq (\mathbb{C}^*)^\E$ and 
\[ {\cal Y}_G \cap (O_P \times {\cal E}) \, = \, Y_{G,P}. \]
More generally, for lower dimensional faces $Q \subset P$, we have 
\[ {\cal Y}_G \cap (O_Q \times {\cal E}) \simeq Y_{G,Q} / T_Q, \]
where $T_Q$ is a $(\E-\dim Q)$-dimensional quasi-torus acting on $(\mathbb{C}^*)^\E$ so that $(\mathbb{C}^*)^\E/T_Q = O_Q$. Importantly, this means that the projections to the parameter space ${\cal E}$ agree: 
\[ \pi_{\cal E}( {\cal Y}_G \cap (O_Q \times {\cal E}) ) \, = \, \pi_{\cal E}(Y_{G,Q}). \]
Let ${\cal Y}_G = \bigcup_{i \in \mathbb{I}(G)} {\cal Y}_G^{(i)}$ be the irreducible decomposition of ${\cal Y}_G$ for some finite indexing set $\mathbb{I}(G)$. As in Sec.~\ref{sec:3}, we define discriminants by projecting and taking the closure:
\[ \nabla_G^{(i),\circ} \, = \, \pi_{\cal E}({\cal Y}_G^{(i)}), \quad \nabla_G^{(i)} \,=\, \overline{\nabla_G^{(i),\circ}}.  \]
Let $\mathbb{I}(G)_1 = \{ i \in \mathbb{I}(G) \, :\, \dim {\nabla}_G^{(i)} = \dim {\cal E} - 1 \}$. For each $i \in \mathbb{I}(G)_1$, there is a unique (up to scale) polynomial $\Delta_G^{(i)}$ with vanishing locus $\nabla_G^{(i)}$.
\begin{definition} \label{def:PLDtoric}
    The principal Landau determinant $E_G({\cal E})$ is the unique (up to scale) square-free polynomial $E_G({\cal E}) \in \mathbb{C}[{\cal E}]$ such that 
    \[ \{ E_G({\cal E}) \, = \, 0 \} \, = \, \left \{ \prod_{i \in \mathbb{I}(G)_1} \Delta_{G}^{(i)}({\cal E}) \, = \, 0 \right \} \, = \, \bigcup_{i \in \mathbb{I}(G)_1} \nabla_G^{(i)}. \]
\end{definition}
This definition is equivalent to Def. \ref{def:PLD}. It is simpler to state, but it is less practical for computations. Our approach in Sec.~\ref{sec:algorithm} uses explicit local equations for ${\cal Y}_G$ on $O_Q \times {\cal E}$. A similar algorithm for Def. \ref{def:PLDtoric} would make use of an explicit \emph{global} set of equations for ${\cal Y}_G$. Such equations can be obtained by using \emph{Cox coordinates} on $X_P$ \cite[Sec.~6]{telen2022introduction}. This way, we do not have to run over all faces of $P$. However, these global equations are more complicated: there is one Cox coordinate for each facet of $P$. 

\begin{example}
    Next to simplifying the definition of the principal Landau determinant and avoiding a for-loop over all faces of $P$, an advantage of using global coordinates is that the ideal in the Cox ring sees embedded components such as the one in Ex. \ref{ex:embcomp}, and in \eqref{eq:panzerberghoff}. We illustrate this for the problem of Ex. \ref{ex:embcomp}. The toric variety $X_P$ is a weighted projective plane with Cox ring $\mathbb{C}[x_0,x_1,x_2]$, where $\deg(x_0) =1, \deg(x_1) = 2, \deg(x_2) = 3$. The incidence variety ${\cal Y}_G$ is defined by the homogeneous ideal 
    \[ \langle (x_2-x_0^3)^2 - x_1^3 + zx_0^2x_1^2, \, 2x_2(x_2-x_0^3), \, 3x_1^3 + 2zx_1^2x_0^2 \rangle. \] 
    This ideal has a primary component $\langle z, x_1^3, x_2-x_0^3 \rangle$, which corresponds to the embedded component $\langle z, \alpha_2-1, \alpha_1^3\rangle$ seen in \eqref{eq:PDcubic}. This suggests that the primary decomposition of the ideal defining ${\cal Y}_G$ in the Cox ring of $X_P$ may bring us closer to the Euler discriminant. In particular, one can ask if all irreducible components of the Euler discriminant correspond to one of its primary components. 
\end{example}
\noindent
\textbf{Declaration of competing interest}\\
The authors declare that they have no known competing financial interests or personal relationships that could have appeared to
influence the work reported in this paper.

\addcontentsline{toc}{section}{References}
\bibliographystyle{JHEP}
\bibliography{references}

\end{document}